\documentclass[a4paper]{article}

\usepackage{latexsym}
\usepackage{amsmath} 
\usepackage{amssymb}
\usepackage{tikz}
\usetikzlibrary{arrows.meta}
\usepackage{stmaryrd} 
\usepackage[colorlinks=false]{hyperref}

\renewcommand{\phi}{\varphi}
\renewcommand{\epsilon}{\varepsilon}
\newcommand{\nothing}{\varnothing}

\newcommand{\se}{\mathsf{e}}
\newcommand{\sv}{\mathsf{v}}
\newcommand{\cA}{{\cal A}}
\newcommand{\cB}{{\cal B}}
\newcommand{\cC}{{\cal C}}

\newcommand{\cF}{{\cal F}}
\newcommand{\cL}{{\cal L}}

\newcommand{\cN}{{\cal N}}

\newcommand{\cR}{{\cal R}}
\newcommand{\cS}{{\cal S}}
\newcommand{\cT}{{\cal T}}

\newcommand{\nat}{{\mathbb N}}

\newcommand{\tup}[1]{\langle #1 \rangle}

\newcommand{\abb}[1]{{\sc \lowercase{#1}}}
\newcommand{\rank}{\operatorname{rank}}
\newcommand{\m}{{\mathit mx}}
\newcommand{\rt}{\mathrm{root}}
\newcommand{\lab}{\mathrm{lab}}
\newcommand{\fin}{\mathrm{back}^1}
\newcommand{\back}{\mathrm{back}^*}

\newcommand{\dom}{\mathrm{dom}}
\newcommand{\ran}{\mathrm{ran}}

\newcommand{\suc}{\mathrm{succ}}

\newcommand{\tmark}{\operatorname{mark}}
\newcommand{\xp}[1]{\operatorname{#1}}
\newcommand{\con}{\xp{Con}}
\newcommand{\inv}{\xp{inv}}
\newcommand{\mon}{\xp{mon}}

\newcommand{\stay}{{\rm stay}}
\newcommand{\up}{{\rm up}} 
\newcommand{\down}{{\rm down}}
\newcommand{\enc}{{\rm enc}}
\newcommand{\dec}{{\rm dec}}
\newcommand{\at}{@}
\newcommand{\fl}{{\rm flat}}

\newcommand{\ttt}{\abb{tt}}
\newcommand{\tts}{\abb{tt}$^{\hspace{0.04cm}\mathrm{s}}$}
\newcommand{\ttl}{\abb{tt}$^\ell$}
\newcommand{\ttd}{\abb{tt}$_\downarrow$}
\newcommand{\ttds}{\abb{tt}$^{\hspace{0.04cm}\mathrm{s}}_\downarrow$}
\newcommand{\ttdl}{\abb{tt}$^\ell_\downarrow$}
\newcommand{\ttsu}{\abb{tt}$_\mathrm{su}$}

\newcommand{\ttr}{\abb{tt}$_\mathrm{rel}$}
\newcommand{\ttp}{\abb{tt}$_\mathrm{pru}$}
\newcommand{\ttps}{\abb{tt}$^{\hspace{0.04cm}\mathrm{s}}_\mathrm{pru}$}
\newcommand{\ttpl}{\abb{tt}$^\ell_\mathrm{pru}$}

\newcommand{\mt}{\abb{mt}}

\newcommand{\dttt}{{\small d}\abb{tt}}
\newcommand{\dttd}{{\small d}\abb{tt}$_\downarrow$}
\newcommand{\dttl}{{\small d}\abb{tt}$^\ell$}
\newcommand{\dmt}{{\small d}\abb{mt}}

\newcommand{\mht}{{\small mh}\abb{tt}}
\newcommand{\amfa}{\abb{amfa}}

\newcommand{\family}[1]{\mbox{\sf #1}}
\newcommand{\REGT}{\family{REGT}}
\newcommand{\REGF}{\family{REGF}}

\newcommand{\DMSOT}{\family{DMSOT}}
\newcommand{\MSOT}{\family{dMSOT}}

\newcommand{\DMT}{\family{DMT$_\text{\abb{OI}}$}}
\newcommand{\MT}{\family{dMT}}
\newcommand{\DTMT}{\family{D$_\mathrm{t}$MT}}
\newcommand{\tMT}{\family{d$_\mathrm{t}$MT}}

\newcommand{\DBQREL}{\family{DBQREL}}
\newcommand{\LSI}{\family{LSIF}}
\newcommand{\LSIR}{\family{LSIR}}

\newcommand{\TT}{\family{dTT}}
\newcommand{\TTS}{\family{dTT$^{\hspace{0.04cm}\mathrm{s}}$}}
\newcommand{\TTL}{\family{dTT$^\ell$}}
\newcommand{\TTD}{\family{dTT$_\downarrow$}}

\newcommand{\TTDS}{\family{dTT$^{\hspace{0.04cm}\mathrm{s}}_\downarrow$}}
\newcommand{\TTDL}{\family{dTT$^\ell_\downarrow$}}
\newcommand{\TTSU}{\family{dTT$_\mathrm{su}$}}
\newcommand{\TTLSU}{\family{dTT$^\ell_\mathrm{su}$}}
\newcommand{\TTR}{\family{dTT$_\mathrm{rel}$}}
\newcommand{\TTRS}{\family{dTT$^{\hspace{0.04cm}\mathrm{s}}_\mathrm{rel}$}}
\newcommand{\TTRL}{\family{dTT$^\ell_\mathrm{rel}$}}
\newcommand{\TTP}{\family{dTT$_\mathrm{pru}$}}

\newcommand{\TTPL}{\family{dTT$^\ell_\mathrm{pru}$}}

\newcommand{\NTT}{\family{TT}}
\newcommand{\NTTS}{\family{TT$^{\hspace{0.04cm}\mathrm{s}}$}}
\newcommand{\NTTL}{\family{TT$^\ell$}}
\newcommand{\NTTD}{\family{TT$_\downarrow$}}

\newcommand{\NTTDS}{\family{TT$^{\hspace{0.04cm}\mathrm{s}}_\downarrow$}}
\newcommand{\NTTDL}{\family{TT$^\ell_\downarrow$}}
\newcommand{\NTTP}{\family{TT$_\mathrm{pru}$}}
\newcommand{\NTTPS}{\family{TT$^{\hspace{0.04cm}\mathrm{s}}_\mathrm{pru}$}}
\newcommand{\NTTPL}{\family{TT$^\ell_\mathrm{pru}$}}
\newcommand{\NTTR}{\family{TT$_\mathrm{rel}$}}
\newcommand{\NTTRL}{\family{TT$^\ell_\mathrm{rel}$}}
\newcommand{\NTTPTIME}{\family{TT$^\text{P}$}}

\newcommand{\FTT}{\family{f\,TT}}
\newcommand{\FTTL}{\family{f\,TT$^\ell$}}
\newcommand{\FTTD}{\family{f\,TT$_\downarrow$}}
\newcommand{\FTTDL}{\family{f\,TT$^\ell_\downarrow$}}
\newcommand{\FTTDS}{\family{f\,TT$^{\hspace{0.04cm}\mathrm{s}}_\downarrow$}}

\newcommand{\NMT}{\family{MT}}
\newcommand{\NMTOI}{\family{MT$_\text{\abb{OI}}$}}
\newcommand{\NMTIO}{\family{MT$_\text{\abb{IO}}$}}
\newcommand{\NMTIOPTIME}{\family{MT$_\text{\abb{IO}}^\text{P}$}}
\newcommand{\MTIO}{\family{dMT$_\text{\abb{IO}}$}}
\newcommand{\mrNMT}{\family{mrMT$_\text{\abb{IO}}$}}
\newcommand{\mrNMTPTIME}{\family{mrMT$_\text{\abb{IO}}^\text{P}$}}

\newcommand{\MON}{\family{MON}}
\newcommand{\YIELD}{\family{YIELD}}
\newcommand{\SET}{\family{SET}}

\newcommand{\Enc}{\family{enc}}
\newcommand{\Dec}{\family{dec}}
\newcommand{\Flat}{\family{flat}}

\newcommand{\yield}{\family{yield}}

\newcommand{\FT}{\family{dFT}}
\newcommand{\NFT}{\family{FT}}
\newcommand{\MFT}{\family{dMFT}}
\newcommand{\NMFT}{\family{MFT}}

\newcommand{\dTT}{\family{dTT}}
\newcommand{\dTTmso}{\family{dTT$^{\text{\abb{MSO}}}$}}

\newcommand{\PTIME}{\family{PTIME}}
\newcommand{\NP}{\family{NPTIME}}
\newcommand{\LOGCF}{\family{LOGCFL}}

\newtheorem{theorem}{Theorem}
\newtheorem{lemma}[theorem]{Lemma}
\newtheorem{corollary}[theorem]{Corollary}
\newtheorem{proposition}[theorem]{Proposition}
\newtheorem{remark}[theorem]{Remark}

\newcommand{\qed}{\ifvmode\penalty10000\vskip -\lastskip\fi%
  \penalty10000
  \ \ \hfill\hbox{\rm$\Box$}}
\newenvironment{proof}{\pagebreak[1]\noindent{\bf
Proof. }\nopagebreak}{\medskip\pagebreak[3]}

\makeatletter
\newtheorem{ex@ple}[theorem]{Example}
\newenvironment{example}{\begin{ex@ple}\normalfont}{\end{ex@ple}}
\makeatother

\newcommand{\smallpar}[1]{\medskip\noindent {\bf #1.}}

\begin{document}

\title{Linear-Bounded Composition of \\ 
Tree-Walking Tree Transducers: \\
Linear Size Increase and Complexity\footnote{Published at https://link.springer.com/article/10.1007/s00236-019-00360-8
}}


\author{
Joost Engelfriet\thanks{LIACS, Leiden University, P.O. Box 9512, 2300 RA Leiden, the Netherlands; 
email: {\tt j.engelfriet@liacs.leidenuniv.nl}} 
\and Kazuhiro Inaba\thanks{Google Japan G.K., Tokyo, Japan;
email: {\tt kinaba@google.com}}
\and Sebastian Maneth\thanks{Department of Mathematics and Informatics, Universit\"at Bremen, 
P.O. Box 330 440, 28334 Bremen, Germany;
email: {\tt maneth@uni-bremen.de}}
}

\date{}

\maketitle


\begin{abstract}
Compositions of tree-walking tree transducers form a hierarchy 
with respect to the number of transducers in the composition.
As main technical result it is proved that any such composition can be realized 
as a linear-bounded composition, which means that 
the sizes of the intermediate results
can be chosen to be at most linear in the size of the output tree. 
This has consequences for the expressiveness and complexity 
of the translations in the hierarchy.
First, if the computed translation is a function of linear size increase, i.e., 
the size of the output tree is at most linear in the size of the input tree,
then it can be realized by just one, deterministic, 
tree-walking tree transducer. 
For compositions of deterministic transducers it is decidable 
whether or not the translation is of linear size increase.
Second, every composition of deterministic transducers can be computed in 
deterministic linear time on a RAM and in deterministic linear space on a Turing machine,
measured in the sum of the sizes of the input and output tree. 
Similarly, every composition of nondeterministic transducers can be computed in
simultaneous polynomial time and linear space on a nondeterministic Turing machine.
Their output tree languages are deterministic context-sensitive, i.e., can be recognized 
in deterministic linear space on a Turing machine. 
The membership problem for compositions of nondeterministic translations
is nondeterministic polynomial time and deterministic linear space. 
All the above results also hold for compositions of macro tree transducers. 
The membership problem for the composition of a nondeterministic and a deterministic 
tree-walking tree translation (for a nondeterministic IO macro tree translation)
is log-space reducible to a context-free language, whereas the membership problem 
for the composition of a deterministic and a nondeterministic 
tree-walking tree translation (for a nondeterministic OI macro tree translation)
is possibly NP-complete.

\end{abstract}

\newpage
\tableofcontents

\newpage
\section{Introduction}

Tree transducers are used, e.g., in compiler theory or, more generally, 
the theory of syntax-directed semantics of context-free languages~\cite{FulVog}, and 
in the theory of XML queries and XML document transformation~\cite{Schw,Hos}. 
One of the most basic types of tree transducer is the top-down tree transducer 
(in short \ttd).
It is a finite-state device that walks top-down on the input tree, from parent to child,
possibly branching into parallel copies of itself at each step 
(thus allowing the transducer to visit all children of the parent).
During this process, the output tree is generated top-down.
The \ttd\ has been generalized in two different ways.
By allowing it to walk also bottom-up, from child to parent, 
still possibly branching at every step and still generating the output tree top-down, one obtains 
the tree-walking tree transducer (in short,~\ttt).\footnote{The name 
``tree-walking tree transducer'' was introduced in~\cite{Eng09}. The adjective ``tree-walking''
stands for the fact that the transducer walks on the input tree (just as the tree-walking automaton 
of~\cite{AhoUll}). The \ttt\ is the generalization 
to trees of the two-way finite-state string transducer, which walks on its input string in both directions
and produces the output string one-way from left to right. Note that ``tree-walking'' and ``two-way'' 
alliterate.
}
On the other hand, restricting its walk to be top-down 
but allowing its states to have parameters of type output tree, 
one obtains the macro tree transducer (in short, \mt). 
In general we consider nondeterministic transducers, 
with deterministic transducers as an important special case
(abbreviated as \dttd, \dttt, and \dmt).

To turn the \ttd\ into a more flexible model of tree transformation, it was enhanced 
with the feature of regular look-ahead, which means that it can test 
whether or not the subtree at the current node of the input tree
belongs to a given regular tree language.
The \mt\ already has the ability to implement regular look-ahead.
Since both the enhanced \ttd\ and the \mt\ process the input tree top-down, 
they can also implement ``regular look-around'', which means that they
can test arbitrary regular properties of the current node of the input tree.
More precisely, they can test whether the input tree, in which the current node is marked,
belongs to a given regular tree language. 
Such regular look-around tests are also called \abb{mso} tests, because they can be expressed 
by formulas of monadic second-order logic with one free node variable. 
The \ttt, as defined in~\cite{MilSucVia03}, does not have regular look-ahead or look-around.
One of the drawbacks of this is that the \ttt\ cannot recognize 
all regular tree languages without branching~\cite{BojCol}.  
Hence, from now on, we assume that the \ttt\ (and the \ttd) is enhanced with regular look-around, 
i.e., with regular tests of the current input node. 
The resulting \ttt~formalism is a quite robust, flexible, and intuitive model of tree transformation. 

The \ttt\ and \mt, generalizations of the \ttd, are closely related,
in particular in the deterministic case.
In fact, every \dttt\ can be simulated by a \dmt,
whereas every \dmt\ can be simulated by a composition of two \dttt's. 
Thus, every composition of \dttt's can be realized by a composition of \dmt's, and vice versa. 
Compositions of \dttt's form a proper hierarchy, in an obvious way.
A single \dttt\ is at most of exponential size increase, which means that the size of 
the output tree is at most exponential in the size of the input tree. However, a composition of 
two \dttt's can be of double exponential size increase. In general, compositions 
of $k$ \dttt's are at most, and can be, of \mbox{$k$-fold} exponential size increase.
Compositions of \dmt's form a proper hierarchy by a similar argument. 
For nondeterministic \ttt's and \mt's the situation is similar but more complicated.
Every \mt\ can be simulated by a composition of two \ttt's. However, 
as opposed to \ttt's, \mt's are always finitary, which means that for every given input tree 
an \mt\ computes finitely many output trees. 

In this paper we investigate compositions of \ttt's (and hence of \mt's) 
with respect to their expressivity and their complexity. Our main technical result is
that every composition of \ttt's can be realized by a \emph{linear-bounded composition} of \ttt's,
which means that, when computing an output tree from an input tree, 
the intermediate results can be chosen in such a way that 
their sizes are at most linear in the size of the output tree. 
More precisely, a composition of two transducers (for simplicity) is linear-bounded 
if there is a constant $c$ such that 
for every pair~$(t,s)$ of an input tree $t$ and output tree $s$ in the composed translation
there is an intermediate tree $r$ (meaning that $(t,r)$ and $(r,s)$ are in the first and second 
translation, respectively) such that the size of $r$ is at most $c$ times the size of $s$. 
Intuitively, to compute $s$ from $t$ there is no need to consider intermediate results 
that are much larger than $s$.
If both transducers are deterministic it means that for every input tree $t$ in the domain of 
the composed translation the size of the unique intermediate tree $r$ is at most linear 
in the size of the unique output tree~$s$. 

To prove that every composition of two \ttt's can be realized by 
a linear-bounded composition of two \ttt's, 
we first show that every \ttt\ can be decomposed into a \ttd\ that ``prunes'' the input tree,  
followed by a \ttt\ that is ``productive'' on at least one of the intermediate trees 
generated by the \ttd, which means that it uses each leaf and each monadic node 
of that intermediate tree in order to generate the output tree. Productivity guarantees 
that the composition of these two transducers is linear-bounded. 
We also prove that the composition of an arbitrary \ttt\ with a ``pruning'' top-down \ttt\ 
can be realized by one \ttt. Thus, when two \ttt's are composed, the second \ttt\ can 
split off the pruning \ttd\ (to the left), which can be absorbed (to the right) by the first \ttt.
The composition of the resulting two \ttt's is then linear-bounded. 
This also holds for deterministic transducers, in which case the pruning \ttt\ is also deterministic. 
Similar results were presented for macro tree transducers 
in~\cite[Section~3]{Man02} and~\cite[Section~4]{InaMan08}.

Thus, roughly speaking, our main technical result provides a method to implement compositions of \ttt's
in such a way that the generation of superfluous nodes, i.e., nodes on which a \ttt\ 
just walks around without producing any output, is avoided by pruning those superfluous parts 
from the intermediate trees. As such it can be viewed as a static garbage collection procedure,
and leads, in principle, to algorithms for automatic compiler and XML query optimization.
Since \ttt's are essentially finite-state automata walking on trees, it is not really surprising that 
only a linearly bounded amount of intermediate information is useful to the final output. 
However, proving this rigorously requires quite some effort. In particular, the subcomputations of
the \ttt\ during which it does not produce output will be determined by regular look-around. 

The above method can be used to obtain results on both the expressivity and the complexity 
of compositions of \ttt's, as discussed in the next paragraphs.

\smallpar{Expressivity}
We have seen above that compositions of \ttt's can be of \mbox{$k$-fold} exponential size increase. 
However, many real world tree transformations are of \emph{linear size increase}.
We prove that the hierarchy of compositions of deterministic \ttt's  
collapses when restricted to translations of linear size increase: 
every composition of \dttt's 
that is of linear size increase can be realized by just one~\dttt.
We also show that it is decidable whether or not a composition of~\dttt's 
is of linear size increase.
This means that a compiler or XML query, no matter how inefficiently programmed in several phases,
can be realized in one efficient phase, provided it is of linear size increase. 
In fact, as we will see below, that single phase can be executed in linear time. 
More theoretically, we additionally prove that a function that can be realized 
by a composition of nondeterministic~\ttt's,
can also be realized by a composition of deterministic \ttt's, and hence
by one deterministic \ttt\ if that function is of linear size increase. 
Thus, the only (functional) tree transformations that can be realized by a composition 
of \ttt's but not by a single \ttt, are tree transformations of superlinear size increase. 

The proof of the collapse of the hierarchy of compositions of \dttt's  
is based on the known fact that every \dttt\ of linear size increase can be realized 
by a \dttt\ that is ``single-use'', which means that it never visits a node 
of the input tree twice in the same state. 
In fact, it is proved in~\cite{EngMan99,EngManLSI} that even \dmt's of linear size increase
can be realized by single-use \dttt's. Vice versa, it is obvious that 
every single-use \dttt\ is of linear size increase. 
In~\cite{BloEng} it is shown that single-use \dttt's have the same power as 
deterministic \abb{mso} tree transducers, which use formulas of monadic second-order logic to 
define the output tree in terms of the input tree (see~\cite{Cou94,thebook}). 

By our main technical result, we may always assume that a composition of two \dttt's
is linear-bounded. If the composition is of linear size increase, 
then the first \dttt\ is obviously also of linear size increase,
and can therefore be realized by a single-use \dttt. 
We also prove that the composition of a single-use~\dttt\ with an arbitrary \dttt\
can be realized by one \dttt. Thus, altogether, if the composition of two \dttt's
is of linear size increase, then it can be realized by a single-use \dttt.
This argument can easily be turned into an inductive proof for a composition of 
any number of \dttt's.

\smallpar{Complexity}
We first consider deterministic \ttt's. 
The translation realized by a deterministic \ttt\ can be computed on a RAM in linear time, 
in the sum of the sizes of the input and output tree. With respect to space, we prove that it can 
be computed on a deterministic Turing machine in linear space (again, in the sum of the sizes 
of the input and output tree). Since we may assume by our main technical result that 
the sizes of the intermediate results are at most linear in the size of the output tree, 
it should be clear that these facts also hold for compositions of \dttt's. 
We also consider output tree languages,
i.e., the images of a regular tree language under a composition of \dttt's.
Since the regular tree languages are closed under prunings, our technical decomposition result 
now implies that these output languages are in $\family{DSPACE}(n)$, 
i.e., can be recognized by a Turing machine in deterministic linear space (or, in other words,
are deterministic context-sensitive). Since the yield of a tree can be computed by a \dttt\ 
(representing it by a monadic tree), the output string languages, which 
are the yields of the output tree languages, are also in $\family{DSPACE}(n)$.
The languages in the well-known \abb{io}-hierarchy are examples of such output languages.
For compositions of top-down tree transducers (even nondeterministic ones) this result on 
output languages was proved in~\cite{Bak78}, using a technical result very similar to ours. 

Our results on nondeterministic \ttt's (and their proofs) are very similar to those for \dttt's. 
The translation realized by a composition of \ttt's can be computed by 
a nondeterministic Turing machine in simultaneous polynomial time and linear space
(in the sum of the sizes of the input and output tree). The corresponding output languages 
can be recognized by such a Turing machine and hence are in $\NP$. 
Using the results on the membership problem for compositions of \ttt's discussed in the next paragraph,
we generalize the result of~\cite{Bak78} and prove that these output languages are even
in $\family{DSPACE}(n)$, which means that they are deterministic context-sensitive.
The languages in the well-known \abb{oi}-hierarchy are examples of such output languages.

Finally, we consider the membership problem for compositions of \ttt's, which asks whether or not
a given pair $(t,s)$ of input tree $t$ and output tree $s$ belongs to the composed translation.
It follows easily from the above complexity results that for (non)deterministic \ttt's 
the problem is decidable in (non)deterministic polynomial time and in (non)deterministic linear space.
For the special case of the composition of a nondeterministic \ttt\ with a deterministic \ttt\ 
we prove that the problem is even in $\LOGCF$, i.e., log-space reducible to a context-free language,
and hence in $\PTIME$ and $\family{DSPACE}(\log^2 n)$.
From this we conclude that for nondeterministic \ttt's 
the problem is even decidable in deterministic linear space.
However, for the special case of the composition of a deterministic \ttt\ with a nondeterministic one, 
the problem can be $\family{NP}$-complete. 
From the two special cases we obtain that the membership problem 
for a (single) nondeterministic macro tree transducer
is in $\LOGCF$ for \abb{io} macro tree transducers
(strengthening the result in~\cite{InaMan09} where it was shown to be in $\PTIME$),
whereas it can be $\family{NP}$-complete for \abb{oi} macro tree transducers.

\smallpar{Structure of the paper}
The reader is assumed to be familiar with the basics of formal language theory, 
in particular tree language theory, and complexity theory. 
The only formalisms used are tree-walking tree transducers
(\ttt's, of course), top-down tree transducers (\ttd's, as a special case of \ttt's), 
context-free grammars, regular tree grammars, and finite-state tree automata. 
Results on macro tree transducers are taken from the literature. 

The main results are proved in Sections~\ref{sec:prod} to~\ref{sec:transcomp}.
Section~\ref{sec:trees} contains a number of preliminary notions, 
in particular linear-bounded composition, linear size increase, and regular look-around. 
In Section~\ref{sec:trans} we define the tree-walking tree transducer
(with regular look-around), together with some of its special cases
such as top-down and single-use. A \ttt\ that does not use regular look-around tests
is called ``local''. A ``pruning'' \ttt\ is a \ttd\ that, 
roughly speaking, removes or relabels each node of the input tree and 
possibly deletes several of its children (together with their descendants).
After giving two examples we present the
composition hierarchy of \dttt's and end the section 
with some elementary syntactic properties of \ttt's. 
In Section~\ref{sec:basic} it is shown how to separate the regular look-around from a \ttt\ 
and incorporate it into another \ttt.
For instance, every~\ttt\ can be decomposed into 
a deterministic pruning \ttd\ that just relabels the nodes of the input tree
(and hence does not really ``prune''), followed by a local \ttt. 
We also state the fact that the domain of a \ttt\ is a regular tree language.  
Consequently, it is possible to define the regular tests of a \ttt\ 
as domains of other \ttt's, which is a convenient technical tool.
Section~\ref{sec:comp} contains three composition results.
We prove that the composition of 
a \ttt\ with a pruning~\ttd\ can be realized by a \ttt\  
(such that determinism is preserved). 
Together with the above-mentioned decomposition, 
this implies for instance that in a composition of two \ttt's,
the second \ttt\ can always be assumed to be local: 
the second \ttt\ splits off a pruning~\ttd\ that is absorbed by the first \ttt.
In the deterministic case, we even prove that the composition of 
a \dttt\ with an arbitrary \dttd\ can be realized by a \dttt, and 
we also prove that the composition of 
a single-use~\dttt\ with a \dttt\ can be realized by a \dttt. 
Section~\ref{sec:dmtmso} presents the known fact that 
every~\dttt\ of linear size increase can be realized by a single-use~\dttt,
and discusses the relationship between \ttt's, macro tree transducers, 
and \abb{mso}~tree transducers. 
In Section~\ref{sec:func} we show that a (partial) function 
that can be realized by a composition of nondeterministic \ttt's, 
can also be realized by a composition of deterministic \ttt's. 
To prove this we first prove a lemma: for every \ttd\ there is a deterministic \ttd\
that realizes a ``uniformizer'' of the translation realized by the given \ttd,
i.e., a function that is a subset of that translation, with the same domain. 
Section~\ref{sec:prod} contains our main technical result: every \ttt\ 
can be decomposed into a pruning \ttt\ and another \ttt\ such that the composition is linear-bounded.
It implies (by splitting and absorbing) that 
a composition of \ttt's can always be assumed to be linear-bounded. 
The ``uniformizer'' lemma of the previous section is applied to the pruning~\ttd, proving the same result 
for deterministic \ttt's.
Section~\ref{sec:lsi} presents the main results on linear size increase,
and Sections~\ref{sec:complex} and~\ref{sec:nondetcomplex} present the main results 
on the complexity of compositions of 
deterministic and nondeterministic \ttt's, respectively. 
In Section~\ref{sec:transcomp} we prove the main results on the complexity of the membership problem 
for the composition of two \ttt's. 
Finally, in Section~\ref{sec:forest} we show (in a straightforward way) 
that all main results also hold for transducers
that transform unranked trees, or forests, which are a natural model of XML documents.

The reader who is interested only in complexity can disregard all results on single-use \ttt's,
and skip Sections~\ref{sub:mso} and~\ref{sec:lsi}. 
The reader who is interested only in expressivity can just skip
Sections~\ref{sec:complex}, \ref{sec:nondetcomplex}, and~\ref{sec:transcomp}.

\smallpar{Remarks on the literature}
Top-down tree transducers were introduced in~\cite{Rou,Tha}; 
regular look-ahead was added in~\cite{Eng77}.
Macro tree transducers were introduced in~\cite{CouFra,EngVog85}. 
Tree-walking tree transducers were introduced in~\cite{MilSucVia03} 
(where they are called 0-pebble tree transducers), and studied in, e.g., \cite{EngMan03,Eng09,ManFriSei}. 
They were already mentioned in~\cite[Section~3(7)]{Eng86} (where they are called RT(Tree-walk) transducers). 
Regular look-around was added to \ttt's in~\cite[Section~8.2]{thebook} 
(where they are called \abb{MS} tree-walking transducers);
for tree-walking automata that was already done in~\cite{bloem}.
However, formal models similar to the \ttt\ were introduced and studied before. 
The tree-walking automaton of~\cite{AhoUll} translates trees into strings.  
As explained in~\cite[Section~3(7)]{Eng86} and~\cite[Section~3.2]{EngMan03}, 
the \ttt\ is closely related to the attribute grammar~\cite{KnuthAG}, 
which is a well-known model of syntax-directed semantics (and a compiler construction tool). 
An attribute grammar translates derivation trees of an underlying context-free grammar into arbitrary values.
Tree-valued attribute grammars were considered, e.g., in~\cite{EngFil}. 
The attributed tree transducer, introduced in~\cite{Ful}, 
is an operational version of the tree-valued attribute grammar, without underlying context-free grammar.
Regular look-around was added to the attributed tree transducer in~\cite{BloEng}
(where it is called look-ahead). 
Attributed tree transducers are a special type of \ttt's, of which the states are viewed as 
attributes of the nodes of the input tree.
By definition a deterministic attributed tree transducer (like an attribute grammar) 
has to be noncircular, which means that it should generate an output tree 
whenever it is started in any state on any node of an input tree. Thus, it is total in a strong sense. 
This is natural from the point of view of syntax-directed semantics,
but quite restrictive and inconvenient from the operational point of view of tree transformation. 
Several of the auxiliary results in Sections~\ref{sec:trans} to~\ref{sec:comp} 
are closely related to (and generalizations of) 
well-known results on attributed tree transducers (see, e.g., \cite{FulVog}). 
As an example, it is proved in~\cite[Theorem~4.3]{Ful} that, for deterministic transducers, 
the composition of an attributed tree transducer with a top-down tree transducer can be realized 
by an attributed tree transducer. That does not immediately imply that the same is true 
for a \dttt\ and a \dttd, which we show in Section~\ref{sec:comp}, 
because \dttt's are not necessarily total and they have regular look-around. 
Moreover, we wanted such results also to be understandable for readers unfamiliar 
with attribute grammars and attributed tree transducers. 

The main results of this paper were first presented at FSTTCS '02 \cite{Man02} 
(on the complexity of compositions of deterministic \mt's),
at FSTTCS '03 \cite{Man03} (on compositions of \mt's that realize functions of linear size increase), 
at FSTTCS '08 \cite{InaMan08} (on the complexity of compositions of nondeterministic \mt's),
at PLAN-X '09 \cite{InaMan09} (on the complexity of the membership problem for \mt's),
and in the Ph.D. Thesis of the second author \cite{Ina} (on the last two subjects).

\section{Preliminaries}\label{sec:trees}

\emph{Convention}: All results stated and/or proved in this paper are effective. 

\smallpar{Sets, strings, and relations}
The set of natural numbers is $\nat=\{0,1,2,\dots\}$.
For $m,n\in \nat$, we denote the interval $\{k\in\nat\mid m\leq k\leq n\}$ by $[m,n]$. 
The cardinality or size of a set $A$ is denoted by $\#(A)$. 
The set of strings over $A$ is denoted by~$A^*$. 
It consists of all sequences $w=a_1\cdots a_m$ with $m\in\nat$ and 
$a_i\in A$ for every $i\in[1,m]$. 
The length~$m$ of $w$ is denoted by $|w|$. 
The empty string (of length $0$) is denoted by~$\epsilon$. 
The concatenation of two strings $v$ and $w$ is denoted by $v\cdot w$ or just~$vw$.
Moreover, $w^0=\epsilon$ and $w^{k+1}=w\cdot w^k$ for $k\in\nat$.

The domain and range of a binary relation $R\subseteq A\times B$ are denoted
by $\dom(R)$ and $\ran(R)$, respectively. 
For $A'\subseteq A$, $R(A')=\{b\in B\mid (a,b)\in R \text{ for some } a\in A'\}$. 
The composition of $R$ with a binary relation $S\subseteq B\times C$ is 
$R\circ S = \{(a,c)\mid \exists\, b\in B: (a,b)\in R, \, (b,c)\in S\}$.
The~inverse of $R$ is $R^{-1}=\{(b,a)\mid (a,b)\in R\}$. 
Note that $\dom(R\circ S)= R^{-1}(\dom(S))$ and $\ran(R\circ S)=S(\ran(R))$.
If $A=B$ then the transitive-reflexive closure of~$R$ is $R^*=\bigcup_{k\in\nat}R^k$
where $R^0=\{(a,a)\mid a\in A\}$ and $R^{k+1}=R\circ R^k$. The composition of 
two classes of binary relations $\cR$ and $\cS$ is 
$\cR\circ \cS = \{R\circ S\mid R\in\cR,\,S\in\cS\}$. 
Moreover, $\cR^1=\cR$ and $\cR^{k+1} = \cR\circ\cR^k$ for $k\geq 1$. 
The relation $R$ is \emph{finitary} if $R(a)$ is finite for every $a\in A$, 
where $R(a)$ denotes $R(\{a\})$. 
It is a (partial) function from $A$ to $B$ 
if $R(a)$ is empty or a singleton for every $a\in A$, 
and it is a total function if, moreover, $\dom(R)=A$. 

\smallpar{Trees}
An alphabet is a finite set of symbols. 
A \emph{ranked} alphabet $\Sigma$ is an alphabet 
together with a  mapping $\rank_\Sigma : \Sigma \to \nat$
(of which the subscript $\Sigma$ will be dropped when it is clear from the context).
The maximal rank of elements of $\Sigma$ is denoted $\m_\Sigma$.
For every $m\in\nat$ we denote by $\Sigma^{(m)}$ the elements of $\Sigma$ that have rank $m$.

Trees over $\Sigma$ are recursively defined
to be strings over $\Sigma$, as follows.  
For every $m\in\nat$, if $\sigma\in\Sigma^{(m)}$ and 
$t_1, \dots, t_m$ are trees over $\Sigma$, then
$\sigma \,t_1 \cdots t_m$ is a tree over~$\Sigma$.
For readability we also write 
the tree $\sigma \,t_1 \cdots t_m$ as the term $\sigma(t_1, \dots, t_m)$. 
The set of all trees over~$\Sigma$ is denoted $T_\Sigma$; thus $T_\Sigma\subseteq \Sigma^*$.
For an arbitrary finite set $A$, disjoint with $\Sigma$, 
we denote by $T_\Sigma(A)$ the set $T_{\Sigma\cup A}$,
where each element of $A$ has rank~0. 

As usual trees are viewed as directed labeled graphs. 
The nodes of a tree~$t$ are indicated by Dewey notation, 
i.e., by elements of $\nat^*$, which are strings of natural numbers. 
The root of $t$ is indicated by the empty string $\epsilon$, 
but will also be denoted by $\rt_t$ for readability.
The $i$-th child of a node $u$ of $t$ is indicated by $ui$, and  
there is a directed edge from the parent $u$ to the child $ui$.
Formally, the set $\cN(t)$ of nodes of a tree $t=\sigma \,t_1 \cdots t_m$ over $\Sigma$
can be defined recursively by $\cN(t)=\{\epsilon\}\cup \{iu\mid i\in [1,m],\,u\in \cN(t_i)\}$. 
Thus, $\cN(t)\subseteq [1,\m_\Sigma]^*$. 
The root of $t=\sigma t_1 \cdots t_m$ has label $\sigma$, and the node $iu$ of $t$
has the same label as the node $u$ of $t_i$. 
The rank of node $u$ is the rank of its label, i.e., the number of its children. 
A leaf is a node of rank~0, and a monadic node is a node of rank~1. 
Every node of $t$ has a child number: each node $ui$ has child number $i$, and 
the root $\epsilon$ is given child number $0$ for technical convenience.
For a node $u$ of $t$ the subtree of $t$ with root $u$ is denoted $t|_u$;
thus, $t|_\epsilon=t$ and $t|_{iu}=t_i|_u$. 
A node $v$ of $t$ is a descendant of a node $u$ of $t$, and $u$ is an ancestor of $v$,
if there exists $w\in\nat^*$ such that $w\neq\epsilon$ and $v=uw$ (thus, $u$ is \emph{not}
a descendant/ancestor of itself). 
The size of a tree $t$ is $|t|$, i.e., its length as a string. 
Note that $|t|=\#(\cN(t))$ because 
the nodes of $t$ correspond one-to-one to the positions in the string $t$, i.e., 
for every $\sigma\in\Sigma$, each occurrence of $\sigma$ in $t$ 
corresponds to a node of $t$ with label $\sigma$. 
The left-to-right linear order on $\cN(t)$ according to this correspondence 
is called the pre-order of the nodes of $t$. 
The yield of $t$ is the string of labels of its leaves, in pre-order. 
The height of $t$ is the number of edges of a longest directed path from the root of $t$ to a leaf;
thus, it is the maximal length of its nodes (which are strings over $\nat$). 

A \emph{tree language} $L$ is a set of trees over $\Sigma$, for some ranked alphabet $\Sigma$,
i.e., $L\subseteq T_\Sigma$.
A \emph{tree translation} $\tau$ is a binary relation between trees over $\Sigma$ and 
trees over $\Delta$, for some ranked alphabets $\Sigma$ and $\Delta$, i.e., 
$\tau\subseteq T_\Sigma\times T_\Delta$.

\smallpar{Linear-bounded composition}
Let $\Sigma$, $\Delta$, and $\Gamma$ be ranked alphabets.
For tree translations $\tau_1\subseteq T_\Sigma\times T_\Delta$ and 
$\tau_2\subseteq T_\Delta\times T_\Gamma$, we say that the pair $(\tau_1,\tau_2)$ is 
\mbox{\emph{linear-bounded}} if there is a constant $c\in\nat$ such that 
for every $(t,s)\in \tau_1\circ\tau_2$ there exists $r\in T_\Delta$ such that 
$(t,r)\in\tau_1$, $(r,s)\in\tau_2$, and $|r|\leq c\cdot|s|$. 
Thus, the intermediate result~$r$ can be chosen such that its size is linear 
in the size of the output $s$. 
Note that if $\tau_1$ and $\tau_2$ are functions, this means that 
$|r|\leq c\cdot|\tau_2(r)|$ for every $r\in\ran(\tau_1)\cap\dom(\tau_2)$. 

For classes $\cT_1$ and $\cT_2$ of tree translations, we define $\cT_1\ast\cT_2$ to consist of all 
translations $\tau_1\circ\tau_2$ such that $\tau_1\in\cT_1$, $\tau_2\in\cT_2$, and
$(\tau_1,\tau_2)$ is linear-bounded.

\begin{lemma}\label{lem:ast}
Let $\cT_1$, $\cT_2$, and $\cT_3$ be classes of tree translations. Then
\[
\cT_1\circ(\cT_2\ast\cT_3)\subseteq (\cT_1\circ\cT_2)\ast\cT_3
\text{\quad and \quad} (\cT_1\ast\cT_2)\ast\cT_3\subseteq \cT_1\ast(\cT_2\circ\cT_3).
\]
\end{lemma}

\begin{proof}
Let $\tau_i\in\cT_i$ for $i\in\{1,2,3\}$.
If the pair $(\tau_2,\tau_3)$ is linear-bounded then 
so is the pair $(\tau_1\circ\tau_2,\tau_3)$, with the same constant~$c$.
If $(\tau_1,\tau_2)$ and $(\tau_1\circ\tau_2,\tau_3)$ are linear-bounded 
with constant $c_1$ and~$c_2$, respectively,
then $(\tau_1,\tau_2\circ\tau_3)$ is linear-bounded with constant $c_1\cdot c_2$. 
\qed
\end{proof}

A function $\tau: T_\Sigma\to T_\Delta$ is 
\emph{of linear size increase} if there is a constant $c\in\nat$ 
such that $|\tau(t)|\leq c\cdot |t|$ for every $t\in\dom(\tau)$.
The class of functions of linear size increase will be denoted 
by $\LSI$.

\begin{lemma}\label{lem:lsilsd}
Let $\tau_1: T_\Sigma\to T_\Gamma$ and $\tau_2: T_\Gamma\to T_\Delta$ be functions 
such that $\ran(\tau_1)\subseteq\dom(\tau_2)$.
If $\tau_1\circ \tau_2\in\LSI$ and $(\tau_1,\tau_2)$ is 
linear-bounded, then $\tau_1\in\LSI$. 
\end{lemma}

\begin{proof}
It follows from $\ran(\tau_1)\subseteq\dom(\tau_2)$ that $\dom(\tau_1\circ \tau_2)=\dom(\tau_1)$.
Since $(\tau_1,\tau_2)$ is linear-bounded, there is a $c$ such that 
$|\tau_1(t)|\leq c\cdot |\tau_2(\tau_1(t))|$ for every $t\in\dom(\tau_1)$.
Since $\tau_1\circ \tau_2\in\LSI$, there is a $c'$ such that $|\tau_2(\tau_1(t))|\leq c'\cdot|t|$ 
for every $t\in\dom(\tau_1)$. Hence $|\tau_1(t)|\leq c\cdot c'\cdot |t|$
for every $t\in\dom(\tau_1)$, which means that $\tau_1\in\LSI$. 
\qed
\end{proof}

\smallpar{Grammars and automata}
Context-free grammars and, in particular, regular tree grammars
will be used to define the computations of tree-walking tree transducers,
and to define the ``regular look-around'' used by these transducers.  
A context-free grammar is specified as a tuple $G=(N,T,\cS,R)$, where 
$N$ is the nonterminal alphabet, $T$ the terminal alphabet (disjoint with $N$), 
$\cS\subseteq N$ the set of initial nonterminals, and $R$ the finite set of rules, 
where each rule is of the form $X\to \zeta$ with $X\in N$ and $\zeta\in(N\cup T)^*$. 
A sentential form of~$G$ is a string $v\in (N\cup T)^*$ such that $S\Rightarrow_G^*v$ for some $S\in\cS$,
where $\Rightarrow_G$ is the usual derivation relation of $G$: if $X\to \zeta$ is in $R$, 
then $v_1 X v_2 \Rightarrow_G v_1 \zeta v_2$ for all $v_1,v_2\in (N\cup T)^*$. 
The language $L(G)$ generated by $G$ is the set of all terminal sentential forms, i.e., 
$L(G)= \{w\in T^*\mid \exists \,S\in\cS: S\Rightarrow_G^*w\}$. 
To formally define the derivation trees of $G$ as ranked trees, 
we need to subscript its nonterminals with ranks 
because $G$ can have rules $X\to \zeta_1$ and $X\to \zeta_2$ with $|\zeta_1|\neq|\zeta_2|$.  
Let $\overline{N}$ be the ranked alphabet consisting of all symbols $X_m$, of rank~$m$, such that
$G$ has a rule $X\to \zeta$ with $|\zeta|=m$. The terminal symbols in $T$ are given rank~0.
Then the derivation trees of $G$ are generated by the context-free grammar 
$G^\mathrm{der}= (N',\overline{N}\cup T,\cS',R^\mathrm{der})$ such that $N'=\{X'\mid X\in N\}$,
$\cS'=\{S'\mid S\in \cS\}$,
and if $R$ contains a rule $X\to \zeta$, then $R^\mathrm{der}$ contains the rule
$X'\to X_m\zeta'$ where $m=|\zeta|$ and $\zeta'$ is obtained from $\zeta$ 
by changing every nonterminal $Y$ into~$Y'$.
Note that we only consider derivation trees 
that correspond to derivations $S\Rightarrow_G^*w$ with $S\in\cS$ and $w\in T^*$.
Such a derivation tree has yield $w$, because when taking the yield of a derivation tree 
we skip the leaves with label $X_0$.
Moreover, when considering a derivation tree of~$G$, we will disregard 
the subscripts of the nonterminals and we will say that a node 
has label $X$ rather than $X_m$. 
As an example, if $G$ has the rules $S\to aXYb$, $X\to aY$, $Y\to ba$, and $Y\to \epsilon$,
then $G^\mathrm{der}$ has the rules 
$S'\to S_4aX'Y'b$, $X'\to X_2aY'$, $Y'\to Y_2ba$, and $Y'\to Y_0$. 
The string $aabab$ is generated by~$G$,
and the derivation tree $S_4aX_2aY_0Y_2bab=S_4(a,X_2(a,Y_0),Y_2(ba),b)$ 
is generated by $G^\mathrm{der}$; the nodes of this tree are labeled by $S$, $X$, $Y$, $a$, and $b$,
and its yield is $aabab$.  

A context-free grammar is \emph{$\epsilon$-free} if it does not have $\epsilon$-rules, 
i.e., rules $X\to\epsilon$. We will mainly deal with $\epsilon$-free context-free grammars.

A context-free grammar $G$ is \emph{finitary} if $L(G)$ is finite.
We need the following elementary lemma on finitary context-free grammars.

\begin{lemma}\label{lem:finpump}
Let $G=(N,T,\cS,R)$ be a finitary context-free grammar.
For every string $w\in L(G)$ there exists a derivation tree $d\in L(G^\mathrm{der})$
such that the yield of $d$ is~$w$ and the height of $d$ is at most $\#(N)$. 
\end{lemma}

\begin{proof}
Let $d$ be a derivation tree with yield $w$ and 
suppose that a node $u$ of~$d$ and a descendant $v$ of $u$ have the same nonterminal label
(disregarding the ranking subscripts). 
Then the tree $d$ can be pumped in the usual way. But since $L(G)$ is finite, 
the yield of the pumped tree remains the same. Hence we can remove the pumped part from $d$.
Repeating this, we obtain a derivation tree as required. 
\qed
\end{proof}

A context-free grammar $G=(N,T,\cS,R)$ is \emph{forward deterministic} 
if $\cS$ is a singleton and distinct rules have distinct left-hand sides.\footnote{That is 
as opposed to a ``backward deterministic'' context-free grammar in which distinct rules have 
distinct right-hand sides, see, e.g., \cite{Eng09}.
A forward deterministic context-free grammar that generates a string is also called
a ``straight-line'' context-free grammar. 
} 
Such a grammar generates at most one string in $T^*$ and has at most one derivation tree. 
If $L(G^\mathrm{der})=\{d\}$, then the height of $d$ is at most $\#(N)$ by Lemma~\ref{lem:finpump}. 

A \emph{regular tree grammar} is a context-free grammar $G=(N,\Sigma,\cS,R)$ such that 
$\Sigma$ is a ranked alphabet, and 
$\zeta \in T_\Sigma(N)$ for every rule $X\to \zeta$ in~$R$.
A~regular tree grammar generates trees over $\Sigma$, i.e., $L(G)\subseteq T_\Sigma$. 
Note that every regular tree grammar is $\epsilon$-free.
Note also that for every context-free grammar $G$, the grammar~$G^\mathrm{der}$ is a regular tree grammar. 
If, in particular, $G$ is itself a regular tree grammar, as above,
then it should be noted that the elements of $\Sigma$ all have rank~0 in~$G^\mathrm{der}$.
As an example, if $G$ has the rules $S\to \sigma(X,Y)$, $X\to \tau(Y)$, $Y\to \tau(a)$, and $Y\to a$,
where $\sigma$, $\tau$, and $a$ have ranks~2, 1 and~0, respectively, then $G^\mathrm{der}$
has the rules $S'\to S_3(\sigma,X',Y')$, $X'\to X_2(\tau,Y')$, $Y'\to Y_2(\tau,a)$, and $Y'\to Y_1(a)$. 
The tree $\sigma(\tau(\tau(a)),a)$ is generated by $G$, 
and the derivation tree $S_3(\sigma,X_2(\tau,Y_2(\tau,a)),Y_1(a))$ by $G^\mathrm{der}$.

A (total deterministic) \emph{bottom-up finite-state tree automaton}
is specified as a tuple $A=(\Sigma,P,F,\delta)$ where $\Sigma$ is a ranked alphabet,
$P$ is a finite set of states, $F\subseteq P$ is the set of final states, and 
$\delta$ is the state transition function such that $\delta(\sigma,p_1,\dots,p_m) \in P$ 
for every $\sigma\in \Sigma$ and $p_1,\dots,p_m\in P$,
where $m$ is the rank of~$\sigma$. For every $t\in T_\Sigma$, we define the state $\delta(t)$ in which $A$ arrives at the root of~$t$ recursively by $\delta(\sigma \,t_1\cdots t_m)=
\delta(\sigma,\delta(t_1),\dots,\delta(t_m))$. The tree language recognized by $A$ is 
$L(A)=\{t\in T_\Sigma\mid \delta(t)\in F\}$. 

A \emph{regular tree language} is a set of trees that can be generated by a regular tree grammar, 
or equivalently, recognized by a bottom-up finite-state tree automaton. 
The class of regular tree languages will be denoted by $\REGT$.
The basic properties of regular tree languages can be found in, 
e.g., \cite{GecSte,GecSte97,tata,Eng75}.

\smallpar{Regular look-around}  
Let $\Sigma$ be a ranked alphabet. 
A \emph{node test over} $\Sigma$ is a set of trees over $\Sigma$ with a distinguished node,
i.e., it is a subset of the set 
\[
T^\bullet_\Sigma = \{(t,u) \mid t\in T_\Sigma, \,u \in \cN(t)\}.
\]
Intuitively it is a property of a node of a tree. 

We introduce a new ranked alphabet $\Sigma \times \{0,1\}$, such that
the rank of $(\sigma,b)$ equals that of $\sigma$ in $\Sigma$.
For a tree $t$ over $\Sigma$ and a node $u$ of $t$
we define $\tmark(t,u)$ to be the tree over $\Sigma \times \{0,1\}$ 
that is obtained from $t$ by changing the label $\sigma$ of $u$ into~$(\sigma,1)$ and 
changing the label $\sigma$ of every other node into $(\sigma,0)$.
Thus, $\tmark(t,u)$ is $t$ with one ``marked'' node $u$.
A \emph{regular (node) test over} $\Sigma$ is a node test $T\subseteq T^\bullet_\Sigma$ 
such that its marked representation 
is a regular tree language, i.e., $\tmark(T)\in \REGT$.
Note that $\nothing$ and $T^\bullet_\Sigma$ are regular tests, and that the 
class of regular tests over $\Sigma$ is closed under the boolean operations 
complement, intersection, and union, 
because $\REGT$ is closed under those operations.
Hence every boolean combination of regular tests is again a regular test. 

For a tree language $L \subseteq T_\Sigma$ 
we define the node test 
\[
T(L) = \{(t,u)\in T^\bullet_\Sigma\mid t|_u\in L\}
\]
over $\Sigma$.
Intuitively it is a property of the distinguished node 
that only depends on the subtree at that node.
Clearly, if $L$ is regular then $T(L)$ is regular. 
A regular test of the form $T(L)$ with $L\in\REGT$ 
will be called a \emph{regular sub-test}.
Note that $T(T_\Sigma)=T^\bullet_\Sigma$ and $T(\nothing)=\nothing$. 
Note also that for regular tree languages $L$ and $L'$ over $\Sigma$, 
$T(L)\cap T(L')= T(L\cap L')$ and $T^\bullet_\Sigma\setminus T(L)=T(T_\Sigma\setminus L)$.
This shows that the class of regular sub-tests over $\Sigma$ is also closed under 
the boolean operations complement, intersection, and union.

For a given node test $T$ over $\Sigma$, we also wish to be able to apply $T$
to a node~$v$ of a tree $\tmark(t,u)$, where $v$ need not be equal to $u$. 
Thus, we define the node test $\mu(T)$ over $\Sigma \times \{0,1\}$ to consist 
of all $(\tmark(t,u),v)$ such that $(t,v)\in T$ and $u\in \cN(t)$. 
The test $\mu(T)$ just disregards the marking of $t$. 
It is easy to see that if~$T$ is regular, then so is $\mu(T)$. 

The reader familiar with monadic second-order logic (abbreviated \abb{mso} logic)
should realize that it easily follows from the result of 
Doner, Thatcher and Wright \cite{Don70,ThaWri}
that a node test is regular
if and only if it is \abb{mso} definable (see~\cite[Lemma~7]{bloem}).
A node test $T$ over $\Sigma$ is \abb{mso} definable
if there is an \abb{mso} formula $\phi(x)$ over $\Sigma$, 
with one free variable $x$, such that 
$T = \{(t,u) \mid t \models \phi(u)\}$,
where $t \models \phi(u)$ means that the formula $\phi(x)$
holds in $t$ for the node $u$ as value of $x$. 
The formulas of \abb{mso} logic on trees over $\Sigma$ use the atomic formulas 
$\lab_\sigma(x)$ and $\down_i(x,y)$, for every $i\in[1,\m_\Sigma]$, meaning that 
node $x$ has label $\sigma\in\Sigma$, and that $y$ is the $i$-th child of $x$, respectively.
In the literature, regular tests are also called \abb{mso} tests.

\section{Tree-Walking Tree Transducers}\label{sec:trans}

In this section we define tree-walking tree transducers, with and without regular look-around,
and discuss some of their properties. 

A \emph{tree-walking tree transducer} (with regular look-around),
in short \ttt,
is a finite state device with one reading head that walks from node to node
over its input tree following the edges in either direction.
In addition to testing the label and child number of the current node,
it can even test any regular property of that node. 
The output tree is produced recursively, in a top-down fashion.
When the transducer produces a node of the output tree, 
labeled by an output symbol of rank $k$, it branches into $k$ copies of itself,
which then proceed independently, in parallel, to produce the subtrees rooted at 
the children of that output node. 

The \ttt\ is specified as a tuple 
$M = (\Sigma, \Delta, Q, Q_0, R)$,
where $\Sigma$ and $\Delta$ are ranked alphabets of input and output symbols,
$Q$ is a finite set of states,
$Q_0 \subseteq Q$ is the set of initial states,
and $R$ is a finite set of rules.
The rules are divided into \emph{move~rules} and \emph{output~rules}.
Each move rule is of the form
$\tup{q,\sigma,j,T} \to \tup{q',\alpha}$ 
such that $q,q'\in Q$, $\sigma \in \Sigma$, 
$j\in[0,\m_\Sigma]$,
$T$ is a regular test over $\Sigma$ (specified in some effective way),
and $\alpha$ is one of the following \emph{instructions}: 
\[
\begin{array}{ll}
\stay, & \\
\up & \text{provided } j\neq 0, \text{ and} \\
\down_i & \text{with } 1 \le i \le \rank_\Sigma(\sigma).
\end{array}
\]
Each output rule is of the form 
$\tup{q,\sigma,j,T} \to 
\delta(\tup{q_1,\alpha_1}, \dots, \tup{q_k,\alpha_k})$
such that the left-hand side is as above, $\delta\in \Delta^{(k)}$, 
$q_1,\dots,q_k \in Q$, and
$\alpha_1,\dots,\alpha_k$ are instructions as above.
A rule $\tup{q,\sigma,j,T} \to \zeta$ with $T=T^\bullet_\Sigma$ 
will be written $\tup{q,\sigma,j} \to \zeta$.
The \ttt~$M$ is \emph{deterministic}, in short a \dttt, if 
$Q_0$ is a singleton, and  
$T\cap T'=\nothing$ for every two distinct rules $\tup{q,\sigma,j,T} \to \zeta$
and $\tup{q,\sigma,j,T'} \to \zeta'$ in~$R$.  
A \dttt\ with initial state~$q_0$ will be specified as $M = (\Sigma, \Delta, Q, q_0, R)$.

A \emph{configuration} $\tup{q,u}$ of the \ttt\ $M$ on a tree $t$ over $\Sigma$ 
is given by the current state $q$ of $M$ and the current position $u$ of the head of $M$ on $t$. 
Formally, $q\in Q$ and $u\in \cN(t)$. 
The set of all configurations of $M$ on~$t$ is denoted $\con(t)$,
i.e., $\con(t)=Q\times \cN(t)$.
A rule $\tup{q,\sigma,j,T} \to \zeta$ 
is \emph{applicable} to
a configuration $\tup{q',u}$ of~$M$ on $t$ if $q'=q$ and
$u$ satisfies the tests $\sigma$, $j$, and $T$,
i.e., $\sigma$ and $j$ are the label and child number of $u$, 
and $(t,u)\in T$. 
For a node $u$ of $t$ and an instruction~$\alpha$ we define the node $\alpha(u)$ of $t$ as follows: 
if $\alpha$ is $\stay$, $\up$, or $\down_i$, 
then $\alpha(u)$ equals $u$, is the parent of $u$, or is the
$i$-th child of $u$, respectively.

For every input tree $t\in T_\Sigma$ we define the regular tree grammar 
$G_{M,t}=(N,\Delta,\cS,R_{M,t})$ where $N=\con(t)$, $\cS=\{\tup{q_0, \rt_t}\mid q_0\in Q_0\}$ 
and $R_{M,t}$ is defined as follows.
Let $\tup{q,u}$ be a configuration of $M$ on $t$ and let 
$\tup{q,\sigma,j,T} \to \zeta$ be a rule of $M$ that is applicable to $\tup{q,u}$.
If $\zeta = \tup{q',\alpha}$ then $R_{M,t}$ contains the rule 
$\tup{q,u}\to \tup{q',\alpha(u)}$, and 
if $\zeta = \delta(\tup{q_1,\alpha_1}, \dots, \tup{q_k,\alpha_k})$
then $R_{M,t}$ contains the rule
$\tup{q,u}\to \delta(\tup{q_1,\alpha_1(u)},\dots,\tup{q_k,\alpha_k(u)})$.
The derivation relation $\Rightarrow_{G_{M,t}}$ will be written as $\Rightarrow_{M,t}$. 
The \emph{translation realized by} $M$, denoted $\tau_M$, is defined as 
$\tau_M = \{(t,s)\in T_\Sigma\times T_\Delta \mid s\in L(G_{M,t})\}$. In other words, 
$\tau_M = \{(t,s)\in T_\Sigma\times T_\Delta \mid 
\exists\,q_0\in Q_0: \tup{q_0,\rt_t}\Rightarrow^*_{M,t} s\}$.
Two \ttt's $M$ and $N$ are \emph{equivalent} if $\tau_M=\tau_N$.

The \emph{domain of} $M$, denoted by $\dom(M)$, is defined to be 
the domain of the translation $\tau_M$, i.e., 
$\dom(M)=\dom(\tau_M)=\{t\in T_\Sigma\mid \exists\, s\in T_\Delta: (t,s)\in\tau_M\}$. 
The \ttt\ $M$ is \emph{total} if $\dom(M)=T_\Sigma$. 

The \ttt\ $M$ is \emph{finitary} if $\tau_M$ is finitary, which means that 
$\tau_M(t)$ is finite (or equivalently, that
$G_{M,t}$ is finitary) for every input tree $t\in T_\Sigma$. 
All classical top-down tree transducers (with or without regular look-ahead)
and all macro tree transducers are finitary. 

If $M$ is deterministic, then at most one rule of $M$ is applicable to 
a given configuration. Hence $G_{M,t}$ is forward deterministic 
and $L(G_{M,t})$ is either empty or a singleton.
Thus, $\tau_M$ is a partial function from $T_\Sigma$ to $T_\Delta$
(and a total function if~$M$ is total). 
For every $(t,s)\in \tau_M$ the context-free grammar $G_{M,t}$ 
has exactly one derivation tree, with root label $\tup{q_0,\rt_t}$ and yield $s$.

Intuitively, the derivation relation $\Rightarrow_{M,t}$ of the grammar $G_{M,t}$
formalizes the computation steps of the \ttt\ $M$ on the input tree $t$, 
the derivations of $G_{M,t}$ are the sequential computations of $M$ on $t$, 
and the derivation trees of $G_{M,t}$, generated by the regular tree grammar $G_{M,t}^\mathrm{der}$, 
model the parallel computations of the independent copies of $M$ on $t$.
If $M$ is deterministic and $t\in\dom(M)$, 
then $M$ has exactly one parallel computation on $t$. 

A sentential form of $G_{M,t}$ will be called 
an \emph{output form} of $M$ on $t$. It is a tree $s\in T_\Delta(\con(t))$
such that $\tup{q_0,\rt_t} \Rightarrow^*_{M,t} s$ for some $q_0\in Q_0$.
Intuitively, such an output form $s$ consists on the one hand of $\Delta$-labeled nodes 
that were produced by~$M$ previously in the computation, using output rules, 
and on the other hand of leaves that represent the independent copies of $M$ 
into which the computation has branched previously, due to those output rules, 
where each leaf is labeled by the current configuration of that copy. 
An output form is \emph{initial} if it is the configuration $\tup{q_0,\rt_t}$ 
for some $q_0\in Q_0$, where $\rt_t$ is the root of $t$, and 
it is \emph{final} if it is in~$T_\Delta$,  
which means that all copies of $M$ have disappeared.
 
Intuitively, the computation steps of $M$ lead from one output form to another, as follows.
Let $s$ be an output form and let $v$ be a leaf of $s$ 
with label $\tup{q,u}\in \con(t)$.
If $\tup{q,u} \to \tup{q',\alpha(u)}$ is a rule of $G_{M,t}$,
resulting from a move rule $\tup{q,\sigma,j,T} \to \tup{q',\alpha}$ of $M$ 
that is applicable to configuration $\tup{q,u}$, as defined above, 
then $s \Rightarrow_{M,t} s'$ where $s'$ is obtained from $s$ 
by changing the label of $v$ into $\tup{q',\alpha(u)}$. 
Thus, this copy of $M$ just changes its configuration.
Moreover, if $\tup{q,u}\to \delta(\tup{q_1,\alpha_1(u)},\dots,\tup{q_k,\alpha_k(u)})$
is a rule of $G_{M,t}$,
resulting from an output rule $\tup{q,\sigma,j,T} \to 
\delta(\tup{q_1,\alpha_1}, \dots, \tup{q_k,\alpha_k})$ of $M$, as defined above, 
then $s \Rightarrow_{t,M} s'$ where $s'$ is obtained from $s$ 
by changing the label of~$v$ into $\delta$ and adding children $v1,\dots, vm$
with labels $\tup{q_1,\alpha_1(u)},\dots,\tup{q_m,\alpha_m(u)}$, respectively.
Thus, $M$ outputs~$\delta$, and for each child $vi$ it branches into a new process,
a copy of itself started in state~$q_i$ at the node $\alpha_i(u)$.
In the particular case that $k=0$, $s'$ is obtained from $s$ by changing the label 
of $v$ into $\delta$; thus, the copy of~$M$ corresponding to the node~$v$ of $s$ disappears. 
The translation $\tau_M$ realized by $M$ consists of all pairs of trees 
$t$ over $\Sigma$ and $s$ over $\Delta$ such that $M$ has a sequential computation on $t$ 
that starts with an initial output form and ends with the final output form $s$. 

Before giving an example of a tree-walking tree transducer, we define six properties of \ttt's 
that will be used throughout this paper. 

The \ttt\ $M$ is \emph{sub-testing}, abbreviated \tts, 
if the regular tests used by $M$ are regular sub-tests, i.e., 
only test the subtree at the current node.
Formally, for every rule $\tup{q,\sigma,j,T} \to \zeta$ 
there is a regular tree language $L$ over $\Sigma$ such that 
$T=T(L)$. Recall that $T(L)=\{(t,u)\mid t|_u\in L\}$.
Thus, informally, $M$ is sub-testing if it uses regular look-ahead 
rather than the more general regular look-around. 

The \ttt\ $M$ is \emph{local}, abbreviated \ttl,  
if it does not use regular tests, i.e., 
$T=T^\bullet_\Sigma$ ($=\{(t,u)\mid t\in T_\Sigma, u\in \cN(t)\}$) 
for every rule $\tup{q,\sigma,j,T} \to \zeta$.
So all its rules are written $\tup{q,\sigma,j} \to \zeta$.
Recall that $T^\bullet_\Sigma=T(T_\Sigma)$; thus,
every local~\ttt\ is sub-testing.
Note that in the formalism of the (non-local) \ttt\, 
the tests on $\sigma$ and $j$ could be dropped from
a rule $\tup{q,\sigma,j,T} \to \zeta$, 
because they can be incorporated in the regular test $T$.

The \ttt\ $M$ is \emph{top-down}, abbreviated \ttd, if it does not use 
the up-instruction in the right-hand sides of its rules. 
Due to the use of stay-instructions, a \ttd\ need not be finitary. 
It is straightforward to show that 
the finitary (deterministic) \ttds\ and \ttdl\ are equivalent to the 
classical nondeterministic (deterministic) top-down tree transducer,
with and without regular look-ahead, respectively; see the end of this section. 
Note that in the rules of a \ttds\ or \ttdl\
the test on the child number~$j$ could be dropped, 
because $j$ can be stored in the finite state if necessary. 

The \ttt\ $M$ is \emph{single-use}, abbreviated \ttsu, if it is deterministic and 
never visits a node of the input tree twice in the same state.
Formally, it should satisfy the following property: 
for every $t\in T_\Sigma$, $s'\in T_\Delta(\con(t))$, $s\in T_\Delta$, and $\tup{q,u}\in\con(t)$, 
if $\tup{q_0,\rt_t}\Rightarrow^*_{M,t} s' \Rightarrow^*_{M,t} s$ 
then $\tup{q,u}$ occurs at most once in $s'$. In other words,
for every $t\in \dom(M)$, no nonterminal occurs twice in the (unique) derivation tree~$d$ 
of the context-free grammar $G=G_{M,t}$. 
Note that, as discussed in the proof of Lemma~\ref{lem:finpump} (and the paragraph following it),
the configuration $\tup{q,u}$ cannot occur at two distinct nodes
on a path from the root of $d$ to a leaf.
The single-use property also forbids $\tup{q,u}$ to occur at two independent nodes of $d$.
It was introduced for attribute grammars in~\cite{Gan,GanGie,Gie}. 

The \ttt\ $M$ is \emph{pruning}, abbreviated \ttp, if it is a top-down \ttt\ 
of which each move rule is of the form $\tup{q,\sigma,j,T}\to \tup{q',\down_i}$, 
and each output rule is of the form  
$\tup{q,\sigma,j,T}\to \delta(\tup{q_1,\down_{i_1}},\dots,\tup{q_k,\down_{i_k}})$
such that $1\leq i_1 < \cdots < i_k\leq \rank(\sigma)$. 
Intuitively, a pruning \ttt\ is a \ttd\ without stay-instructions that, 
when arriving at an input node $u$, either removes $u$ and 
all its children except one (together with the descendants of those children),
or relabels $u$ and possibly removes some of its children (together with their descendants).
Since a \ttp\ does not use the stay-instruction, it is finitary 
(and single-use if it is deterministic). 
Every \ttps\ and \ttpl\ is equivalent to a classical \emph{linear} top-down tree transducer,
with and without regular look-ahead, but not vice versa because the latter transducer can 
generate an arbitrary finite number of output nodes at each computation step, rather than zero or one. 

The \ttt\ $M$ is \emph{relabeling}, abbreviated \ttr, if 
every rule of $M$ is an output rule of the form 
$\tup{q,\sigma,j,T} \to 
\delta(\,\tup{q_1,\down_1}, \dots, \tup{q_m,\down_m})$
where $m=\rank_\Sigma(\sigma)=\rank_\Delta(\delta)$. 
Thus, the label $\sigma$ is replaced by the label $\delta$.
Obviously, every relabeling~\ttt\ is pruning. 

We use the notation $\NTT$ for the class of translations
realized by tree-walking tree transducers, and $\FTT$ and $\TT$ for the subclasses realized 
by finitary and deterministic \ttt's, respectively. Thus, $\TT \subseteq \FTT \subseteq \NTT$.
The subclasses of $\NTT$, $\FTT$, and $\TT$ realized by \ttt's with the above six properties 
(and their combinations) are indicated 
by the superscripts `s' and `$\ell$', and 
the subscripts `$\downarrow$', `su', `pru', and `rel', as above. 
For instance, $\TTDS$ denotes the class of translations 
realized by deterministic tree-walking tree transducers that are both sub-testing and top-down.
Note that $\NTTL$ is a proper subclass of $\NTTS$, 
because a local \ttt\ of which all output symbols have rank~0 
can be viewed as a tree-walking automaton, 
which cannot recognize all regular tree languages by the result of~\cite{BojCol}.
 
By~\cite[Section~8.4]{thebook}, the \ttt\ is equivalent to the \abb{ms} 
tree-walking transducer of~\cite[Section~8.2]{thebook}.
As discussed in the Introduction, the \ttl\ generalizes the 
attributed tree transducer of~\cite{Ful},
which is required to be noncircular and hence finitary;
the deterministic attributed tree transducer is also required to be total.\footnote{The 
\ttt\ $M$ is \emph{circular} 
if there exist $t\in T_\Sigma$, $u\in \cN(t)$, $q\in Q$, and $s\in T_\Delta(\con(t))$ such that
$\tup{q,u}\Rightarrow^*_{M,t} s$ and $\tup{q,u}$ occurs in $s$.
Thus, $M$ is noncircular if and only if $G_{M,t}$ is nonrecursive for every $t\in T_\Sigma$,
which implies that $L(G_{M,t})$ is finite. 
Note that a total deterministic \ttt\ is noncircular if and only if for every $t\in T_\Sigma$, 
$u\in \cN(t)$, and $q\in Q$ there exists $s\in T_\Delta$ such that $\tup{q,u}\Rightarrow^*_{M,t} s$.
It can be shown that for every finitary \ttt\ there is an equivalent noncircular \ttt,
but that will not be needed in this paper. 
}
In the same way the deterministic \ttt\ generalizes the (deterministic) attributed tree transducer 
with look-ahead of~\cite{BloEng}.
In~\cite{Eng09} all tree-walking tree transducers are local. 

\begin{example}\label{exa:query}
Let $\Sigma= \{\sigma,e\}$ with 
$\rank_\Sigma(\sigma)=2$ and $\rank_\Sigma(e)=0$,
and let $\Delta=\{\sigma,e\}\cup[1,\m_\Sigma]$ 
with $\rank_\Delta(\sigma)=2$, $\rank_\Delta(e)=0$, and 
$\rank_\Delta(j)=0$ for every $j\in[1,\m_\Sigma]=\{1,2\}$. 
Moreover, let $T$ be an arbitrary regular node test over $\Sigma$.
For simplicity we assume that $T$ is not satisfied at the leaves of $t$, 
i.e., if $(t,u)\in T$ then $u$ is not a leaf of $t$.
For instance, $T$ consists of all $(t,u)\in T^\bullet_\Sigma$ such that $u$ has 
at least one ancestor that has exactly one child that is a leaf, and 
at least one descendant with that same property. 
We consider a total deterministic \ttt\ $M = (\Sigma, \Delta, Q, q_0, R)$ that performs $T$ as a query, 
i.e., for every input tree $t$ it outputs all nodes of $t$ that satisfy $T$, in pre-order.
More precisely, if $u_1,\dots,u_n$ are the nodes $u$ of $t$ such that $(t,u)\in T$, in pre-order,
then $M$ outputs the tree $s=\sigma(s_1,\sigma(s_2,\dots \sigma(s_n,e)\cdots ))$ where 
$s_i = \sigma(\cdots\sigma(\sigma(e,j_1),j_2)\dots,j_k)$ if $u_i=j_1j_2\cdots j_k$ with 
$j_1,j_2,\dots,j_k\in [1,\m_\Sigma]$. 
Note that the yield of $s$ is $eu_1eu_2\cdots eu_ne$. 
The transducer $M$ performs a left-to-right depth-first traversal of the input tree $t$
and applies the test $T$ to every node of $t$, in pre-order. Whenever $M$ finds a node $u_i$ 
that satisfies the test, it branches into two copies. The first copy outputs the tree $s_i$ 
with yield~$eu_i$, walking from $u$ to the root, and the second copy continues the traversal. 

Formally, $M$ has the set of states $Q=\{d,u_1,u_2,p,p'\}$ and initial state $q_0=d$.
Intuitively, $d$ stands for `down', $u_j$ for `up from the $j$-th child', and $p$ for `print'.  
It has the following rules, where $j'\in[0,\m_\Sigma]$, $j\in[1,\m_\Sigma]$, 
$T^\mathrm{c}=T^\bullet_\Sigma\setminus T$, and $\tau\in\Sigma$: 
\[
\begin{array}{llllll}
\tup{d,\sigma,j',T^\mathrm{c}} & \to & \tup{d,\down_1} & 
\tup{d,e,j} & \to & \tup{u_j,\up} \\[0.4mm]
\tup{d,\sigma,j',T} & \to & \sigma(\tup{p,\stay},\tup{d,\down_1}) \quad & 
\tup{d,e,0} & \to & e \\[2mm]
\tup{u_1,\sigma,j'} & \to & \tup{d,\down_2} & 
\tup{p,\tau,j} & \to & \sigma(\tup{p,\up},j) \\[0.4mm]
\tup{u_2,\sigma,j} & \to & \tup{u_j,\up} & 
\tup{p,\tau,0} & \to & e \\[0.4mm]
\tup{u_2,\sigma,0} & \to & e &  &  & 
\end{array}
\]
where the rule $\tup{p,\tau,j} \to \sigma(\tup{p,\up},j)$ abbreviates the two rules
\[
\tup{p,\tau,j} \to \sigma(\tup{p,\up},\tup{p',\stay}) \text{\quad and \quad} \tup{p',\tau,j} \to j. 
\]
The \ttt\ $M$ does not have any of the six properties defined above.
Note that $M$ is not single-use because it pays $n$ visits to the root of $t$ in state~$p$.
For the example test $T$ it is not clear whether there is a local \ttt\ equivalent to~$M$, 
but that does not seem likely.
\qed
\end{example}

\begin{example}\label{exa:exp}
Let $\Sigma=\{\sigma,e\}$ as in Example~\ref{exa:query}.
We consider a total deterministic local \ttt\ $M_\mathrm{exp}$ that translates each tree $t$ with $n$ leaves 
into the full binary tree of height $n$ with $2^n$ leaves. 
As in Example~\ref{exa:query}, it performs a depth-first left-to-right traversal of $t$,
and branches into two copies whenever it visits a leaf of~$t$. 
Formally, $M_\mathrm{exp} = (\Sigma, \Sigma, Q, q_0, R)$ with $Q=\{d,u_1,u_2,q\}$ and $q_0=d$.
Its rules are similar to those of $M$ in Example~\ref{exa:query}.  
In particular, the three rules for states $u_1$ and $u_2$ are the same. 
The rules for state $d$ are the following, 
with $j'\in[0,\m_\Sigma]$ and $j\in[1,\m_\Sigma]$:
\[
\begin{array}{lll}
\tup{d,\sigma,j'} & \to & \tup{d,\down_1}  \\[0.4mm]
\tup{d,e,j} & \to & \sigma(\tup{u_j,\up},\tup{u_j,\up})  \\[0.4mm]
\tup{d,e,0} & \to & \sigma(e,e)
\end{array}
\]
where the last rule abbreviates the two rules 
$\tup{d,e,0} \to \sigma(\tup{q,\stay},\tup{q,\stay})$ and $\tup{q,e,0} \to e$. 
\qed
\end{example}

An elementary property of the translation realized by a deterministic \ttt\ is that it is 
of ``linear size-height increase'', as stated in the next lemma. 
Since the size of a tree is at most exponential in its height, this implies that 
it is of exponential size increase. 
This is well known for attributed tree transducers~\cite[Lemma~4.1]{Ful} 
(see also~\cite[Lemma~5.40]{FulVog})
and for local \ttt's~\cite[Lemma~7]{EngMan03}, and obviously also holds for~\ttt's. 
If, moreover, the \ttt\ is single-use, then it is of linear size increase.

\begin{lemma}\label{lem:sizeheight}
For every $\tau\in\TT$ there is a constant $c$ such that for every $(t,s)\in\tau$
the height of $s$ is at most $c\cdot|t|$. Moreover, $\TTSU\subseteq \LSI$. 
\end{lemma}

\begin{proof}
Let $M = (\Sigma, \Delta, Q, q_0, R)$ be a \dttt\ and let $(t,s)\in\tau_M$.
Let~$d$ be the unique derivation tree generated by $G^\mathrm{der}_{M,t}$.
Clearly, since each rule of $M$ outputs at most one node of $s$, 
the height of $s$ is at most the height of~$d$.
By Lemma~\ref{lem:finpump} the height of $d$ is at most $\#(\con(t))$, 
which equals $\#(Q)\cdot |t|$. Thus, we can take $c=\#(Q)$. 

It should also be clear that the size of $s$ is at most 
the number of nodes of~$d$ that are labeled by a configuration.
If $M$ is single-use, then no configuration occurs twice in $d$. 
Hence $|s|\leq \#(Q)\cdot |t|$, i.e., the function $\tau_M$ is of linear size increase. 
\qed
\end{proof}

Example~\ref{exa:exp} and Lemma~\ref{lem:sizeheight} imply that compositions of deterministic \ttt's
form a proper hierarchy. This was proved for attributed tree transducers 
in~\cite[Corollary~4.1]{Ful} (see also~\cite[Theorem~5.45]{FulVog}),
and the proof for \ttt's is exactly the same. 

\begin{proposition}\label{pro:hier}
For every $k\geq 1$, $\TT^k \subsetneq \TT^{k+1}$. 
\end{proposition}

\begin{proof}
Let $\tau_\mathrm{exp}$ be the translation realized by 
the \dttt\ $M_\mathrm{exp}$ of Example~\ref{exa:exp}.
Then $\tau_\mathrm{exp}\circ\tau_\mathrm{exp}$ translates each tree $t$ with $n$ leaves 
into the full binary tree of height~$2^n$ with $2^{2^n}$ leaves.
Since $|t|=2n-1$, it follows from Lemma~\ref{lem:sizeheight} that 
$\tau_\mathrm{exp}\circ\tau_\mathrm{exp}$ is not in $\TT$. Hence $\TT \subsetneq \TT^2$.
In a similar way it can be shown that $\tau_\mathrm{exp}^{k+1}$ is not in~$\TT^k$.
Since the size of a tree is at most exponential in its height, 
it follows from Lemma~\ref{lem:sizeheight} that for every $\tau\in\TT^2$ 
there is a constant $c$ such that for every $(t,s)\in\tau$
the height of $s$ is at most $2^{c\cdot|t|}$. Similarly for $\tau\in\TT^k$,
the height of $s$ is at most ($k-1$)-fold exponential in $|t|$. 
\qed
\end{proof}

Thus, in terms of size increase, a composition of $k$ \dttt's can create at most 
a $k$-fold exponentially large output tree, whereas a composition of $k+1$ \dttt's
can naturally create an output tree of $(k+1)$-fold exponential size.
In Section~\ref{sec:func} we will prove that compositions of nondeterministic \ttt's 
also form a hierarchy, with the same counter-examples.
One of our aims is to show that these hierarchies collapse for functions of linear size increase, 
i.e., that $\NTT^k \cap\LSI\subseteq \TT$ for every $k\geq 1$. 

We end this section by discussing some syntactic properties of \ttt's. 
First, for an arbitrary \ttt\ it may always be assumed that its output rules only use 
the stay-instruction: an output rule 
$\tup{q,\sigma,j,T} \to \delta(\tup{q_1,\alpha_1}, \dots, \tup{q_k,\alpha_k})$
can be replaced by the output rule
$\tup{q,\sigma,j,T} \to \delta(\tup{p_1,\stay}, \dots, \tup{p_k,\stay})$
and the move rules $\tup{p_i,\sigma,j,T} \to \tup{q_i,\alpha_i}$ for every $i\in[1,k]$,
where $p_1,\dots,p_k$ are new states. This replacement preserves determinism and the  
sub-testing, local, top-down, and single-use properties (but not pruning or relabeling). 

Second, we may always assume that the regular tests of a \ttt\ are disjoint. 
For a \ttt\ $M$, let $\cT_M$ be the set of regular tests in the left-hand sides of the rules of $M$.

\begin{lemma}\label{lem:mutdis}
For every \ttt\ $M$ there is an equivalent \ttt\ $M'$ such that the tests in~$\cT_{M'}$ 
are mutually disjoint. The construction preserves determinism and the 
sub-testing, local, top-down, single-use, pruning, and relabeling properties. 
\end{lemma}

\begin{proof}
If $T,T'\in \cT_M$ and $T\cap T'\neq \nothing$, then every rule $\tup{q,\sigma,j,T}\to \zeta$
can be replaced by the two rules $\tup{q,\sigma,j,T\cap T'}\to \zeta$ and 
$\tup{q,\sigma,j,T\setminus T'}\to \zeta$. The transducer $M'$ is obtained by repeating this procedure.
\qed
\end{proof}

Third, we can extend the definition of a \ttt\ 
$M = (\Sigma, \Delta, Q, q_0, R)$ by allowing 
``general rules'', which can generate any finite number of output nodes,
cf.~\cite[Lemma~2]{EngMan03}. Simple examples of general rules are 
$\tup{p,\tau,j} \to \sigma(\tup{p,\up},j)$ in Example~\ref{exa:query}
and $\tup{d,e,0}\to \sigma(e,e)$ in Example~\ref{exa:exp}. 
Formally, a \emph{general rule} is
of the form $\tup{q,\sigma,j,T}\to \zeta$ such that $\zeta$ is a tree 
in $T_\Delta(Q\times I_{\sigma,j})$, where $I_{\sigma,j}$ is the usual set of instructions: 
$\stay$, $\up$ (provided $j\neq 0$), and $\down_i$ with $i\in[1,\rank(\sigma)]$. 
If this rule is applicable to a configuration $\tup{q,u}$ of $M$ on $t\in T_\Sigma$,
then $G_{M,t}$ has the rule $\tup{q,u}\to \zeta_u$, where $\zeta_u$ is obtained from $\zeta$
by changing every label $\tup{q',\alpha}$ into $\tup{q',\alpha(u)}$. 
It is easy to see that a general rule can be replaced by 
the set of ordinary rules defined as follows. Let $p_u$ be a new state for every $u\in \cN(\zeta)$.
Then the rules are $\tup{q,\sigma,j,T}\to \tup{p_\epsilon,\stay}$, 
where $\epsilon$ is the root of $\zeta$, and all rules 
$\tup{p_u,\sigma,j,T}\to \lambda(\tup{p_{u1},\stay},\dots,\tup{p_{uk},\stay})$ 
where $\lambda$ is the label of $u$ in $\zeta$ and $k$~is its rank.
The first rule is a move rule that just changes state, and the latter rules 
output the $\Delta$-labeled nodes of $\zeta$ one by one ($\lambda\in\Delta$), and 
then make the required moves ($\lambda\in Q\times I_{\sigma,j}$). 
This construction preserves determinism and the sub-testing, local, top-down, and single-use properties. 
Note that the classical top-down tree transducer has general rules. 

If we allow general rules, then the stay-instruction is not needed any more in finitary \ttt's.
Let us say that a \ttt\ is \emph{stay-free} if it does not use the stay-instruction in its rules.
For every \ttt\ $M$ (with general rules) we can construct 
an equivalent stay-free \ttt\ $M_\mathrm{sf}$ with general rules, 
\emph{with possibly infinitely many rules} but such that the right-hand sides of rules 
with the same left-hand side form a regular tree language. 
If $M$ is finitary, then we can transform $M_\mathrm{sf}$ into 
an equivalent stay-free~\ttt\ with finitely many rules.
The construction is as follows, where we may assume that the node tests in $\cT_M$ 
are mutually disjoint, by (the proof of) Lemma~\ref{lem:mutdis}.

For every left-hand side $\tup{q,\sigma,j,T}$ of a rule of $M = (\Sigma, \Delta, Q, Q_0, R)$
we define a regular tree grammar 
$G_{q,\sigma,j,T}$ that simulates the computations of $M$, starting in 
a configuration $\tup{q,u}$ to which $\tup{q,\sigma,j,T}$ is applicable,
without leaving the current node $u$, i.e., 
executing stay-instructions only. Its set of nonterminals is $\{\tup{q',\stay}\mid q'\in Q\}$
with initial nonterminal $\tup{q,\stay}$.
Its set of terminals is $\Delta \cup D_{\sigma,j}$, where 
$D_{\sigma,j} =Q\times (I_{\sigma,j}\setminus \{\stay\})$ 
each element of which has rank~0. Finally, if $\tup{q',\sigma,j,T}\to \zeta$ is a rule of $M$
(with $q'\in Q$ and the same $\sigma$, $j$, and~$T$), 
then $G_{q,\sigma,j,T}$ has the rule $\tup{q',\stay}\to \zeta$. 

We now define $M_\mathrm{sf} = (\Sigma, \Delta, Q, Q_0, R_\mathrm{sf})$ 
where $R_\mathrm{sf}$ consists of all general rules $\tup{q,\sigma,j,T}\to \zeta$ 
such that $\zeta\in L(G_{q,\sigma,j,T})$,
for every left-hand side $\tup{q,\sigma,j,T}$ of a rule of $M$. 
Even if $M_\mathrm{sf}$ has infinitely many rules, 
it should be clear that (with all the definitions as in the finite case)
$M_\mathrm{sf}$ is equivalent to $M$. 

Note that if $M$ is deterministic, then so is $M_\mathrm{sf}$, because 
$G_{q,\sigma,j,T}$ is forward deterministic and hence $L(G_{q,\sigma,j,T})$
is empty or a singleton. Thus, $M_\mathrm{sf}$ has finitely many rules. 

Assume now that $M$, and hence $M_\mathrm{sf}$, is finitary.
Let $\tup{q,\sigma,j,T}$ be the left-hand side of a rule of $M$,
and let $D\subseteq D_{\sigma,j}$. If $M_\mathrm{sf}$ has infinitely many rules 
$\tup{q,\sigma,j,T}\to \zeta$ with $\zeta\in T_\Delta(D)$, then we remove those rules from $R_\mathrm{sf}$.
In fact, if~$M_\mathrm{sf}$ would have a computation 
$\tup{q_0,\rt_t}\Rightarrow^*_{M_\mathrm{sf},t} s$ with $q_0\in Q_0$
in which one of those rules is applied, then it would have a similar computation 
(with the same $q_0$ and $t$, but, in general, another $s$) 
in which any other of those rules is applied. Since $s$ contains 
at least as many occurrences of symbols in $\Delta$ as~$\zeta$, 
that would contradict the finitariness of $M_\mathrm{sf}$. 
Removing all these rules, for every $D\subseteq D_{\sigma,j}$, 
we are left with an equivalent version of $M_\mathrm{sf}$ with finitely many rules. 
The construction is effective because $L(G_{q,\sigma,j,T})\cap T_\Delta(D)$
is a regular tree language and hence its finiteness can be decided. 

The above constructions also preserve the sub-testing, local, top-down, and single-use properties.
Note that if $M$ is a finitary \ttds\ or \ttdl, then $M_\mathrm{sf}$ is a classical top-down tree transducer
(after incorporating the child number in its finite state), with or without regular look-ahead, respectively.

\section{Regular Look-Around}\label{sec:basic}

In this section we discuss some basic properties of \ttt's 
with respect to the feature of regular look-around.
We start with the simple fact that the domain of a \ttt\ can always 
be restricted to a regular tree language, except when the \ttt\ is local. 

\begin{lemma}\label{lem:domrestr}
For every \ttt\ $M$ and every $L\in\REGT$ there is a \ttt\ $M'$ such that 
$\tau_{M'}= \{(t,s)\in\tau_M\mid t\in L\}$. 
The construction preserves determinism and the 
sub-testing, top-down, single-use, pruning, and relabeling properties.
\end{lemma}

\begin{proof}
The \ttt\ $M'$ simulates $M$, but additionally verifies that the input tree~$t$ is in~$L$,
by using the regular sub-test $T(L)$ at the root of $t$. 
Formally, $M'$ is obtained from $M$ by changing every rule $\tup{q_0,\sigma,0,T}\to\zeta$ into 
$\tup{q_0,\sigma,0,T\cap T(L)}\to\zeta$, for every initial state~$q_0$.
\qed
\end{proof}

In the remainder of this section we show how to separate the regular look-around from a \ttt,
by incorporating it into another \ttt. 
We first prove that every \ttt\ $M$ can be decomposed into a deterministic relabeling \ttt\ $N$
and a local~\ttt\ $M'$. The relabeling \ttt\ $N$ preprocesses the input tree $t$ by adding to the label 
of each node $u$ of $t$ the truth values of the regular tests of $M$ at that node. 
This allows~$M'$, during its simulation of $M$, to inspect the new label of $u$ instead of testing $u$. 
The idea is similar to that of removing regular look-ahead in~\cite[Theorem~2.6]{Eng77}. 
The translation realized by $N$ is called an \abb{mso} relabeling in~\cite{BloEng,thebook}
and~\cite[Section~4]{EngMan99}.

\begin{lemma}\label{lem:msorel}
$\NTT \subseteq \TTR\circ \NTTL$, i.e., 
for every \ttt\ $M$ there are a deterministic relabeling \ttt\ $N$ and a local \ttt\ $M'$ 
such that $\tau_N\circ\tau_{M'}=\tau_M$.
The construction preserves determinism, the top-down property, and the pruning property.
\end{lemma}

\begin{proof}
Let $M = (\Sigma, \Delta, Q, Q_0, R)$ be a \ttt, and 
let $\cT$ be the set of regular tests in the left-hand sides of the rules in $R$.
By Lemma~\ref{lem:mutdis} we may assume that the tests in~$\cT$ are mutually disjoint. 
Now let $\cT_\bot=\cT\cup\{\bot\}$ where $\bot$ is the intersection of 
the complements of the tests in $\cT$. Thus, 
for every $t\in T_\Sigma$ and $u\in \cN(t)$, 
$(t,u)$ belongs to a unique node test in $\cT_\bot$. 
Let $\Sigma\times \cT_\bot$ be the ranked alphabet such that 
$\tup{\sigma,T}$ has the same rank as $\sigma$. 

We define the relabeling \ttt\ $N = (\Sigma,\Sigma\times \cT_\bot,\{p\},p,R_N)$ such that 
for every $\sigma\in \Sigma$, $j\in[0,\m_\Sigma]$, and $T\in\cT_\bot$, 
the output rule 
\[
\tup{p,\sigma,j,T}\to \tup{\sigma,T}(\tup{p,\down_1},\dots,\tup{p,\down_m})
\]
is in $R_N$, where $m$ is the rank of $\sigma$. 
Additionally we define the local \ttt\ $M'=(\Sigma\times\cT_\bot,\Delta,Q,Q_0,R')$ 
with the following rules.
If $\tup{q,\sigma,j,T}\to \zeta$ is a rule in~$R$, then $R'$ contains the rule
$\tup{q,\tup{\sigma,T},j}\to \zeta$.
Note that $N$ is total and deterministic. Also, if $M$ is deterministic, then so is $M'$.
It should be clear that $\tau_{M'}(\tau_N(t))=\tau_{M}(t)$ for every $t\in T_\Sigma$,
i.e., $\tau_N\circ \tau_{M'}= \tau_{M}$.
\qed
\end{proof}

We will also need a variant of this lemma, for nondeterministic \ttt's only.

\begin{lemma}\label{lem:nondetrel}
$\NTTS \subseteq \NTTRL\circ \NTTL$ and 
$\NTTPS \subseteq \NTTRL\circ \NTTPL$.
\end{lemma}

\begin{proof}
Let $M = (\Sigma, \Delta, Q, Q_0, R)$ be a sub-testing \ttt, and 
let $\cT$ be the set of regular tests in the left-hand sides of the rules in $R$.
As in the proof of Lemma~\ref{lem:msorel} we may assume that the tests in $\cT$ 
are mutually disjoint (by Lemma~\ref{lem:mutdis}), and we define 
$\cT_\bot=\cT\cup\{\bot\}$ as in that proof. Let $\cT_\bot=\{T(L_1),\dots,T(L_n)\}$ 
where $L_1,\dots,L_n$ are regular tree languages. Clearly, there is a  
bottom-up finite-state tree automaton $A=(\Sigma,P,F,\delta)$ (where $F$ is irrelevant)
and a partition $\{F_1,\dots,F_n\}$ of $P$ such that for every $t\in T_\Sigma$ and $i\in[1,n]$,
$t\in L_i$ if and only if $\delta(t)\in F_i$. 
We define the local relabeling \ttt\ $N=(\Sigma,\Sigma\times \cT_\bot,P,P,R_N)$
such that it nondeterministically simulates~$A$ top-down. 
For every $\sigma\in\Sigma$ of rank $m$, every sequence of states $p_1,\dots,p_m\in P$,
and every $j\in[0,\m_\Sigma]$, if $\delta(\sigma,p_1,\dots,p_m)=p\in F_i$,
then $R_N$ contains the rule 
$\tup{p,\sigma,j}\to \tup{\sigma,T(L_i)}(\tup{p_1,\down_1},\dots,\tup{p_m,\down_m})$. 
The local \ttt\ $M'$ is defined as in the proof of Lemma~\ref{lem:msorel}. 
\qed
\end{proof}

The next lemma is based on the folklore technique
of computing the states of a bottom-up finite-state tree automaton that are ``successful''
at the current node (see, e.g., the proofs of~\cite[Theorem~10]{BloEng} and~\cite[Theorem~8]{bloem}).
The lemma shows that every top-down \ttt\ is equivalent to 
one that is sub-testing, and hence to a classical top-down tree transducer with regular look-ahead
if it is finitary.
It is a slight generalization of the fact that every \abb{mso} relabeling can be computed by 
a top-down tree transducer with regular look-ahead, 
as shown in~\cite[Theorem~10]{BloEng} and~\cite[Theorem~4.4]{EngMan03}. 

\begin{lemma}\label{lem:d=ds}
$\NTTD=\NTTDS$. The construction preserves determinism, pruning, and relabeling. 
\end{lemma}

\begin{proof}
Let $M = (\Sigma, \Delta, Q, Q_0, R)$ be a \ttd\  
that uses a regular test $T$ over~$\Sigma$ in its rules.
For simplicity we first assume that $M$ uses $T$ in each of its rules. 
Let $A=(\Sigma\times\{0,1\},P,F,\delta)$ be a bottom-up finite-state tree automaton 
that recognizes $\tmark(T)$. We identify the symbols $(\sigma,0)$ and $\sigma$;
thus, $A$ can also handle trees over~$\Sigma$.
For every tree $t\in T_\Sigma$ and every node $u\in \cN(t)$, we define the set $\suc_t(u)$
of \emph{successful states} of $A$ at $u$ to consist of all states $p\in P$ such that 
$A$ recognizes~$t$ when started at $u$ in state $p$. To be precise, $\suc_t(\rt_t)=F$ and 
if $u$ has label $\sigma\in\Sigma^{(m)}$ and $i\in[1,m]$, then $\suc_t(ui)$ is 
the set of all states $p\in P$
such that $\delta(\sigma,p_1,\dots,p_{i-1},p,p_{i+1},\dots,p_m)\in \suc_t(u)$,
where $p_j=\delta(t|_{uj})$, i.e., $p_j$ is the state in which $\cA$~arrives at the $j$-th child of $u$, 
for every $j\in[1,m]\setminus\{i\}$.
Obviously, $\tmark(t,u)$ is recognized by $A$ if and only if 
$\delta((\sigma,1),\delta(t|_{u1}),\dots,\delta(t|_{um}))\in \suc_t(u)$. 

For every $\sigma\in \Sigma^{(m)}$ and every sequence of states $p_1,\dots,p_m\in P$
let $L_{\sigma,p_1,\dots,p_m}$ be the regular tree language consisting of all trees 
$\sigma(t_1,\dots,t_m)\in T_\Sigma$ such that $\delta(t_i)=p_i$ for every $i\in[1,m]$.
Thus, the regular sub-test $T(L_{\sigma,p_1,\dots,p_m})$ verifies that 
$A$ arrives at the $i$-th child of the current node
in state $p_i$ for every $i\in[1,m]$.

We construct a sub-testing \ttd\ $M' = (\Sigma, \Delta, Q', Q'_0, R')$ that is equivalent to~$M$.
It keeps track of $\suc_t(u)$ in its finite state. 
Its set of states is $Q'=Q\times \{S\mid S\subseteq P\}$ with set of initial states 
$Q'_0=\{(q_0,F)\mid q_0\in Q_0\}$.
The set of rules $R'$ is defined as follows.
Let $\tup{q,\sigma,j,T} \to \zeta$ be a rule in $R$, let $S\subseteq P$,
and let $p_1,\dots,p_m\in P$ such that 
$\delta((\sigma,1),p_1,\dots,p_m)\in S$ where $m=\rank_\Sigma(\sigma)$.
Then $R'$ contains the rule $\tup{(q,S),\sigma,j,T(L_{\sigma,p_1,\dots,p_m})} \to \zeta'$
where $\zeta'$ is obtained from $\zeta$ by changing 
every $\tup{q',\stay}$ into $\tup{(q',S),\stay}$ and 
every $\tup{q',\down_i}$ into $\tup{(q',S_i),\down_i}$
with $S_i=\{p\in P\mid \delta(\sigma,p_1,\dots,p_{i-1},p,p_{i+1},\dots,p_m)\in S\}$.

In the general case where $M$ uses regular tests $T_1,\dots,T_n$,
the transducer $M'$ must keep track of $\suc_t(u)$ 
for each of the corresponding bottom-up finite-state tree automata $A_1,\dots,A_n$.
\qed
\end{proof}

The proof of Lemma~\ref{lem:d=ds} also shows that in a rule $\tup{q,\sigma,j,T(L)} \to \zeta$
of a sub-testing \ttd\ we may assume that $L$ is of the form 
$L=\sigma(L_1,\dots,L_m)= \{\sigma(t_1,\dots,t_m)\mid t_1\in L_1,\dots,t_m\in L_m\}$
for regular tree languages $L_1,\dots,L_m$ (where $m=\rank(\sigma)$). 
This is how regular look-ahead is usually defined for classical top-down tree transducers. 

By Lemmas~\ref{lem:msorel} and~\ref{lem:d=ds}, $\TT \subseteq \TTDS\circ \TTL$.
It is proved in~\cite[Lemmas~49 and~50]{EngHooSam} that even $\TT \subseteq \TTDL\circ \TTL$,
but this will not be needed in what follows.\footnote{In~\cite{EngHooSam},
$\TT$ and $\TTL$ are denoted by $\dTTmso$ and $\dTT$, respectively.
}
Using Lemmas~\ref{lem:msorel} and~\ref{lem:d=ds} we can now prove three essential properties of \ttt's, 
based on well-known results from the literature. 

\begin{lemma}\label{lem:inversett}
The regular tree languages are closed under inverses of \ttt\ translations,
i.e., if $L\in\REGT$ and $\tau\in\NTT$, then $\tau^{-1}(L)\in\REGT$. 
\end{lemma}

\begin{proof}
Since the inverse of a composition is the composition of the inverses, 
it suffices to show this for \TTRS\ and \NTTL\ 
by Lemmas~\ref{lem:msorel} and~\ref{lem:d=ds}.
For \TTRS\ it follows from~\cite[Theorem~2.6 and Lemma~1.2]{Eng77},
and for \NTTL\ it is proved in~\cite[Lemma~3]{Eng09}.\footnote{We note that 
an alternative proof is by Lemma~\ref{lem:ttinmt} (in Section~\ref{sec:dmtmso})
and~\cite[Theorem~7.4]{EngVog85} (see also~\cite[Section~5]{PerSei}).
For the reader familiar with \abb{MSO} translations, see~\cite{thebook}, we note that it is proved 
in~\cite[Section~4]{EngMan99} that \TTRS\ is the class of \abb{mso} (tree) relabelings,
and that \REGT, which is the class of \abb{mso} definable tree languages, 
is closed under inverse \abb{MSO} (tree) transductions by~\cite[Corollary~7.12]{thebook}.
}
\qed
\end{proof}

\begin{corollary}\label{cor:regdom}
The domain of a \ttt\ $M$ is regular, i.e., $\dom(M)\in\REGT$. 
More generally, for every $k\geq 1$, if $\tau\in\NTT^k$ then $\dom(\tau)\in\REGT$. 
\end{corollary}

Corollary~\ref{cor:regdom} was proved for (nondeterministic) attributed tree transducers 
in~\cite{Bar82}, from which it is easy to conclude that Lemma~\ref{lem:inversett} holds for 
attributed tree transducers, as explained in~\cite[Lemma~3]{Eng09}. 

\begin{lemma}\label{lem:pru}
The regular tree languages are closed under pruning \ttt\ translations,
i.e., if $L\in\REGT$ and $\tau\in\NTTP$, then $\tau(L)\in\REGT$. 
\end{lemma}

\begin{proof}
By Lemma~\ref{lem:d=ds}, $\NTTP=\NTTPS$. As observed before, every $\tau\in\NTTPS$
can be realized by a classical linear top-down tree transducer
with regular look-ahead. It is well known that, due to linearity, 
$\REGT$ is closed under such translations,
see, e.g., \cite[Corollary~IV.6.7]{GecSte}.
\qed
\end{proof}

Lemma~\ref{lem:inversett}, Corollary~\ref{cor:regdom} and Lemma~\ref{lem:pru} are powerful technical 
tools because they allow us to show that certain node tests of a \ttt\ $M$ are regular
by defining them in terms of, e.g., the domains of other \ttt's or of variants of~$M$ itself.
In other words, a \ttt\ can use \ttt's ``to look around''. 
For instance, Lemma~\ref{lem:inversett} is used for this purpose in the proof of 
Lemma~\ref{lem:inversemso} below, where we show the following.

In a composition of a \dttt\ with a sub-testing \ttt\ 
the second transducer can even be assumed to be local, because the first transducer 
can determine the truth values of the regular sub-tests of the output tree by executing 
appropriate regular tests on its input tree. 

\begin{lemma}\label{lem:inversemso}
$\TT\circ\NTTS \subseteq \TT\circ\NTTL$. 
The construction preserves determinism (of the second transducer) and the top-down, single-use, pruning, 
and relabeling properties of both transducers.
\end{lemma}

\begin{proof}
Let $M_1 = (\Sigma, \Delta, Q, q_0, R)$ be a \dttt\ and 
let $M_2$ be a sub-testing \ttt\ with input alphabet $\Delta$.
We will construct a \dttt\ $M'_1$ and a local \ttt\ $M'_2$  
that simulate the composition of $M_1$ and $M_2$. 
The construction preserves the top-down, single-use, pruning, 
and relabeling property of each transducer, 
i.e., if $M_1$ has one of these properties, then so has $M'_1$,
and similarly for $M_2$ and $M'_2$. 
Moreover, if $M_2$ is deterministic, then so is $M'_2$. 

Let $(t,s)\in\tau_{M_1}$.
The \dttt\ $M'_1$ simulates $M_1$ on the input tree $t$. 
Simultaneously it executes the sub-tests of $M_2$ at every node $v$ of the output tree $s$ and 
preprocesses~$s$ by adding to the label 
of $v$ the truth values of these sub-tests at~$v$, cf. the text before Lemma~\ref{lem:msorel}. 
This allows $M'_2$, during its simulation of $M_2$ on $s$, to inspect the new label of $v$ 
instead of sub-testing $v$. 

Every node of $s$ is produced by an output rule of $M_1$ during its computation on $t$.
Let $\bar{s}$ be an output form of $M_1$ on $t$, 
and let $v$ be a leaf of $\bar{s}$ with label $\tup{q,u}$. 
It should be clear that $\tup{q,u}\Rightarrow^*_{M_1,t} s|_v$.
Now let $L$ be a regular tree language over $\Delta$ such that $M_2$ uses the sub-test $T'=T(L)$.
We claim that, in configuration $\tup{q,u}$, $M'_1$ can test whether $(s,v)\in T'$ by a  
regular test $\inv_q(T')$. Note that $(s,v)\in T(L)$ if and only if $s|_v\in L$.
Thus, $\inv_q(T')$ should test whether the output tree generated by the configuration $\tup{q,u}$
is in $L$. To prove that $\tmark(\inv_q(T'))$ is regular, 
we define a \dttt\ $N_q$ such that $\tmark(\inv_q(T'))=\tau_{N_q}^{-1}(L)$ 
and we use Lemma~\ref{lem:inversett}.
The transducer $N_q$ first uses a regular test at the root to verify that the input tree is 
of the form $\tmark(t,u)$.\footnote{To be precise, the regular sub-test $T(\tmark(T^\bullet_\Sigma))$.
}
After that it walks to the (unique) marked node $u$, 
using move rules to execute a depth-first search of the input tree,
and then simulates $M_1$ starting in state~$q$ at~$u$, producing the output tree $s|_v$. 
During that simulation it treats each symbol $(\sigma,0)$ or $(\sigma,1)$ as $\sigma$,
and for each regular test $T$ of $M_1$ it instead uses the test
$\mu(T)$, which is the set of all $(\tmark(t,u),v)$ 
such that $(t,v)\in T$ and $u\in \cN(t)$, see Section~\ref{sec:trees}.

The construction of $M'_1$ and $M'_2$ is similar to the construction of $N$ and $M'$ 
in the proof of Lemma~\ref{lem:msorel}.
Let $\cT$ be the set of regular tests in the left-hand sides of the rules of $M_2$.
As in the proof of Lemma~\ref{lem:msorel} 
we may assume that the tests in~$\cT$ are mutually disjoint (by Lemma~\ref{lem:mutdis}), 
and we define $\cT_\bot=\cT\cup\{\bot\}$ as in that proof. 
Note that the elements of $\cT_\bot$ are still regular sub-tests.
Note also that for every $q\in Q$, $t\in \dom(M_1)$ and $u\in \cN(t)$, 
$(t,u)$ belongs to a unique regular test in $\{\inv_q(T')\mid T'\in\cT_\bot\}$. 

We define the \dttt\ $M'_1 = (\Sigma,\Delta\times \cT_\bot,Q,q_0,R')$ such that 
$R'$ contains all move rules in $R$, and moreover, 
if $\tup{q,\sigma,j,T} \to \delta(\tup{q_1,\alpha_1},\dots,\tup{q_k,\alpha_k})$ 
is an output rule in $R$, then $R'$ contains the rule
\[
\tup{q,\sigma,j,T\cap \inv_q(T')} \to \tup{\delta,T'}(\tup{q_1,\alpha_1},\dots,\tup{q_k,\alpha_k})
\]
for every $T'\in\cT_\bot$. 
We define the local~\ttt\ $M'_2$ with input alphabet $\Delta\times\cT_\bot$ 
and the following rules.
If $\tup{q,\delta,j,T'}\to \zeta$ is a rule of $M_2$, then $M'_2$ has the rule
$\tup{q,\tup{\delta,T'},j}\to \zeta$.
It should now be clear that $\tau_{M_2'}(\tau_{M_1'}(t))=\tau_{M_2}(\tau_{M_1}(t))$ 
for every $t\in T_\Sigma$,
i.e., $\tau_{M'_1}\circ \tau_{M'_2}= \tau_{M_1}\circ \tau_{M_2}$.
If $M_1$ is single-use, then $M'_1$ is also single-use, because 
$M'_1$~visits the nodes of the input tree in the same states as~$M_1$; 
the same is true for $M_2$ and $M'_2$. Preservation of the other properties easily follows
from the construction of $M'_1$ and $M'_2$.
\qed
\end{proof}

\section{Composition}\label{sec:comp}

In this section we prove three composition results for \ttt's.
Our first aim is to prove that \dttt's are closed under right-composition 
with top-down \dttt's, and hence in particular with pruning \dttt's.
As already mentioned at the end of the Introduction, 
this generalizes the result of~\cite[Theorem~4.3]{Ful} for attributed tree transducers,
because \dttt's need not be total and they have regular look-around.
By Lemma~\ref{lem:d=ds} we may assume that the top-down \ttt\ is sub-testing.
It may even be assumed to be local by Lemma~\ref{lem:inversemso}.

\begin{lemma}\label{lem:ttottdl}
$\TT\circ\TTDL \subseteq \TT$. In particular 
\[
\TTD\circ\TTDL \subseteq \TTD \text{\quad and \quad} \TTP\circ\TTPL \subseteq \TTP.
\]
\end{lemma}

\begin{proof}
Since the domain of a \ttt\ can always be restricted to $\dom(M_1)$ 
by Lemma~\ref{lem:domrestr} and Corollary~\ref{cor:regdom}, it suffices to show 
that for every \dttt\ $M_1$ and every local top-down \dttt\ $M_2$,
a \dttt\ $M$ can be constructed such that 
$\tau_M(t)=\tau_{M_2}(\tau_{M_1}(t))$ for every input tree~$t\in \dom(M_1)$. 
For the case where $M_1$ is also local 
this construction was presented 
in the proof of~\cite[Theorem~55]{EngHooSam},
which can easily be adapted to the general case. 
We repeat it here for completeness sake, 
and because the proofs of the other two composition closure results 
will be based on it. 

The transducer $M$ is obtained by a straightforward product construction. 
For every $(t,s)\in \tau_{M_1}$, 
$M$ simulates $M_1$ on the input tree $t$ until $M_1$ uses an output rule that generates a node $v$ of $s$.
Then $M$ switches to the simulation of $M_2$ on $v$, as long as $M_2$ executes stay-instructions.
When $M_2$ executes a $\down_i$-instruction, $M$ switches again to the simulation of $M_1$ 
in order to generate the $i$-th child of $v$. 

Formally, let $M_1=(\Sigma,\Delta,P,p_0,R_1)$ and $M_2=(\Delta,\Gamma,Q,q_0,R_2)$. 
To simplify the construction of $M$ we assume that $M_1$ keeps track in its finite state 
of the child number of the output node to be generated. To be precise, we assume that 
there is a mapping $\chi: P\to [0,\m_\Delta]$ such that 
for every output form $s'$ and every leaf $v$ of $s'$ that is labeled by a configuration $\tup{p,u}$, 
the child number of~$v$ in~$s'$ is $\chi(p)$. 
That is possible because the output tree is generated top-down.
If~$M_1$ does not satisfy this assumption, then we change $M_1$ as follows. 
The new set of states is $P\times [0,\m_\Delta]$, and we define $\chi(p,i)=i$. 
The new initial state is $(p_0,0)$, 
because $M_1$ starts by generating the root of the output tree. 
Each move rule $\tup{p,\sigma,j,T}\to\tup{p',\alpha}$ of $M_1$ is changed into
the rules $\tup{(p,i),\sigma,j,T}\to\tup{(p',i),\alpha}$ 
and each output rule $\tup{p,\sigma,j,T}\to\delta(\tup{p_1,\alpha_1},\dots,\tup{p_k,\alpha_k})$
into $\tup{(p,i),\sigma,j,T}\to\delta(\tup{(p_1,1),\alpha_1},\dots,\tup{(p_k,k),\alpha_k})$,
for every $i\in [0,\m_\Delta]$.  
For the sake of the proof of Lemma~\ref{lem:ttsuottl} we note that this transformation of $M_1$ 
preserves the single-use property, because we have only added information to the states of $M_1$. 

The \dttt\ $M$ has input alphabet $\Sigma$ and output alphabet $\Gamma$. 
Its states are of the form $(p,q)$ or $(\rho,q)$, 
where $p\in P$, $q\in Q$, and 
$\rho$ is an output rule of~$M_1$, i.e., a rule of the form 
$\tup{p,\sigma,j,T}\to \delta(\tup{p_1,\alpha_1},\dots,\tup{p_k,\alpha_k})$. 
Its initial state is $(p_0,q_0)$. 
A state $(p,q)$ is used by $M$ to simulate the computation of $M_1$ 
that generates the next current node of $M_2$ when $M_2$ moves down
(keeping the state $q$ of $M_2$ in memory).
Initially $M$ simulates the computation of $M_1$ that generates the root of the output tree.
A state $(\rho,q)$ is used by $M$ to simulate the computation of~$M_2$
on the node that $M_1$ has generated with rule $\rho$. 
The rules of $M$ are defined as follows. 

First, rules that simulate $M_1$. Let $\rho: \tup{p,\sigma,j,T}\to\zeta$ be a rule in~$R_1$.
If $\zeta=\tup{p',\alpha}$, 
then $M$ has the rules $\tup{(p,q),\sigma,j,T}\to \tup{(p',q),\alpha}$ for every $q\in Q$. 
If $\rho$ is an output rule, then  
$M$ has the rules $\tup{(p,q),\sigma,j,T}\to \tup{(\rho,q),\stay}$ for every $q\in Q$.

Second, rules that simulate $M_2$.  
Let $\tup{q,\delta,i}\to \zeta$ be a rule in $R_2$ and 
let $\rho\colon \tup{p,\sigma,j,T}\to \delta(\tup{p_1,\alpha_1},\dots,\tup{p_k,\alpha_k})$
be an output rule in~$R_1$, with the same~$\delta$ and with $\chi(p)=i$. 
Then $M$ has the rule $\tup{(\rho,q),\sigma,j,T}\to \zeta'$ 
where $\zeta'$ is obtained from~$\zeta$ by changing 
every $\tup{q',\stay}$ into $\tup{(\rho,q'),\stay}$,
and every $\tup{q',\down_\ell}$ into $\tup{(p_\ell,q'),\alpha_\ell}$. 
Note that the test on $\sigma$, $j$, and $T$ is actually superfluous, 
because that was already tested when $M$ included $\rho$ in its state. 

It is easy to see that
$\tau_M(t)=\tau_{M_2}(\tau_{M_1}(t))$ for every input tree $t\in \dom(M_1)$.
If the rules of $M_2$ do not contain stay-instructions,
then $M$ does not need the states $(\rho,q)$. Its rules can then be simplified as follows.
Let $\tup{p,\sigma,j,T}\to\zeta$ be a rule in~$R_1$.
As above, if $\zeta=\tup{p',\alpha}$, then $M$ has the rules 
$\tup{(p,q),\sigma,j,T}\to \tup{(p',q),\alpha}$ for every $q\in Q$. 
If $\zeta=\delta(\tup{p_1,\alpha_1},\dots,\tup{p_k,\alpha_k})$
and $\tup{q,\delta,i}\to \zeta'$ is a rule in $R_2$, 
with the same~$\delta$ and with $\chi(p)=i$, 
then $M$ has the rule $\tup{(p,q),\sigma,j,T}\to \zeta''$ 
where $\zeta''$ is obtained from $\zeta'$ by changing 
every $\tup{q',\down_\ell}$ into $\tup{(p_\ell,q'),\alpha_\ell}$. 
This shows that if both $M_1$ and $M_2$ are pruning, then $M$ is pruning too. 
\qed
\end{proof}

We obtain our first composition closure result from
Lemmas~\ref{lem:d=ds}, \ref{lem:inversemso}, and~\ref{lem:ttottdl}. 
Note that the closure under composition of $\TTD$ 
already follows from Lemma~\ref{lem:d=ds} and~\cite[Theorem~2.11(2)]{Eng77}.

\begin{theorem}\label{thm:rightcomp}
$\TT \circ \TTD \subseteq \TT$. In particular, 
$\TTD$ and $\TTP$ are closed under composition. 
\end{theorem}

Theorem~\ref{thm:rightcomp} can be used to show that in a composition 
of two \dttt's we may always assume that 
the second one is local (thus strengthening Lemma~\ref{lem:inversemso}): 
by Lemma~\ref{lem:msorel} the second \ttt\ can be decomposed 
into a top-down \ttt\ and a local~\ttt, and then (by Theorem~\ref{thm:rightcomp}),
the top-down one can be absorbed by the first \ttt.
Hence $\TT \circ \TT \subseteq \TT \circ \TTD\circ \TTL \subseteq \TT\circ \TTL$.
This was already proved in~\cite[Theorem~53]{EngHooSam}
by means of pebble tree transducers.

Our second composition result generalizes Theorem~\ref{thm:rightcomp}
to nondeterministic \ttt's, restricted to right-composition with pruning \ttt's.
The proof of the next lemma is similar to 
that of Lemma~\ref{lem:ttottdl}.

\begin{lemma}\label{lem:nondetttottdl}
$\NTT\circ\NTTPL \subseteq \NTT$. In particular 
\[
\NTTD\circ\NTTPL \subseteq \NTTD \text{\quad and \quad}
\NTTP\circ\NTTPL \subseteq \NTTP.
\]
\end{lemma}

\begin{proof}
Let $M_1=(\Sigma,\Delta,P,P_0,R_1)$ be a \ttt\ and $M_2=(\Delta,\Gamma,Q,Q_0,R_2)$ a local pruning \ttt. 
The construction of the transducer $M$ such that $\tau_M=\tau_{M_1}\circ\tau_{M_2}$ 
is a straightforward variant of the one in 
the last paragraph of the proof of Lemma~\ref{lem:ttottdl}.
This time, we do not verify at the start that the input tree is in the domain of $M_1$,
because it has to be checked at each step of $M$ that $M_1$ can produce an output tree,
in particular when $M_2$ deletes part of that output tree
(cf. the proof of~\cite[Lemma~2.9]{Eng77}). 

We define $M=(\Sigma,\Gamma,P\times Q,P_0\times Q_0,R)$ as follows. 
As in the proof of Lemma~\ref{lem:ttottdl} we assume that 
$M_1$ keeps track in its finite state of the child number of the output node to be generated,
through a mapping $\chi:P\to[0,\m_\Sigma]$.
Let $\tup{p,\sigma,j,T}\to\zeta$ be a rule in $R_1$.
As before, if $\zeta=\tup{p',\alpha}$, then $M$ has the rules 
$\tup{(p,q),\sigma,j,T}\to \tup{(p',q),\alpha}$ for every $q\in Q$. 
If $\zeta=\delta(\tup{p_1,\alpha_1},\dots,\tup{p_k,\alpha_k})$
and $\tup{q,\delta,i}\to \zeta'$ is a rule in $R_2$, 
with the same $\delta$ and with $\chi(p)=i$, 
then $M$ has the rule $\tup{(p,q),\sigma,j,T\cap T'}\to \zeta''$ 
where $\zeta''$ is obtained (as before) from $\zeta'$ by changing 
every $\tup{q',\down_\ell}$ into $\tup{(p_\ell,q'),\alpha_\ell}$,
and the node test $T'$ consists of all $(t,u)$ such that for every $\ell\in[1,k]$
there exists a computation $\tup{p_\ell,\alpha_\ell(u)} \Rightarrow^*_{M_1,t} s_\ell$
for some $s_\ell\in T_\Delta$. Thus, the only difference with the proof of Lemma~\ref{lem:ttottdl}
is the additional test $T'$. In fact, it suffices that $T'$ tests every $\ell\in[1,k]$
for which $\down_\ell$ does not occur in $\zeta'$. That guarantees the existence of an output tree 
of $M_1$ on which $M_2$ is simulated by $M$. It should be clear that 
$T'$~is regular by Corollary~\ref{cor:regdom}:
it can be written as $\bigcap_{\ell\in[1,k]}T'_\ell$ where $\tmark(T'_\ell)$ is the domain of a \ttt\ 
that walks to node $\alpha_\ell(u)$ and then simulates $M_1$ starting in state $p_\ell$. 

We note that this construction does not work for an arbitrary top-down $M_2$ 
without stay-instructions. If some $\down_\ell$ occurs twice in $\zeta'$, 
then there are two occurrences $\tup{(p_\ell,q'),\alpha_\ell}$ and $\tup{(p_\ell,q''),\alpha_\ell}$
in $\zeta''$ and it is not guaranteed (as it should) that from both occurrences 
the same output subtree of $M_1$ is generated by $M$. 
We finally note that, as in the proof of Lemma~\ref{lem:ttottdl}, 
if both $M_1$ and $M_2$ are pruning, then so is $M$. 
\qed
\end{proof}

We obtain our second composition result from Lemma~\ref{lem:d=ds},
the second inclusion of Lemma~\ref{lem:nondetrel}, and two applications of 
Lemma~\ref{lem:nondetttottdl} (taking into account that \mbox{$\NTTRL\subseteq\NTTPL$}).

\begin{theorem}\label{thm:nondetrightcomp}
$\NTT \circ \NTTP \subseteq \NTT$. 
In particular $\NTTD \circ \NTTP \subseteq \NTTD$, and  
$\NTTP$~is closed under composition. 
\end{theorem}

Hence, also in a composition 
of two nondeterministic \ttt's we may always assume that 
the second one is local: 
$\NTT \circ \NTT \subseteq \NTT \circ \TTR\circ \NTTL \subseteq \NTT\circ \NTTL$
by Lemma~\ref{lem:msorel} and Theorem~\ref{thm:nondetrightcomp}, respectively. 

The range of a deterministic \ttt\ $M$ can be restricted to a regular tree language $L$
by restricting its domain to $\tau_M^{-1}(L)$, 
using Lemmas~\ref{lem:domrestr} and~\ref{lem:inversett}.
For a nondeterministic~\ttt\ we can use the next corollary.

\begin{corollary}\label{cor:ranrestr}
The translation $\tau'= \{(t,s)\in\tau\mid s\in L\}$ is in $\NTT$
for every $\tau\in\NTT$ and $L\in\REGT$.
If $\tau$ is in $\NTTD$ or $\NTTP$, then so is $\tau'$.
\end{corollary}

\begin{proof}
Let $\Sigma$ be the output alphabet of $\tau$ and let $A=(\Sigma,P,F,\delta)$ 
be a bottom-up finite-state tree automaton such that $L(A)=L$.
Obviously $\tau'=\tau\circ\tau_L$ where $\tau_L$ is the identity on $L$, 
and obviously $\tau_L\in \NTTRL$: it is realized by the local relabeling~\ttt\
$(\Sigma,\Sigma,P,F,R)$ where $R$ consists of all rules 
\[
\tup{p,\sigma,j}\to \sigma(\tup{p_1,\down_1},\dots,\tup{p_m,\down_m})
\]
such that $\delta(\sigma,p_1,\dots,p_m)=p$. 
By Theorem~\ref{thm:nondetrightcomp}, $\tau'$ satisfies the requirements.  
\qed
\end{proof}

Our third composition result is that
deterministic \ttt's are closed under left-composition with (deterministic) single-use \ttt's. 
This is a variant of one of the main results of~\cite{Gan,GanGie,Gie}  
for (a variant of) attribute grammars, cf. the last paragraph of~\cite{BloEng}.
It is proved for attributed tree transducers in~\cite[Theorem~3]{Kuh98}
(see also~\cite[Satz~6.5]{Kuh97}).

\begin{lemma}\label{lem:ttsuottl}
$\TTSU\circ\TTL \subseteq \TT$.
\end{lemma}

\begin{proof}
Let $M_1=(\Sigma,\Delta,P,p_0,R_1)$ and $M_2=(\Delta,\Gamma,Q,q_0,R_2)$
be a single-use \dttt\ and a local \dttt, respectively. 
We extend the proof of Lemma~\ref{lem:ttottdl} to the case that $M_2$ is an arbitrary local \dttt.
Thus, we have to deal with the fact that now $M_2$ can also move up on the output tree of $M_1$. 
Let $(t,s)\in \tau_{M_1}$, and let $d$ be the derivation tree 
of the computation $\tup{p_0,\rt_t}\Rightarrow^*_{M_1,t} s$. 
Since $M_1$ is single-use, we can identify each node of $d$ that is labeled by a configuration 
with that configuration, because a configuration $\tup{p,u}$ of $M_1$ occurs at most once in~$d$.
Suppose that $M_1$, in configuration $\tup{p,u}$ on~$t$, has generated a node $v$ of~$s$.
When $M_2$ executes an up-instruction at node~$v$, 
the new transducer $M$ has to backtrack on the computation of $M_1$,
back to the moment that the parent of $v$ in $s$ was generated by $M_1$. 
Thus, starting with the configuration 
$\tup{p,u}$ of $M_1$, $M$ has to determine the ancestors of $\tup{p,u}$ in $d$,
and stop at the first ancestor that is a configuration generating an output node. 
Since $M_1$ is single-use, each configuration $\tup{p,u}$ 
has a unique parent configuration $\tup{p',u'}$ in $d$.
That allows us to find $\tup{p',u'}$ by a regular test, as follows.

For every $p,p'\in P$ and every instruction $\alpha$ of $M_1$, we will define a regular test 
$T_{p,p',\alpha}$ such that for every $t\in\dom(M_1)$ and $u\in \cN(t)$, 
$(t,u)\in T_{p,p',\alpha}$ if and only if 
$\tup{p',\alpha(u)}$ is the parent of $\tup{p,u}$ in the derivation tree of the computation 
$\tup{p_0,\rt_t}\Rightarrow^*_{M_1,t} \tau_{M_1}(t)$.\footnote{For
the definition of $\alpha(u)$ see Section~\ref{sec:trans}.
}
We will construct a \ttt\ $N$ and define $T_{p,p',\alpha}=\{(t,u)\mid \tmark(t,u)\in\dom(N)\}$.
Then $T_{p,p',\alpha}$ is regular by Corollary~\ref{cor:regdom}.
To be able to describe $N$, we change notation and
consider the node test $T_{\bar{p},\bar{p}',\bar{\alpha}}$
for $\bar{p},\bar{p}'\in P$ and instruction $\bar{\alpha}$. 

Let $M'_1=(\Sigma,\nothing,P,\{p_0\},R'_1)$ be the nondeterministic \ttt\ obtained 
from $M_1$ by changing every output rule 
$\tup{p,\sigma,j,T}\to\delta(\tup{p_1,\alpha_1},\dots,\tup{p_k,\alpha_k})$
into the move rules $\tup{p,\sigma,j,T}\to\tup{p_i,\alpha_i}$ for every $i\in[1,k]$. 
Intuitively, for an input tree $t\in\dom(M_1)$, the tree-walking automaton $M'_1$ 
follows an arbitrary path in the unique derivation tree $d\in L(G^\mathrm{der}_{M_1,t})$,
from the root of $d$ down to the leaves. Whenever $M_1$~branches, 
$M'_1$ nondeterministically follows one of those branches. 
The transducer $N$, which is a variant of $M'_1$, has states $(p,p',\alpha)$ with $p,p',\alpha$ as above. 
The initial state is $(p_0,-,-)$, with the second and third component fixed, but irrelevant
(e.g., $(p_0,p_0,\stay)$). On a tree $\tmark(t,u)$, $N$~uses the state $(p,p',\alpha)$ to simulate 
the computations of $M'_1$ in state $p$ on~$t$,
but additionally keeps the previous configuration of $M'_1$ in its finite state, 
as the pair $(p',\alpha)$. When it arrives at the marked node $u$ 
in state $(\bar{p},\bar{p}',\bar{\alpha})$, it outputs a symbol of rank~0. 
Formally, let $\tup{p,\sigma,j,T}\to\zeta$ be a rule in $R'_1$,
let $p'\in P$, let $\alpha$ be an instruction, and let $b\in\{0,1\}$. 
Then $N$ has the rule $\tup{(p,p',\alpha),(\sigma,b),j,\mu(T)}\to\zeta'$
where $\tup{\tilde{p},\down_i}'=\tup{(\tilde{p},p,\up),\down_i}$, 
$\tup{\tilde{p},\up}'=\tup{(\tilde{p},p,\down_j),\up}$, and 
$\tup{\tilde{p},\stay}'=\tup{(\tilde{p},p,\stay),\stay}$
for every $\tilde{p}\in P$ and $i\in[1,\rank(\sigma)]$.
Additionally, $N$ has the rule $\tup{(\bar{p},\bar{p}',\bar{\alpha}),(\sigma,1),j,\mu(T)}\to\top$,
where $\top$ is its unique output symbol, of rank~0.
Thus, if the tree-walking automaton $N$ arrives in state $(\bar{p},\bar{p}',\bar{\alpha})$ 
at the marked node $u$, it can accept $\tmark(t,u)$. 
Hence, for every $t\in\dom(M_1)$, $N$ accepts $\tmark(t,u)$ if and only if
$\tup{\bar{p}',\bar{\alpha}(u)}$ is the parent of $\tup{\bar{p},u}$ 
in the derivation tree of the computation 
$\tup{p_0,\rt_t}\Rightarrow^*_{M_1,t} \tau_{M_1}(t)$.

The transducer $M$ is an extension of the one in the proof of Lemma~\ref{lem:ttottdl}.
It additionally has states $\fin_{p,q}$ and $\back_{p,q}$ 
to simulate the first and the following backward steps of the computation of $M_1$. 
Its rules are obtained as follows. 
First, it has the same rules that simulate (the forward computation of) $M_1$.
Second, the rules of $M$ that simulate $M_2$ are extended in such a way that, 
to obtain $\zeta'$ from $\zeta$, one has to change additionally 
every $\tup{q',\up}$ into $\tup{\fin_{p,q'},\stay}$.
Third, $M$ additionally has rules that simulate the backward computation of~$M_1$.
For each state $\fin_{p,q}$ it has all rules 
$\tup{\fin_{p,q},\sigma,j,T_{p,p',\alpha}} \to \tup{\back_{p',q},\alpha}$
(where the tests on~$\sigma$ and $j$ are irrelevant, 
because $M$ arrived in state $\fin_{p,q}$ by a stay-instruction).
For each state $\back_{p,q}$ it has the following rules. 
Let $\rho: \tup{p,\sigma,j,T}\to \zeta$ be a rule of~$M_1$. 
If $\rho$ is a move rule, then $M$ has all rules 
$\tup{\back_{p,q},\sigma,j,T\cap T_{p,p',\alpha}}\to \tup{\back_{p',q},\alpha}$.
If $\rho$ is an output rule, then $M$ has the rule
$\tup{\back_{p,q},\sigma,j,T}\to \tup{(\rho,q),\stay}$.
\qed
\end{proof}

\begin{theorem}\label{thm:ttsuott}
$\TTSU\circ\TT \subseteq \TT$.
\end{theorem}

\begin{proof}
It follows from Lemmas~\ref{lem:msorel}, \ref{lem:d=ds}, and~\ref{lem:inversemso} that 
\[
\TTSU\circ\TT \subseteq \TTSU\circ\TTRL\circ\TTL.
\]
Thus, by Lemma~\ref{lem:ttsuottl}, it suffices to show that $\TTSU\circ\TTRL \subseteq\TTSU$. 
For a single-use \dttt\ $M_1$ and a local relabeling \dttt\ $M_2$, 
consider the construction of the \dttt~$M$ in the last paragraph of the proof of Lemma~\ref{lem:ttottdl}.
It should be clear that $M$ is single-use: if $M_1$ visits an input node in state $p$,
then~$M$ visits that node in state~$(p,q)$ for some $q$. 
\qed
\end{proof}

It can be proved that $\TTSU$ is closed under composition, which also follows from 
Proposition~\ref{pro:mso=ttsu} in the next section. The inclusion $\TTSU\circ\TTRL \subseteq\TTSU$
in the previous proof is a special case of that.

\section{Macro and MSO}\label{sec:dmtmso}

In this section we collect some results on 
the connection between \ttt's, 
macro tree transducers (in short \mt's) and \abb{mso} tree transducers.
They are taken from the literature or can easily be proved 
using results from the literature.
This section can be skipped on first reading,
except that the reader interested in linear size increase 
should glance at Corollaries~\ref{cor:lsi} and~\ref{cor:lsidec}.

\subsection{Macro Tree Transducers}\label{sub:macro}

Let \NMT\ denote the class of translations realized by \mt's, 
with unrestricted or outside-in (\abb{OI}) derivation mode, 
let \MT\ denote the subclass realized by deterministic \mt's,
and let \tMT\ denote the class of \emph{total} translations in \MT\ 
(see~\cite{EngVog85} where they are denoted by \NMTOI, \DMT, and \DTMT, respectively).
We first consider the relationship between deterministic \ttt's and \mt's. 

It is proved in~\cite[Lemma~49 and Corollary~51]{EngHooSam}
that $\TT\subseteq \MT$, and in~\cite[Theorem~8.22]{thebook} 
(see also~\cite[Corollary~51]{EngHooSam}) that $\MT = \TTDL\circ \TTL$. 
Here we prove the following variant.

\begin{lemma}\label{lem:ttvsmt}
$\TT\subseteq \MT=\TTD\circ \TT$.
\end{lemma}

\begin{proof}
We first show that $\TTD\circ \MT\subseteq \MT$.
By Lemma~\ref{lem:d=ds} it suffices to show that $\TTDS\circ \MT\subseteq \MT$.
The inclusion $\TTDL\circ\MT\subseteq\MT$ is proved in~\cite[Theorem~7.6(3)]{EngVog85}.
As also argued before~\cite[Theorem~7.5]{EngManLSI}, this implies the inclusion 
$\TTDS\circ\MT\subseteq\MT$ as follows. 
By~\cite[Theorem~2.6]{Eng77} $\TTDS \subseteq \DBQREL\circ\TTDL$,
where $\DBQREL$ is the class of deterministic bottom-up finite-state relabelings.
Hence $\TTDS\circ\MT \subseteq \DBQREL\circ\MT$. 
Since \MT\ is closed under regular look-ahead by~\cite[Theorem~6.15]{EngVog85},
it is straightforward to prove that $\DBQREL\circ\MT \subseteq \MT$, 
similar to the proof of~\cite[Lemma~6.17]{EngVog85}.

By Lemma~\ref{lem:msorel}, $\TT\subseteq \TTD\circ\TTL$. 
It is proved in~\cite[Theorem~35 for $n=0$]{EngMan03} that 
$\TTL\subseteq \MT$.\footnote{By mistake, 
\cite[Theorem~35]{EngMan03} is stated for $n\geq 1$ only. 
It also holds for $n=0$ by~\cite[Lemma~34 and Theorem~31]{EngMan03}.
}
Hence $\TT\subseteq \TTD\circ\TTL\subseteq  \TTD\circ\MT\subseteq \MT$, which implies that
$\TTD\circ \TT\subseteq \TTD\circ \MT \subseteq \MT$.
It now remains to show that $\MT\subseteq \TTD\circ \TT$.
It is proved in~\cite[Section~5.5]{EngMan03} that $\tMT \subseteq \TTDL\circ \TTL$.
As shown in~\cite[Theorem~6.18]{EngVog85}, every translation $\tau\in\MT$ 
is the restriction to a regular tree language $L$ of a translation $\tau'\in\tMT$.
Hence $\tau'\in \TTDL\circ\TTL$ and so $\tau\in \TTD\circ\TTL$, 
because the first \ttt\ can start by verifying that the input tree is in $L$
with a regular test at the root of $t$, by Lemma~\ref{lem:domrestr}. 
\qed
\end{proof}

From Lemma~\ref{lem:ttvsmt}, together with Theorem~\ref{thm:rightcomp},
we obtain the following corollary on compositions. 

\begin{corollary}\label{cor:ttmttt}
For every $k\geq 1$, $\TT^k \subseteq \MT^k = \TTD\circ\TT^k \subseteq \TT^{k+1}$.
\end{corollary}

The above two inclusions are proper, cf.~\cite[Lemma~6.54]{FulVog} and~\cite[Theorem~4.16]{EngVog85}. 
In fact, the macro tree transducer is, 
and can be, of exponential height increase~\cite[Theorem~3.24]{EngVog85}. Hence $\tau_\mathrm{exp}^{k+1}$ 
is not in $\MT^k$, cf. the proof of Proposition~\ref{pro:hier}. Also, $\tau_M^k$ is not in $\TT^k$
where $M$ is an \mt\ that translates $\tau^na$ into $\tau^{2^n}a$ 
(with $\tau$ of rank~1 and $a$ of rank~0). 

The relationship between nondeterministic \ttt's and \mt's is less straightforward.
On the one hand, even $\NTTD$ is not included in $\NMT$ because all macro tree translations are finitary.
But we can express every \ttt\ as a composition of two top-down \ttt's and an \mt.

\begin{lemma}\label{lem:ttinmt}
$\NTT\subseteq \NTTD\circ \NTTD\circ \NMT$. 
\end{lemma}

\begin{proof}
By Lemma~\ref{lem:msorel}, $\NTT\subseteq \NTTD\circ \NTTL$. 
It follows from~\cite[Lemmas~34 and~27]{EngMan03} that $\NTTL\subseteq \MON\circ \NMT$,
where $\MON$ is a specific simple subclass of $\NTTDL$ 
defined before~\cite[Lemma~27]{EngMan03}.

We note that by Lemma~\ref{lem:msorel}, $\NTT\subseteq \TTR\circ \NTTL$
and that it is easy to prove that 
$\TTR\circ \NTTDL\subseteq \NTTD$. Hence we even obtain that 
$\NTT\subseteq \NTTD\circ \NMT$.
\qed
\end{proof}

On the other hand, every \mt\ can still be realized by a composition of two (finitary) \ttt's.

\begin{lemma}\label{lem:mtintt}
$\NMT\subseteq \TTD\circ\TT\circ\NTTP\subseteq \TTD\circ\FTT$.
\end{lemma}

\begin{proof}
By~\cite[Theorem~6.10]{EngVog85}, 
$\NMT=\tMT\circ\SET$, and by the proof of~\cite[Theorem~6.10]{EngVog85}, $\SET\subseteq \NTTPL$.
Hence $\NMT \subseteq \tMT\circ \NTTPL \subseteq \TTD\circ\TT\circ\NTTP$ by Lemma~\ref{lem:ttvsmt}.
That is included in $\TTD\circ\FTT$ by Theorem~\ref{thm:nondetrightcomp}.
\qed
\end{proof}

It can be shown that $\FTT\subseteq \NMT=\TTD\circ \FTT$, 
thus generalizing Lemma~\ref{lem:ttvsmt} to the finitary case,
but that will not be needed in what follows. 

Finally, let $\NMTIO$ denote the class of translations realized by \mt's 
with inside-out (\abb{io}) derivation mode (see~\cite{EngVog85}), and let 
$\mrNMT$ denote the class of translations realized by 
the multi-return macro tree transducers of~\cite{InaHos,InaHosMan}, 
which generalize \abb{io}~macro tree transducers. 

\begin{lemma}\label{lem:nmtio}
$\NMT_\text{\abb{io}}\subseteq \mrNMT\subseteq \FTTD\circ \TT$.
\end{lemma}

\begin{proof}
It is shown in~\cite[Lemma~5.5]{EngVog85} that $\NMTIO\subseteq \FTTDL\circ \YIELD$,
and in~\cite[Lemma~36]{EngMan03} that $\YIELD\subseteq \TTL$, and so $\NMTIO\subseteq \FTTDL\circ \TT$.
It follows from~\cite[Lemma~4]{InaHosMan} that $\mrNMT\subseteq \TTDL\circ\NMTIO\circ\TTDL$.
Hence 
\[
\mrNMT\subseteq \TTDL\circ\FTTDL\circ \TT\circ\TTD
\]
which is included in $\FTTDS\circ \TT$
by~\cite[Theorem~2.11(2)]{Eng77} and Theorem~\ref{thm:rightcomp}.
\qed
\end{proof}

\subsection{MSO Tree Transducers}\label{sub:mso}

Let \MSOT\ denote the class of deterministic \abb{mso} tree translations 
(see~\cite[Chapter~8]{thebook}, where it is denoted \DMSOT, 
and where \abb{mso} tree translations are called \abb{ms}-transductions of terms).
The next result is a variant of the main result of~\cite{BloEng}, 
which concerns attributed tree transducers with look-ahead instead of \ttt's. 
In its present form it is proved in~\cite[Theorems~8.6 and~8.7]{thebook}.

\begin{proposition}\label{pro:mso=ttsu}
$\MSOT=\TTSU$.
\end{proposition}

The next proposition is the main result of~\cite{EngManLSI}.

\begin{proposition}\label{pro:lsi}
$\tMT \cap \LSI \subseteq \MSOT$.
\end{proposition}

This can be extended to arbitrary deterministic \abb{oi} macro tree translations as follows.

\begin{lemma}\label{lem:lsi}
$\MT \cap \LSI \subseteq \MSOT$.
\end{lemma}

\begin{proof}
Since the domain $L$ of any \mt\ $M$ is regular (\cite[Theorem~7.4]{EngVog85}), 
and $\MT$ is closed under regular look-ahead (\cite[Theorem~6.15]{EngVog85}),
there is a total \mt\ $M'$ that extends $M$ by the identity on the complement of $L$. 
Clearly, $\tau_{M'}$ is of linear size increase if and only if $\tau_M$ is. 
Hence, by Propositions~\ref{pro:mso=ttsu} and~\ref{pro:lsi}, 
if $\tau_M$ is of linear size increase, then $\tau_{M'}$ is in $\TTSU$. 
And so $\tau_M$, which is the restriction of $\tau_{M'}$ to the regular tree language $L$, 
is also in $\TTSU$ by Lemma~\ref{lem:domrestr}.
\qed
\end{proof}

From Lemma~\ref{lem:ttvsmt}, Lemma~\ref{lem:lsi}, Proposition~\ref{pro:mso=ttsu},
and Lemma~\ref{lem:sizeheight} we obtain the following corollary. 

\begin{corollary}\label{cor:lsi}
$\TT \cap \LSI = \TTSU$.
\end{corollary}

It is also shown in~\cite{EngManLSI} that it is decidable for a total deterministic \mt\ 
whether or not it is of linear size increase. 
That also holds for arbitrary deterministic \mt's by the proof of Lemma~\ref{lem:lsi},
and hence also for \dttt's by Lemma~\ref{lem:ttvsmt}. 

\begin{corollary}\label{cor:lsidec}
It is decidable for a deterministic \ttt\ whether or not it is of linear size increase. 
\end{corollary}

Note that since Corollary~\ref{cor:lsi} is effective, 
if the \dttt\ is indeed of linear size increase, then
an equivalent \ttsu\ can be constructed.
One of our aims is to extend Corollaries~\ref{cor:lsi} and~\ref{cor:lsidec}
to arbitrary compositions of \dttt's.

\section{Functional Nondeterminism}\label{sec:func}

In this section we prove that for every nondeterministic top-down \ttt\ $M$
a deterministic top-down \ttt\ $M'$ can be constructed that realizes 
a ``uniformizer'' of $\tau_M$, i.e., a subset of $\tau_M$ with the same domain. 
This is a generalization of~\cite[Lemma]{Eng78}, where it is proved for 
classical nondeterministic top-down tree transducers. Note that, 
as opposed to the deterministic case, the nondeterministic top-down \ttt\
is more powerful than the classical nondeterministic top-down tree transducer with regular look-ahead,
because, due to the stay-instructions, it may not be finitary, i.e., 
it possibly translates one input tree into infinitely many output trees. 

A \emph{uniformizer} of a tree translation $\tau$ is a function $f$ such that $f\subseteq \tau$ 
and $\dom(f)=\dom(\tau)$. Intuitively, $f$ selects for every input tree $t\in\dom(\tau)$ 
one of the elements of $\tau(t)$.

\begin{lemma}\label{lem:cut}
Every $\tau\in \NTTD$ has a uniformizer $\tau'\in\TTD$. 
If $\tau\in \NTTP$, then \mbox{$\tau'\in\TTP$}.
\end{lemma}

\begin{proof}
Let $M = (\Sigma, \Delta, Q, Q_0, R)$ be a nondeterministic \ttd.
Without loss of generality we assume that $M$ has exactly one initial state $q_0$, 
i.e., $Q_0=\{q_0\}$.
We have to construct a deterministic \ttd\ $M'$  
that computes one possible output tree in $\tau_M(t)$ for every $t\in\dom(M)$.
The idea of the proof of~\cite[Lemma]{Eng78} is to pick, at the current node of $t$, 
one of the rules that lead to the generation of an output tree 
(which can be checked by a regular test). However, that idea does not work here,
because $M$ may have an infinite computation on~$t$ (see~\cite[New Observation~5.10]{Eng86}).
Thus, we have to be more careful. 
Note that an infinite computation is entirely due to the stay-instructions in the rules of~$M$.

The stay-instructions can be removed from $M$ by constructing the equivalent stay-free \ttt\ 
$M_\mathrm{sf} = (\Sigma, \Delta, Q, \{q_0\}, R_\mathrm{sf})$,
with general rules, as we did at the end of Section~\ref{sec:trans}. 
Recall that we assume that the regular tests in $\cT_M$ are mutually disjoint, and that
the set $R_\mathrm{sf}$ consists of all general rules $\tup{q,\sigma,j,T}\to \zeta$ 
such that $\zeta\in L(G_{q,\sigma,j,T})$,
for every left-hand side $\tup{q,\sigma,j,T}$ of a rule of $M$.
In this case $M_\mathrm{sf}$ is a top-down \ttt, with possibly infinitely many rules. 
Since its rules do not contain stay-instructions any more,
it does not have infinite computations on the trees in its domain. 
Thus, the idea above can be applied to $M_\mathrm{sf}$, 
which means that for every $q$, $\sigma$, $j$, and $T$ 
we have to pick one general rule $\tup{q,\sigma,j,T}\to \zeta$ from $R_\mathrm{sf}$, 
under the condition that its application leads to the generation of an output tree. 
This condition can be checked by a regular sub-test, as follows. 
Note that $\zeta\in T_\Delta(D_\sigma)$ where 
$D_\sigma=\{\tup{q',\down_i}\mid q'\in Q, \,i\in[1,\rank_\Sigma(\sigma)]\}$.

For every $\sigma\in\Sigma$, $q'\in Q$, and $i\in [1,\rank(\sigma)]$, 
let $T_{\sigma,q',i}$ be the node test over~$\Sigma$ consisting of all $(t,u)$ such that 
$u$ has label $\sigma$ in $t$ 
and there is a computation $\tup{q',ui}\Rightarrow^*_{M,t} s$ for some $s\in T_\Delta$. 
This node test is regular by Corollary~\ref{cor:regdom} 
because $\tmark(T_{\sigma,q',i})$ is the domain of 
a \ttt\ $M_{q',i}$
that on input $\tmark(t,u)$ walks to the marked node $u$, checks that its label is $\sigma$,
moves to the $i$-th child of $u$, and then simulates $M$ on $t$, starting in state $q$.
For every $\sigma\in \Sigma$ and $D\subseteq D_\sigma$, let 
$T_{\sigma,D}$ be the regular node test that is the intersection of all $T_{\sigma,q',i}$
such that $\tup{q',\down_i}\in D$ and all $T_\Sigma^\bullet\setminus T_{\sigma,q',i}$ 
such that $\tup{q',\down_i}\notin D$. Obviously the node tests $T_{\sigma,D}$ are mutually disjoint. 

We now define the deterministic \ttd\ $M' = (\Sigma, \Delta, Q, q_0, R')$, 
where $R'$ consists of the following general rules.
For every left-hand side $\tup{q,\sigma,j,T}$ of a rule of $M$
and every $D\subseteq D_\sigma$,
if $L(G_{q,\sigma,j,T})\cap T_\Delta(D)\neq\nothing$, then $R'$ contains  
the general rule $\tup{q,\sigma,j,T\cap T_{\sigma,D}}\to \zeta$ where $\zeta$ is a fixed element of 
$L(G_{q,\sigma,j,T})\cap T_\Delta(D)$.

It should be clear that $M'$ satisfies the requirements, i.e., it has the same domain as $M_\mathrm{sf}$
and it realizes a subset of $\tau_{M_\mathrm{sf}}$. Note that $M'$ can be constructed effectively,
because $L(G_{q,\sigma,j,T})\cap T_\Delta(D)$ is a regular tree language, and hence 
its nonemptiness can be decided and, if so, an element can be computed. 
Finally, the general rules of $M'$ can be replaced by ordinary rules, 
as discussed after Lemma~\ref{lem:mutdis}.
\qed
\end{proof}

At the end of this section we prove that 
any function that is realized by a composition of nondeterministic \ttt's
can also be realized by a composition of deterministic \ttt's. 
That will (only) be used to show that the results of Section~\ref{sec:lsi} also hold for 
nondeterministic \ttt's and \mt's. 
Let $\cF$ be the class of all partial functions from trees to trees. 

\begin{theorem}\label{thm:fun}
For every $k\geq 1$, $(\NTTD\circ\NTT^k)\cap \cF \subseteq \TTD\circ\TT^k$. 
\end{theorem}

\begin{proof}
By Lemmas~\ref{lem:ttinmt} and~\ref{lem:mtintt}, 
$\NTT\subseteq \NTTD\circ\NTTD\circ\TTD\circ\TT\circ\NTTD$.
Now let $\tau\in (\NTTD\circ\NTT^k)\cap \cF$.
Then $\tau=\tau_1\circ \cdots \circ \tau_m$
where $m=5k+1$, $\tau_{\,5j}\in\TT$ for every $j\in[1,k]$, and 
$\tau_i\in\NTTD$ for every $i\in[1,m]\setminus\{5j\mid j\in[1,k]\}$. 
By Corollary~\ref{cor:regdom},
the domain of a translation in $\NTT$ is regular.
Hence, we may assume that $\ran(\tau_i)\subseteq \dom(\tau_{i+1})$ for every $i\in[1,m-1]$.
If not, then we change $\tau_i$ into~$\bar{\tau}_i$ for \mbox{$i=m,\dots,1$} inductively as follows. 
First, $\bar{\tau}_m=\tau_m$. Second, for $i<m$ 
we obtain $\bar{\tau}_i$ from $\tau_i$ by restricting its range to  
$\dom(\bar{\tau}_{i+1})$, see Corollary~\ref{cor:ranrestr} and the paragraph preceding it. 

Since $\tau$ is a function, it should be clear that 
$\tau=\tau'_1\circ \cdots \circ \tau'_m$ where $\tau'_i\in\TTD$ 
is the uniformizer of $\tau_i$ that exists by Lemma~\ref{lem:cut}
if $\tau_i \in \NTTD$, and $\tau'_i=\tau_i$ if $\tau_i \in \TT$.
Thus, $\tau\in\TTD\circ(\TTD\circ\TTD\circ\TTD\circ\TT\circ\TTD)^k$ and so, 
by Theorem~\ref{thm:rightcomp}, $\tau\in\TTD\circ\TT^k$. 
\qed
\end{proof}

\begin{corollary}\label{cor:fun}
For every $k\geq 1$, $\NMT^k\cap \cF \subseteq \MT^k$. 
\end{corollary}

\begin{proof}
By the same argument as in the proof of Theorem~\ref{thm:fun},
using Lemma~\ref{lem:mtintt} only, we obtain that 
$\NMT^k\cap \cF \subseteq (\TTD\circ\TT\circ\TTD)^k$.
By Theorem~\ref{thm:rightcomp} that is included in $\TTD\circ\TT^k$,
which equals $\MT^k$ by Corollary~\ref{cor:ttmttt}.
\qed
\end{proof}

Since the inclusions in Corollary~\ref{cor:ttmttt} are proper, as discussed after that corollary,
Theorem~\ref{thm:fun} and Corollary~\ref{cor:fun} imply that $\NTT^k$ and $\NMT^k$ 
are also proper hierarchies, i.e., $\NTT^k\subsetneq \NTT^{k+1}$ and $\NMT^k\subsetneq \NMT^{k+1}$
for every $k\geq 1$.

\section{Productivity}\label{sec:prod}

In this section we prove that every \ttt\ can be decomposed into a pruning \ttt\ 
and another \ttt\ such that the composition is linear-bounded.
It implies that we may always assume that a composition of two \ttt's is linear-bounded.  
Recall from Section~\ref{sec:trees} that the composition 
of tree translations $\tau_1\subseteq T_\Sigma\times T_\Delta$ and 
$\tau_2\subseteq T_\Delta\times T_\Gamma$ is 
linear-bounded if there is a constant $c\in\nat$ such that 
for every $(t,s)\in \tau_1\circ\tau_2$ there exists $r\in T_\Delta$ such that 
$(t,r)\in\tau_1$, $(r,s)\in\tau_2$, and $|r|\leq c\cdot|s|$. 
Formally we say that the pair $(\tau_1,\tau_2)$ is linear-bounded. 
Recall also that for classes $\cT_1$ and $\cT_2$ of tree translations, 
the class $\cT_1\ast\cT_2$ consists of all 
translations $\tau_1\circ\tau_2$ such that 
$\tau_1\in\cT_1$, $\tau_2\in\cT_2$, and $(\tau_1,\tau_2)$ is linear-bounded.
Two elementary properties of this class operation were stated in Lemma~\ref{lem:ast}.
We will prove the following theorem. 

\begin{theorem}\label{thm:prod}
$\NTT\subseteq \NTTP \ast \NTT$ and $\TT\subseteq \TTP \ast \TT$.
\end{theorem}

Since pruning \ttt's can be absorbed to the right by arbitrary \ttt's 
(by Theorems~\ref{thm:nondetrightcomp} and~\ref{thm:rightcomp}), 
Theorem~\ref{thm:prod} can be generalized to compositions of \ttt's.
It implies that we may always assume that a composition of a \ttt\ 
with any number of \ttt's is linear-bounded. 

\begin{corollary}\label{cor:prod}
Let $k\geq 1$.
\begin{enumerate}
\item[$(1)$] $\NTT^k\subseteq \NTTP\ast\NTT^k$ and $\NTT\circ\NTT^k = \NTT\ast\NTT^k$, and
\item[$(2)$] $\TT^k\subseteq \TTP\ast\TT^k$ and $\TT\circ\TT^k = \TT\ast\TT^k$.
\end{enumerate}
\end{corollary}

\begin{proof}
(1) The proof of the inclusion is by induction on $k$. For $k=1$ it is Theorem~\ref{thm:prod}.
The induction step is proved as follows:
\[
\begin{array}{lll}
\NTT\circ\NTT^k & \subseteq & \NTT\circ (\NTTP\ast\NTT^k) \\
                & \subseteq & (\NTT\circ \NTTP)\ast\NTT^k \\
                & \subseteq & \NTT\ast\NTT^k \\
                & \subseteq & (\NTTP \ast \NTT)\ast\NTT^k \\
                & \subseteq & \NTTP \ast (\NTT\circ\NTT^k)
\end{array}
\]
where the first inclusion is by the induction hypothesis and the remaining inclusions are 
by Lemma~\ref{lem:ast}, Theorem~\ref{thm:nondetrightcomp} (which says that 
$\NTT\circ \NTTP \subseteq \NTT$), Theorem~\ref{thm:prod}, and Lemma~\ref{lem:ast} again.
The equation now follows from the inclusions above. 

(2) The proof is exactly the same as in~(1), using Theorem~\ref{thm:rightcomp} instead of
Theorem~\ref{thm:nondetrightcomp}.
\qed
\end{proof}

The remainder of this section is devoted to the proof of Theorem~\ref{thm:prod}.
It is essentially a variant of the proof of~\cite[Lemma~4.1]{EngCompAG},
which is the key lemma of~\cite{EngCompAG} and concerns 
the removal of ``superfluous computations'' in attribute grammars.
In its turn, that proof generalized the proof of~\cite[Lemma~1]{Bak78}
where this was done for top-down tree transducers (and strangely enough, the author
of~\cite{EngCompAG} did not mention that). 

To prove Theorem~\ref{thm:prod} it suffices, 
by Lemma~\ref{lem:msorel}, Lemma~\ref{lem:ast}, and 
Theorems~\ref{thm:nondetrightcomp} and~\ref{thm:rightcomp},
to consider local \ttt's, i.e., to prove that 
$\NTTL\subseteq \NTTP \ast \NTT$ and that $\TTL\subseteq \TTP \ast \TT$.
We prove the first and second inclusion in a first and second subsection, respectively. 
In the first subsection we additionally take care that the construction preserves the determinism 
of the given \ttt.

\subsection{Nondeterministic Productivity}\label{sub:nondetprod}

Let $M=(\Sigma,\Delta,Q,Q_0,R)$ be a \ttt.
For a pair $(t,s)\in\tau_M$ and a computation $\tup{q_0,\rt_t}\Rightarrow^*_{M,t} s$ with $q_0\in Q_0$,
we say that a node $u$ of $t$ is \emph{productive} (in that computation) if there is a $q\in Q$
such that an output rule is applied to the configuration $\tup{q,u}$ in the computation.
Obviously, the size of $s$ is at least the number of productive nodes of $t$.
For $i\in\{0,1\}$ we define the computation to be \emph{$i$-productive}
if all nodes of $t$ of rank~$i$ are productive.\footnote{Recall from Section~\ref{sec:trees} that 
the rank of a node is the rank of its label, i.e., the number of its children.
}
Moreover, the computation is \emph{productive} if it is both $0$-productive and $1$-productive,
i.e., all leaves and monadic nodes of $t$ are productive.
Finally, we define $\tau^0_M$ to consist of all $(t,s)\in\tau_M$ 
for which there is a 0-productive computation 
$\tup{q_0,\rt_t}\Rightarrow^*_{M,t} s$ for some $q_0\in Q_0$, and 
we define $\tau^{01}_M$ to consist of all 
$(t,s)\in\tau_M$ for which there is a productive computation of that form. 
Since the size of $t$ is at most twice the number of leaves 
plus the number of monadic nodes of $t$,\footnote{To be
precise, $|t|\leq (2\cdot|t|_0-1)+|t|_1$ where $|t|_0$ and $|t|_1$ 
are the number of leaves and monadic nodes of $t$, respectively. 
}
it follows that $|t|\leq 2\cdot|s|$ for every $(t,s)\in\tau^{01}_M$.

To prove that $\NTTL\subseteq \NTTP \ast \NTT$, 
our goal is to construct, for a given \ttl\ $M$, 
a pruning \ttt\ $N$ and a \ttl\ $M'$ 
in such a way that 
$\tau_N\circ\tau_{M'}\subseteq\tau_M$ and $\tau_M\subseteq \tau_N\circ \tau^{01}_{M'}$.
This obviously implies that $\tau_N\circ\tau_{M'}=\tau_M$.
The second inclusion says that 
for every $(t,s)\in\tau_M$ there exists a tree $t'$ 
such that $(t,t')\in\tau_N$ and $(t',s)\in\tau^{01}_{M'}$. 
Thus, as observed above, $|t'|\leq 2\cdot|s|$,
and hence $(\tau_N,\tau_{M'})$ is linear-bounded (for the constant $c=2$).

To this aim, $N$ will remove sufficiently many unproductive nodes from the input tree, 
and add state transition information of $M$ to the labels of the remaining nodes,
thus allowing $M'$ to simulate $M$ without having to visit those unproductive nodes.
Since productivity of a node of the input tree $t$ depends on the computation of $M$ on $t$,
$N$ nondeterministically guesses which nodes to remove, and uses its regular tests 
to determine the possible behaviour of $M$ on the remaining nodes. 
To reduce the technical complexity of the proof,
the construction of $N$ and $M'$ will be done in two steps, 
removing unproductive leaves and monadic nodes in the first and second step,
respectively. 

\begin{lemma}\label{lem:nondet0-prod}
For every \ttl\ $M$ there are a \ttp\ $N$ and a \ttl\ $M'$
such that 
\[
\tau_N\circ\tau_{M'}\subseteq\tau_M \subseteq \tau_N\circ \tau^0_{M'}.
\]
If $M$ is deterministic, then so is $M'$. 
\end{lemma}

\begin{lemma}\label{lem:nondet1-prod}
For every \ttl\ $M$ there are a \ttp\ $N$ and a \ttl\ $M'$
such that 
\[
\tau_N\circ\tau_{M'}\subseteq\tau_M \text{\quad and \quad} \tau^0_M \subseteq \tau_N\circ \tau^{01}_{M'}.
\]
If $M$ is deterministic, then so is $M'$. 
\end{lemma}

It is easy to see that applying these lemmas one after the other, we have obtained the goal above; 
note that pruning \ttt's are closed under composition by Theorem~\ref{thm:nondetrightcomp}.
It remains to prove the two lemmas. The constructions in their proofs are  
similar to the removal of $\epsilon$-rules and chain rules 
from a context-free grammar, respectively.
As is well known, one should not remove these rules in the reverse order,
because the removal of $\epsilon$-rules can create new chain rules. Similarly in our case,
we should remove unproductive leaves and monadic nodes in that order, because the 
removal of unproductive leaves can create new unproductive monadic nodes. 
Note also that removing $\epsilon$-rules and chain rules in one construction 
is technically more complex.

\medskip
{\bf Proof of Lemma~\ref{lem:nondet0-prod}.}
Let $M=(\Sigma,\Delta,Q,Q_0,R)$ be a \ttl.
As discussed in the second paragraph after Proposition~\ref{pro:hier} (in Section~\ref{sec:trans}),
we may assume that the output rules of $M$ only use the stay-instruction. 
Let us consider $(t,s)\in\tau_M$ and 
a computation $\tup{q_0,\rt_t}\Rightarrow^*_{M,t} s$ with $q_0\in Q_0$. 
The idea of the construction of the \ttp\ $N$ and \ttl\ $M'$
is that $N$ (nondeterministically) preprocesses $t$ by removing the maximal subtrees of $t$ 
that consist of unproductive nodes only, and that $M'$ simulates $M$ on the rest of $t$. 
Let us say that a node $u$ of $t$ is \emph{superfluous} (in this computation)
if it is unproductive and all its descendants are unproductive. 
Note that the root of $t$ is \emph{not} superfluous. 
Thus, $N$ changes $t$ into $t'$ by pruning all superfluous nodes of $t$.
Moreover, it adds state transition information of $M$ to the labels of the remaining nodes 
to allow $M'$ on $t'$ to simulate the above computation of $M$ on $t$. 
In the resulting computation of~$M'$ on $t'$, the input tree~$t'$ of $M'$ 
has no superfluous nodes, which means in particular 
that all its leaves are productive. Note that, due to the removal of the superfluous nodes,
each remaining node loses its superfluous children. 
Since the pruning \ttt\ $N$ does not know which nodes are going to be superfluous in $M$'s computation,
it just nondeterministically removes subtrees of the input tree $t$
and adds to the label of each remaining node all possible state transitions of $M$ in computations on 
the removed subtrees that use move rules only.  
Whereas $N$ just guesses the superfluous nodes, it uses its regular tests 
to determine the state transitions of~$M$ on those nodes.

As intermediate alphabet we use the ranked alphabet $\Gamma$ consisting of all symbols 
$\tup{\sigma,(i_1,\dots,i_n),\gamma}$ 
such that $\sigma\in\Sigma$, $n\in[0,\rank(\sigma)]$, $1\leq i_1 < i_2<\cdots < i_n\leq \rank(\sigma)$, 
and $\gamma\subseteq Q\times Q$. The rank of $\tup{\sigma,(i_1,\dots,i_n),\gamma}$ is $n$.
In the case where $M$ is deterministic we require $\gamma$ to be a partial function from $Q$ to $Q$. 
Intuitively, a node~$u$ of $t$ with label $\sigma$ that is not removed by $N$, will be relabeled by 
$\tup{\sigma,(i_1,\dots,i_n),\gamma}$ such that
the subtrees at its children $ui$ with $i\notin\{i_1,\dots,i_n\}$ are removed by $N$
and $\gamma$ is the set of all $(q,\bar{q})$ such that $M$ has a computation from $\tup{q,u}$
to $\tup{\bar{q},u}$ (using move rules only) that visits one of the removed subtrees. 

Formally, we define $N=(\Sigma,\Gamma,\{p\},\{p\},R_N)$ with one state $p$. 
For every symbol $\tup{\sigma,(i_1,\dots,i_n),\gamma}$ in $\Gamma$ and every $j\in [0,\m_\Sigma]$,
it has the rule
\[
\tup{p,\sigma,j,T}\to 
\tup{\sigma,(i_1,\dots,i_n),\gamma}(\tup{p,\down_{i_1}},\dots,\tup{p,\down_{i_n}})
\]
where $T$ is defined as follows. Let $t\in T_\Sigma$ and let $u\in \cN(t)$. 
The state transition relation $\gamma$ is uniquely determined by $(i_1,\dots,i_n)$, and is expressed by $T$.
Let us say that a node $v\in \cN(t)$ is a \emph{ghost} if 
$v=uiw$ for some $i\notin\{i_1,\dots,i_n\}$ and $w\in\nat^*$. 
Moreover, let us say that a computation 
\[ 
\tup{q_1,u_1}\Rightarrow_{M,t} \tup{q_2,u_2}\Rightarrow_{M,t} \cdots \Rightarrow_{M,t}\tup{q_m,u_m},
\]
$m\geq 3$, is a \emph{ghost computation} from $\tup{q_1,u_1}$ to $\tup{q_m,u_m}$ if
$u_j$ is a ghost for every $j\in[2,m-1]$. 
Note that such a computation is due to move rules only, that it visits at least one ghost, 
and that the ghosts $u_2,\dots,u_{m-1}$
all belong to a subtree at the same child $ui$.   
Finally, for states $q,\bar{q}\in Q$ we will write $q\hookrightarrow\bar{q}$ 
if there is a ghost computation from $\tup{q,u}$ to $\tup{\bar{q},u}$.
We now define~$T$ to consist of all $(t,u)$ such that 
$\gamma=\{(q,\bar{q})\in Q\times Q\mid q \hookrightarrow \bar{q}\}$.
Note that $\gamma$~is indeed a partial function if $M$ is deterministic.
The test $T$ is regular because it is a boolean combination of tests 
$T_{q,\bar{q}}=\{(t,u)\mid q \hookrightarrow \bar{q}\}$, which are regular because 
the tree language $\{\tmark(t,u)\mid q \hookrightarrow \bar{q}\}$
is regular for every $(q,\bar{q})\in Q\times Q$ by Corollary~\ref{cor:regdom}:
it is the domain of a \ttt\ that first walks to $u$, then simulates a ghost computation of $M$ on~$t$
from $\tup{q,u}$ to $\tup{\bar{q},u}$, and finally outputs a symbol of rank~0. 
 
We define $M'=(\Gamma,\Delta,Q,Q_0,R')$ with the following rules. 
Let $\rho: \tup{q,\sigma,j}\to\zeta$ be a rule in $R$, and let 
$\tup{\sigma,(i_1,\dots,i_n),\gamma}$ be an element of $\Gamma$ (with the same~$\sigma$). 
If $\rho$ is an output rule or $\zeta=\tup{q',\alpha}$ with $\alpha\in\{\up,\stay\}$, then
$R'$ contains the rule $\tup{q,\tup{\sigma,(i_1,\dots,i_n),\gamma},j}\to\zeta$.
If $\zeta=\tup{q',\down_{i_k}}$ with $k\in[1,n]$, then 
$R'$~contains the rule $\tup{q,\tup{\sigma,(i_1,\dots,i_n),\gamma},j}\to\tup{q',\down_k}$.
Otherwise (i.e., $\zeta=\tup{q',\down_i}$ with $i\notin\{i_1,\dots,i_n\}$), 
$R'$ contains the rule $\tup{q,\tup{\sigma,(i_1,\dots,i_n),\gamma},j}\to\tup{\bar{q}, \stay}$
for every $(q,\bar{q})\in\gamma$. 
Note that if $M$ is deterministic, then so is $M'$. 

It should be clear that $\tau_N\circ\tau_{M'}\subseteq\tau_M$, because for every $t'\in\tau_N(t)$
the computations of $M'$ on $t'$ simulate computations of $M$ on $t$. 

To understand that $\tau_M \subseteq \tau_N\circ \tau^0_{M'}$, consider a computation 
$\tup{q_0,\rt_t}\Rightarrow^*_{M,t} s$ with $q_0\in Q_0$, and let  
$t'\in\tau_N(t)$ be such that all superfluous nodes of $t$ (in this computation) are removed.
Then it should be clear that the computation of~$M$ on~$t$ can be simulated 
by a computation $\tup{q_0,\rt_{t'}}\Rightarrow^*_{M',t'} s$ of $M'$ on $t'$.
In fact, if $M$~visits a superfluous child of the current (non-superfluous) node $u$ of $t$, 
then $M'$~just stays in the node $v$ corresponding to $u$ in $t'$
and changes its state to the one in which $M$ returns to $u$.
For a completely formal correctness proof one would have to formalize 
the obvious bijective correspondence $f$
between the non-superfluous nodes of $t$ and the nodes of $t'$. 
In fact, $f(\epsilon)=\epsilon$, and
if~$u$ is non-superfluous and $ui_1,\dots,ui_n$ are all the non-superfluous children of~$u$,
then $f(ui_k)= f(u)k$ for every $k\in[1,n]$. 
Note that $u$ and $f(u)$ have the same child number. 
However, the correctness of the construction 
should be clear without such a proof. The configurations $\tup{q,u}$ of $M$ on $t$, 
for every non-superfluous node~$u$, are simulated by the configurations 
$\tup{q,f(u)}$ of $M'$ on $t'$.
Finally, the above computation of $M'$ on~$t'$ is 0-productive,
because each leaf $f(u)$ of~$t'$ 
corresponds to a non-superfluous node $u$ of $t$ of which all descendants are superfluous,
i.e., to a productive node. Since $M'$ simulates $M$, it follows 
that $f(u)$ is a productive node of $t'$. 
This ends the proof of Lemma~\ref{lem:nondet0-prod}.

\medskip
{\bf Proof of Lemma~\ref{lem:nondet1-prod}.}
This proof is similar to the previous one. 
Let $M=(\Sigma,\Delta,Q,Q_0,R)$ be a \ttl.
Again, we assume that the output rules of~$M$ only use the stay-instruction.
And again, let us consider $(t,s)\in\tau_M$ and 
a computation $\tup{q_0,\rt_t}\Rightarrow^*_{M,t} s$ with $q_0\in Q_0$. 
This time we define a node of~$t$ to be \emph{superfluous} if 
it is unproductive (in this computation) and has rank~1. 
As before, $N$ changes $t$ into $t'$ by pruning all superfluous nodes of $t$,
and adds information to the labels of the remaining nodes 
to allow $M'$ on $t'$ to simulate the above computation of $M$ on $t$. 
Whereas in the previous case, $M'$ had to shortcut the subcomputations of $M$ 
on maximal subtrees of superfluous nodes, in the present case 
$M'$ has to shortcut the subcomputations of $M$ on maximal sequences $u_1,\dots,u_n$ of
superfluous nodes ($n\geq 1$), where $u_{i+1}$ is the unique child of $u_i$ for every $i\in[1,n-1]$. 
For such a sequence, the unique child $u_{n+1}$ of $u_n$ is non-superfluous,
and either $u_1$ is the root of~$t$, or the parent $u_0$ of $u_1$ is non-superfluous.
In the second case, a subcomputation of $M$ on $u_1,\dots,u_n$ is as follows.
When it moves from $u_0$ down to $u_1$, it either returns to $u_0$, or it walks to $u_{n+1}$.
And when it moves from $u_{n+1}$ up to $u_n$, it either returns to~$u_{n+1}$, or it walks to $u_0$.
In the first case, $M$ can only move from $u_{n+1}$ up to~$u_n$ and return to $u_{n+1}$.
Thus, to the label of every non-superfluous node $u$ of $t$ we have to add information both on 
trips to superfluous nodes above $u$ and trips to superfluous nodes below $u$.
In the first case, $u_{n+1}$ will be the root of $t'$. In the second case, 
$u_{n+1}$ will be the $i$-th child of $u_0$ in $t'$, where $i$ is the child number of $u_1$ in $t$.
Thus, the child number of $u_{n+1}$ changes from 1 to 0, or from 1 to $i$, respectively. 

As in the previous proof, the pruning \ttt\ $N$ does not know in advance which nodes 
are going to be superfluous in $M$'s computation. 
Thus, it just nondeterministically removes monadic nodes of the input tree $t$
and adds to the label of each remaining node all possible state transitions of $M$ 
in subcomputations on the removed nodes that use move rules only.  
Rather than constructing $N$ directly, it is more convenient to realize this pruning of $t$
by two consecutive pruning \ttt's $N_1$ and $N_2$, 
and use Theorem~\ref{thm:nondetrightcomp}.
The local relabeling \ttt\ $N_1$ nondeterministically marks monadic nodes of $t$, by possibly changing 
the label $\sigma$ of a monadic node into $\widehat{\sigma}$. 
The (deterministic) \ttt\ $N_2$ then removes the marked nodes, and relabels 
the unmarked nodes, adding the appropriate state transitions of $M$ (determined by regular tests). 
Since it is easy to construct $N_1$, we only discuss $N_2$. 

The intermediate alphabet $\Gamma$ now consists of all symbols $\tup{\sigma,j,U,\gamma}$ such that 
$\sigma\in\Sigma$, $j\in[0,\m_\Sigma)]$, 
$U\subseteq \{\up\}\cup\{\down_i\mid i\in[1,\rank(\sigma)]\}$, 
and $\gamma\subseteq Q\times (Q\times I)$, 
where $I$ is the set of all possible instructions. 
The rank of $\tup{\sigma,j,U,\gamma}$ is the rank of $\sigma$.
As before, in the case where $M$ is deterministic 
we require $\gamma$ to be a partial function from $Q$ to $Q\times I$. 
Intuitively, a node $u$ of $t$ with label $\sigma$ that is not marked by $N_1$, will be relabeled by 
$\tup{\sigma,j,U,\gamma}$ such that $j$ is its child number in~$t$, 
$\alpha\in U$ if and only if $\alpha(u)$ is marked by $N_1$,
and $\gamma$ is the set of all $(q,\tup{\bar{q},\beta})$ such that 
the following holds: $M$ has a computation 
from $\tup{q,u}$ to $\tup{\bar{q},\bar{u}}$ (using move rules only) 
that visits a maximal sequence of marked nodes, for some unmarked node $\bar{u}$ 
such that $\beta(v)=\bar{v}$, where $v$ and $\bar{v}$ are the nodes corresponding to 
$u$ and $\bar{u}$ in the tree~$t'$. 

We define $N_2=(\Sigma\cup\widehat{\Sigma},\Gamma,P,p_0,R_2)$, 
where $\widehat{\Sigma}=\{\widehat{\sigma}\mid \sigma\in\Sigma^{(1)}\}$, 
$P=\{p_j\mid j\in[0,\m_\Sigma]\}$, and $R_2$ is defined as follows. 
For every $\sigma\in\Sigma^{(1)}$ and $j,j'\in[0,\m_\Sigma]$ the transducer $N_2$ has the rule 
$\tup{p_j,\widehat{\sigma},j'}\to\tup{p_j,\down_1}$. Moreover, 
for every $\tup{\sigma,j,U,\gamma}\in\Gamma$ and $j'\in[0,\m_\Sigma]$ it has the rule 
\[
\tup{p_j,\sigma,j',T}\to \tup{\sigma,j,U,\gamma}(\tup{p_1,\down_1},\dots,\tup{p_m,\down_m})
\]
where $m=\rank(\sigma)$ and $T$ is defined as follows.
Let $\hat{t}$ be a tree over $\Sigma\cup\widehat{\Sigma}$ 
and let $u\in \cN(\hat{t}\,)$. 
We define $\pi(\hat{t}\,)$ to be the tree 
over $\Sigma$ that is obtained from $\hat{t}$ 
by changing every label $\widehat{\sigma}$ into $\sigma$. 
Both $U$ and $\gamma$ are uniquely determined, and they are expressed by $T$.
Let us say that a node $v\in \cN(\hat{t}\,)$ is a \emph{ghost} 
if its label is in~$\widehat{\Sigma}$.
A \emph{ghost computation} is defined as in the previous proof, 
for $t=\pi(\hat{t}\,)$; note that $\cN(t)=\cN(\hat{t}\,)$. 
And let us write $\tup{q,u}\hookrightarrow\tup{\bar{q},\bar{u}}$ if there is
a ghost computation from~$\tup{q,u}$ to $\tup{\bar{q},\bar{u}}$.
We now define $T$ to consist of all $(\hat{t},u)$ such that 
\begin{itemize}
\item $\up\in U$ if and only $u$ has a parent and that parent is a ghost,
\item $\down_i\in U$ if and only if $ui$ is a ghost,
\item $(q,\tup{\bar{q},\stay})\in\gamma$ if and only if 
$\tup{q,u}\hookrightarrow\tup{\bar{q},u}$,
\item $(q,\tup{\bar{q},\up})\in\gamma$ if and only if 
$\tup{q,u}\hookrightarrow\tup{\bar{q},\bar{u}}$
for some ancestor $\bar{u}$ of~$u$,
\item $(q,\tup{\bar{q},\down_i})\in\gamma$ if and only if 
$\tup{q,u}\hookrightarrow\tup{\bar{q},\bar{u}}$
for some descendant~$\bar{u}$ of~$ui$.
\end{itemize}
As before, if $M$ is deterministic, then $\gamma$ is indeed a partial function. 
It is straightforward to prove, using Corollary~\ref{cor:regdom}, that $T$ is regular;
we leave that to the reader. 
 
We define $M'=(\Gamma,\Delta,Q,Q_0,R')$ with the following rules. 
Let $\rho: \tup{q,\sigma,j}\to\zeta$ be a rule of $M$,
and let $\tup{\sigma,j,U,\gamma}$ be in $\Gamma$ (with the same $\sigma$ and~$j$). 
If $\rho$ is an output rule or $\zeta=\tup{q',\alpha}$ with $\alpha\notin U$, then
$R'$ contains the rule $\tup{q,\tup{\sigma,j,U,\gamma},j'}\to\zeta$ 
for every $j'\in [0,\m_\Sigma]$ (except $j'=0$ when $\alpha=\up$). 
If $\zeta=\tup{q',\alpha}$ with $\alpha\in U$, then
$R'$ contains the rule $\tup{q,\tup{\sigma,j,U,\gamma},j'}\to\tup{\bar{q},\beta}$
for every $(q,\tup{\bar{q},\beta})\in\gamma$ 
and every $j'\in [0,\m_\Sigma]$ (except $j'=0$ when $\beta=\up$).

Let $\tau=\tau_{N_1}\circ \tau_{N_2}$. 
It should be clear that $\tau\circ\tau_{M'}\subseteq\tau_M$, as in the previous proof. 
To understand that $\tau^0_M \subseteq \tau\circ \tau^{01}_{M'}$, consider a 0-productive computation 
$\tup{q_0,\rt_t}\Rightarrow^*_{M,t} s$ with $q_0\in Q_0$, and 
let $t'\in\tau(t)$ be obtained from $t$ by removing all superfluous nodes of $t$. 
As in the previous proof, there is an obvious bijective correspondence $f$ between 
the non-superfluous nodes of $t$ and the nodes of $t'$.
For a node $u$ of $t$ we define $g(u)=u$ if $u$ is non-superfluous, and 
$g(u)$ is the first (i.e., shortest) non-superfluous descendant of $u$ otherwise.
Then $f(g(\epsilon))= \epsilon$, and 
if $u$ is non-superfluous and $ui$ is a child of $u$, then $f(g(ui))= f(u)i$. 
And as before, there is a computation $\tup{q_0,\rt_{t'}}\Rightarrow^*_{M',t'} s$
of $M'$ on $t'$ that simulates the computation of $M$ on $t$, such that 
the configurations $\tup{q,u}$ of $M$, 
for every non-superfluous node~$u$ of $t$, are simulated by the configurations 
$\tup{q,f(u)}$ of $M'$. Since $\tau$ does not remove leaves of $t$, 
the computation of $M'$ is still 0-productive. Moreover, it is also 1-productive 
because all unproductive monadic nodes were removed by $\tau$. 
This ends the proof of Lemma~\ref{lem:nondet1-prod}.

\begin{remark}\label{rem:size}\rm
In the Introduction we observed that our main technical result can be viewed as 
a static garbage collection procedure, which leads, in principle, 
to algorithms for automatic compiler and XML query optimization.
For practical applicability our proof of this result is, however, of restricted value because 
the sizes of the involved transducers are blown up exponentially.
This is due to the fact that, in the proof of Lemmas~\ref{lem:nondet0-prod} and~\ref{lem:nondet1-prod}, 
the pruning \ttt\ $N$ uses regular tests to determine 
the relevant state transition information $\gamma\subseteq Q\times Q$ 
(or $\gamma\subseteq Q\times (Q\times I)$)
of the given \ttt\ $M$, due to its ghost computations.  
These regular tests are constructed through Corollary~\ref{cor:regdom}, 
applied to variants of $M$. Naturally, the number of states of the finite-state tree automaton 
recognizing the domain of such a variant is exponential in the number $\#(Q)$ of states of $M$, 
cf. the proof of~\cite[Lemma~1]{Eng09}.
If one now considers the proof of $\NTT\circ\NTT \subseteq \NTT\ast\NTT$ in Corollary~\ref{cor:prod}
(in which the pruning \ttt\ $N$ for the second \ttt\ $M$ is incorporated in the first \ttt\ 
by Theorem~\ref{thm:nondetrightcomp}),
it can be seen that the number of states of the first constructed \ttt\ is $2$-fold exponential
in the number of states of $M$. The additional exponential jump is due to 
Lemma~\ref{lem:d=ds}, which turns the pruning \ttt~$N$ into one that is sub-testing.
This implies that in the construction for the inclusion $\NTT\circ\NTT^k \subseteq \NTT\ast\NTT^k$
of Corollary~\ref{cor:prod}, the size of the first constructed \ttt\ can be $2(k-1)$-fold exponential
in the size of the last given \ttt. This will also hold for the deterministic version.
\qed
\end{remark}

\subsection{Deterministic Productivity}\label{sub:detprod}

Let $M=(\Sigma,\Delta,Q,q_0,R)$ be a deterministic \ttt.
For $t\in\dom(M)$ we say that a node~$u$ of $t$ is productive if 
it is productive in the computation $\tup{q_0,\rt_t}\Rightarrow^*_{M,t} \tau_M(t)$,
and we say that $t$ is productive (for $M$) if that computation is productive, i.e., 
if all leaves and monadic nodes of $t$ are productive.\footnote{There
are several such computations, but they all have the same unique derivation tree 
in $L(G^\mathrm{der}_{M,t})$. The definition of productivity clearly does not depend on 
the particular choice of the derivation.
} 
We define $L_{M,\text{prod}}$ to be the set of all productive trees $t\in\dom(M)$. 
Note that $\tau^{01}_{M}$ is the restriction of~$\tau_M$ to~$L_{M,\text{prod}}$.
The next lemma shows that the set of productive input trees is a regular tree language. 

\begin{lemma}\label{lem:prodtest}
Let $M=(\Sigma,\Delta,Q,q_0,R)$ be a deterministic \ttt. 
\begin{enumerate}
\item[$(1)$] There is a regular test $T_{M,\mathrm{prod}}$ over $\Sigma$ 
such that for every $t\in\dom(M)$ and \mbox{$u\in \cN(t)$}, 
$(t,u)\in T_{M,\mathrm{prod}}$ if and only if 
$u$ is productive.
\item[$(2)$] $L_{M,\text{prod}}$ is a regular tree language over $\Sigma$.
\end{enumerate}
\end{lemma}

\begin{proof}
(1) Let $M'=(\Sigma\times\{0,1\},\{\top\},Q,\{q_0\},R')$ be the nondeterministic \ttt\ such that 
$\top$ has rank~0, and $R'$ is defined as follows. 
If $\tup{q,\sigma,j,T}\to\tup{q',\alpha}$ is a move rule in $R$,
then $\tup{q,(\sigma,b),j,\mu(T)}\to\tup{q',\alpha}$ is a rule in $R'$ for every $b\in\{0,1\}$. 
If $\tup{q,\sigma,j,T}\to\delta(\tup{q_1,\alpha_1},\dots,\tup{q_k,\alpha_k})$
is an output rule in $R$, then $R'$ contains the rules 
$\tup{q,(\sigma,0),j,\mu(T)}\to\tup{q_i,\alpha_i}$ for every $i\in[k]$
and it also contains the rule $\tup{q,(\sigma,1),j,\mu(T)}\to\top$. 
Intuitively, for an input tree $\tmark(t,u)$ with $t\in\dom(M)$, the tree-walking automaton $M'$ 
follows an arbitrary path in the unique derivation tree $d\in L(G^\mathrm{der}_{M,t})$,
from the root of $d$ down to the leaves (cf. $M'_1$ and $N$ in the proof of Lemma~\ref{lem:ttsuottl}). 
Whenever $M$ branches at an unmarked node, 
$M'$ nondeterministically follows one of those branches. It~accepts $\tmark(t,u)$ 
when an output rule is applied to the marked node $u$.
It should be clear that $T_{M,\mathrm{prod}}=\tmark^{-1}(\dom(M'))$ satisfies the requirements. 
It~is regular by Corollary~\ref{cor:regdom}.

(2) Let $M''$ be a \dttt\ that performs a depth-first left-to-right traversal 
of the input tree $t\in T_\Sigma$ and verifies that $(t,u)\in T_{M,\mathrm{prod}}$ 
for every leaf and monadic node~$u$ of $t$. Then $L_{M,\text{prod}}=\dom(M)\cap\dom(M'')$,
which is regular by Corollary~\ref{cor:regdom}.
\qed
\end{proof}

For a given deterministic \ttt\ $M$ 
there are a nondeterministic pruning \ttt\ $N$ and a deterministic \ttl\ $M'$
such that $\tau_N\circ\tau_{M'}=\tau_M$ and $\tau_M\subseteq \tau_N\circ \tau^{01}_{M'}$,
by Lemmas~\ref{lem:nondet0-prod} and~\ref{lem:nondet1-prod}. 
Our aim is to transform $N$ and $M'$ in such a way that $N$ becomes deterministic. 
We basically do this by applying Lemma~\ref{lem:cut} to $\tau_N$,
replacing it by one of its uniformizers. 
But to preserve the above two properties we first restrict the domain of $M'$ 
to productive input trees and then restrict the range of~$N$ to the new domain, as follows. 

By Lemma~\ref{lem:prodtest}, the tree language $L_{M',\text{prod}}$ is regular. 
Let $M''$ be the \dttt\ that is obtained from $M'$ by restricting its domain
to $L_{M',\text{prod}}$, see Lemma~\ref{lem:domrestr}. 
Hence, $\tau_{M''}=\tau^{01}_{M'}$ and so $\tau_N\circ\tau_{M''}=\tau_M$. 
Since $M''$ behaves in the same way as $M'$, 
every tree $t'\in\dom(M'')$ is productive (for $M''$).
Next, we change $N$ into the nondeterministic pruning \ttt\ $N'$ 
by restricting its range to $\dom(M'')$, by Corollary~\ref{cor:ranrestr}. 
Now $\tau_{N'}\circ\tau_{M''}=\tau_M$ and 
$\ran(\tau_{N'})\subseteq \dom(\tau_{M''})$. 
Finally, we define $\tau\in \TTP$ to be the uniformizer of $\tau_{N'}$
according to Lemma~\ref{lem:cut}. Then $\tau\circ\tau_{M''}=\tau_M$.
Now consider $(t,s)\in\tau_M$. Then $s=\tau_{M''}(r)$ for $r=\tau(t)$.
Since $r$ is productive for $M''$, it follows that $|r|\leq 2\cdot |s|$ 
as observed at the end of the second paragraph of Section~\ref{sub:nondetprod}. 
Hence $(\tau,\tau_{M''})$ is linear-bounded, which shows that 
$\tau_M\in \TTP \ast \TT$.

\section{Linear Size Increase}\label{sec:lsi}

In this section we show our first main result: 
the hierarchy of \ttt's collapses for functions of linear size increase.

\begin{theorem}\label{thm:mainlsi}
For every $k\geq 1$, $\TT^k \cap \LSI = \TTSU$.
\end{theorem}

\begin{proof}
The proof is by induction on $k$. For $k=1$ it is Corollary~\ref{cor:lsi}.
To prove that $\TT^{k+1} \cap \LSI \subseteq \TTSU$,
let $\tau\in \TT^k$ and let $M$ be a \dttt\ 
such that $\tau_M\circ\tau\in\LSI$. 
By Corollary~\ref{cor:prod}(2) we may assume that $(\tau_M,\tau)$ is linear-bounded.
Moreover, by restricting the domain of $M$ to $\dom(\tau_M\circ\tau)$ 
we may assume that $\ran(\tau_M)\subseteq\dom(\tau)$, 
see Lemma~\ref{lem:domrestr} and Corollary~\ref{cor:regdom}.
Hence $\tau_M\in\LSI$ by Lemma~\ref{lem:lsilsd} and so $\tau_M\in\TTSU$ by Corollary~\ref{cor:lsi}.
Then $\tau_M\circ \tau\in\TT^k$ by Theorem~\ref{thm:ttsuott}. 
Hence $\tau_M\circ \tau\in \TTSU$ by induction.
\qed
\end{proof}

\begin{theorem}\label{thm:mainlsidec}
It is decidable for a composition of deterministic \ttt's 
whether or not it is of linear size increase. 
\end{theorem}

\begin{proof}
The proof is, again, by induction on $k$, the number of \dttt's in the composition.
It goes along the lines of the proof of Theorem~\ref{thm:mainlsi}, using Corollary~\ref{cor:lsidec}
instead of Corollary~\ref{cor:lsi} for the case $k=1$. 
Assuming that we have an algorithm~$\cA_k$ for a composition of $k$ \dttt's, 
we construct $\cA_{k+1}$ as follows. 
Let $M,M_1,\dots,M_k$ be~\dttt's, $k\geq 1$, 
and let $\tau=\tau_{M_1}\circ\cdots\circ\tau_{M_k}$. 
Since all our results are effective, we may assume as in the proof of Theorem~\ref{thm:mainlsi}
that $(\tau_M,\tau)$ is linear-bounded and $\ran(\tau_M)\subseteq\dom(\tau)$.
To decide whether or not $\tau_M\circ\tau$ is of linear size increase, 
we first decide whether or not $\tau_M$ is of linear size increase by Corollary~\ref{cor:lsidec}.
If~not, then $\tau_M\circ\tau$ is not of linear size increase, by Lemma~\ref{lem:lsilsd}. 
If so, then a \dttt\ $M'_1$ that realizes $\tau_M\circ \tau_{M_1}$ can be constructed 
by Corollary~\ref{cor:lsi} and Theorem~\ref{thm:ttsuott}, and we apply $\cA_k$ to 
$M'_1,M_2,\dots,M_k$.
\qed
\end{proof}

Together with Lemma~\ref{lem:ttvsmt} and Proposition~\ref{pro:mso=ttsu} in Section~\ref{sec:dmtmso},
Theorems~\ref{thm:mainlsi} and~\ref{thm:mainlsidec}
imply the following two corollaries on macro tree transducers.

\begin{corollary}\label{cor:mtlsi}
For every $k\geq 1$, $\MT^k \cap \LSI = \MSOT = \TTSU \subseteq \MT$.
\end{corollary}

\begin{corollary}\label{cor:mtlsidec}
It is decidable for a composition of deterministic \mt's 
whether or not it is of linear size increase. 
\end{corollary}

For the class $\MTIO$ of translations realized by deterministic macro tree transducers
with inside-out (\abb{io}) derivation mode, 
we obtain that $\MT_\text{\abb{io}}^k \cap \LSI \subseteq \TTSU$ for every $k\geq 1$,
for the simple reason that $\MTIO$ is 
a (proper) subclass of $\MT$ by~\cite[Theorem~7.1(1)]{EngVog85}.
For the same reason Corollary~\ref{cor:mtlsidec} is also valid for those transducers.
However, $\TTSU$ is not included in $\MTIO$, because not every regular tree language 
is the domain of a deterministic \abb{io} macro tree transducer 
(see~\cite[Corollary~5.6]{EngVog85}). 

Since $\LSI\subseteq \cF$, it follows from Theorems~\ref{thm:mainlsi} and~\ref{thm:fun}
that Theorem~\ref{thm:mainlsi} also holds for nondeterministic \ttt's, i.e., 
$\NTT^k \cap\LSI = \TTSU$ for every $k\geq 1$.\footnote{We do not know 
whether Theorem~\ref{thm:mainlsidec} holds for nondeterministic \ttt's,
i.e., whether it is decidable for a composition of nondeterministic \ttt's
whether or not it realizes a translation in~$\LSI$.
} 
Similarly, it follows from Corollaries~\ref{cor:mtlsi} and~\ref{cor:fun} that 
Corollary~\ref{cor:mtlsi} also holds for nondeterministic \mt's, i.e., 
$\NMT^k \cap \LSI = \MSOT = \TTSU \subseteq \MT$ for every $k\geq 1$.
This even holds for the so-called stay-macro tree transducers that can use stay-instructions, 
introduced in~\cite[Section~5.3]{EngMan03}, because it is shown in~\cite[Lemma~37]{EngMan03} 
that the stay-macro tree translations are in $\NTT^4$. 
For the class $\NMT_\text{\abb{io}}$ of nondeterministic \abb{io} macro tree translations
we also obtain that $\NMT_\text{\abb{io}}^k \cap \LSI \subseteq \TTSU$ for every $k\geq 1$,
because $\NMT_\text{\abb{io}}\subseteq \NTT^2$ by Lemma~\ref{lem:nmtio}; the same is true for 
multi-return macro tree transducers. 

The \emph{$k$-pebble tree transducer} was introduced in~\cite{MilSucVia03}
as a model of XML document transformation. 
It is a \ttt\ that additionally can use $k$ distinct pebbles
to drop on, and lift from, the nodes of the input tree. 
The life times of these pebbles must be nested.
The \ttt\ is the 0-pebble tree transducer. 
It is shown in~\cite[Theorem~10]{EngMan03} that every (deterministic)
$k$-pebble tree translation can be realized by a composition of (deterministic) $k+1$ \ttt's.
Hence Theorems~\ref{thm:mainlsi} and~\ref{thm:mainlsidec} also hold for 
deterministic $k$-pebble tree transducers, while Theorem~\ref{thm:mainlsi} additionally holds
for the nondeterministic case.
In~\cite[Theorems~5 and~55]{EngHooSam} this is extended to $k$-pebble tree transducers 
that, in addition to the $k$ distinct ``visible'' pebbles, can use an arbitrary number of
``invisible'' pebbles, still with nested life times: they can be realized by 
a composition of $k+2$ \ttt's. Thus, Theorems~\ref{thm:mainlsi} and~\ref{thm:mainlsidec}
also hold for such transducers,  
cf.~\cite[Theorem~57]{EngHooSam}.\footnote{A ``visible'' pebble can be observed by the transducer 
during its entire life time (as usual for pebbles), whereas an ``invisible'' pebble $p$ cannot be observed 
during the life time of a pebble $p'$ of which the life time is nested within the one of $p$;
thus, such a pebble $p'$ ``hides'' the pebble $p$. 
}
 
The \emph{high-level tree transducer} was introduced in~\cite{EngVog88}
as a generalization of both the top-down tree transducer and the macro tree transducer.
It is proved in~\cite[Theorem~8.1(b)]{EngVog88} that nondeterministic high-level tree transducers 
can be simulated by compositions of nondeterministic \mt's.
Since every deterministic high-level tree transducer realizes a partial function 
(as should be clear from the proof of~\cite[Lemma~5.7]{EngVog88}), 
it follows from Corollary~\ref{cor:fun} that, similarly, deterministic high-level tree transducers 
can be simulated by compositions of deterministic \mt's.
Consequently, Corollaries~\ref{cor:mtlsi} and~\ref{cor:mtlsidec} 
also hold for deterministic high-level tree transducers, and 
Corollary~\ref{cor:mtlsi} additionally for the nondeterministic case.

\section{Deterministic Complexity}\label{sec:complex}

Our first main complexity result says that a composition of deterministic \ttt's can be computed
by a RAM program in linear time, more precisely in time $O(n)$ where $n$ is 
the sum of the sizes of the input and the output tree. 

\begin{theorem}\label{thm:lintime}
For every $k\geq 1$ and every $\tau\in \TT^k$ there is an algorithm that computes, 
given an input $t$, the output $s=\tau(t)$ in time $O(|t|+|s|)$.
\end{theorem}

\begin{proof}
The proof is by induction on $k$. 
We first prove the case $k=1$, which is a slight generalization 
of the well-known fact for attribute grammars that the attribute evaluation 
of an input tree takes linear time (see, e.g., \cite{DJL88,Eng84}).
Let~$\tau\in\TT$ and let $t$ be an input tree of~$\tau$. 
By Corollary~\ref{cor:regdom}, $\dom(\tau)$ is regular and 
hence can be recognized by a bottom-up finite-state tree automaton. Thus, 
we can decide whether or not $t\in\dom(\tau)$ in time $O(|t|)$
by running that automaton on $t$. 
By Lemmas~\ref{lem:msorel} and~\ref{lem:d=ds}, 
$\tau=\tau_1\circ \tau_2$ with $\tau_1\in\TTRS$ and $\tau_2\in\TTL$. 
As observed in Section~\ref{sec:trans}, $\tau_1$ can be realized by a classical 
linear deterministic top-down tree transducer with regular look-ahead. 
Thus, by (the proof of)~\cite[Theorem~2.6]{Eng77}, it can be realized by 
a deterministic bottom-up finite-state relabeling (DBQREL) and a local relabeling \ttt.
To run these two relabelings on $t\in\dom(\tau)$ obviously takes time $O(|t|)$.
Thus, it remains to consider the case that $\tau\in\TTL$.
Let~$M$ be a local \dttt\ that realizes~$\tau$. 
To compute $\tau_M(t)$, we first construct the regular tree grammar $G_{M,t}$
in time $O(|t|)$, the number of configurations of~$M$ on~$t$. 
Then we remove the chain rules from the context-free grammar $G_{M,t}$, 
i.e., the rules $\tup{q,u}\to\tup{q',u'}$ resulting from the move rules of~$M$. 
Since $G_{M,t}$ is forward deterministic, this can also be done in time $O(|t|)$,
as follows. Viewing the chain rules as edges of a directed graph 
with configurations as nodes, we compute an evaluation order of the graph 
by topological sorting, in time $O(|t|)$. Then we compute the new rules 
by traversing this order from right to left, again in time $O(|t|)$. 
For an edge $\tup{q,u}\to\tup{q',u'}$,
if the (old or new) rule for $\tup{q',u'}$ is 
$\tup{q',u'}\to \delta(\tup{q_1,u_1},\dots,\tup{q_k,u_k})$, then the new rule for
$\tup{q,u}$ is $\tup{q,u}\to \delta(\tup{q_1,u_1},\dots,\tup{q_k,u_k})$.
Finally, we use this new regular tree grammar, equivalent to $G_{M,t}$, to generate $s=\tau_M(t)$,
which takes time $O(|s|)$ because each rule generates a node of $s$. 

Now let $\tau=\tau_1\circ\tau_2$ such that $\tau_1\in\TT$ and $\tau_2\in \TT^k$, $k\geq 1$. 
By Corollary~\ref{cor:prod}(2) we may assume that $(\tau_1,\tau_2)$ is linear-bounded.
Let $t$ be an input tree of $\tau$. 
Since $\dom(\tau)$ is regular by Corollary~\ref{cor:regdom}, 
we can check that $t\in\dom(\tau)$ in linear time, as above. 
By the case $k=1$, the intermediate tree $r=\tau_1(t)$ can be computed in time $O(|t|+|r|)$,
and by induction the output tree $s=\tau(t)=\tau_2(r)$ 
can be computed in time $O(|r|+|s|)$.
Since $(\tau_1,\tau_2)$ is linear-bounded, there is a constant $c\in\nat$ such that 
$|r|\leq c\cdot|s|$, i.e., $|r|=O(|s|)$. Hence the total time is 
$O(|t|+|r|)+O(|r|+|s|)=O(|t|+|s|)$.
\qed
\end{proof}

It should be noted that the constant in the time complexity $O(|t|+|s|)$ can be large 
in terms of the size of the given transducers due to the use of linear-boundedness, 
cf. Remark~\ref{rem:size}.

Since deterministic macro tree transducers, pebble tree transducers, 
and high-level tree transducers can be realized as compositions of deterministic \ttt's
(see Section~\ref{sec:lsi}),
Theorem~\ref{thm:lintime} also holds for such transducers. 
For $k$-pebble tree transducers this improves the result of~\cite[Proposition~3.5]{MilSucVia03},
where the time bound is $O(|t|^k+|s|)$.

Before we proceed, we need an elementary lemma on leftmost derivations of context-free grammars.
For a context-free grammar $G=(N,T,\cS,R)$, a leftmost sentential form is a string $v\in (N\cup T)^*$ 
such that $S\Rightarrow^*_{G,\mathrm{lm}} v$ for some $S\in\cS$, 
where $\Rightarrow_{G,\mathrm{lm}}$ is the usual 
leftmost derivation relation of $G$: if $X\to\zeta$ is in $R$, 
then $v_1Xv_2\Rightarrow_{G,\mathrm{lm}} v_1\zeta v_2$ for all $v_1\in T^*$ and $v_2\in (N\cup T)^*$. 

\begin{lemma}\label{lem:leftmost}
Let $G=(N,T,\cS,R)$ be an $\epsilon$-free context-free grammar,
and let $G'=(N',T,\cS,R')$ be the equivalent context-free grammar such that $N'=N\cup\{Z\}$ and
$R'=\{X\to \zeta Z\mid X\to\zeta \in R\}\cup\{Z\to\epsilon\}$, where $Z$ is a new nonterminal.
Let $v$ be a leftmost sentential form of $G'$, and let 
$S\Rightarrow^*_{G',\mathrm{lm}} v \Rightarrow^*_{G',\mathrm{lm}} w$ 
be a leftmost derivation of $G'$ with $S\in\cS$ and $w\in L(G)$. 
Moreover, let $d$ be the derivation tree corresponding to that derivation. 
Then the number of occurrences of $Z$ in $v$ is at most the height of~$d$.\footnote{Note that
there is a straightforward one-to-one correspondence between the leftmost derivations of $G$ and $G'$,
and between their derivation trees. Since $G$ is $\epsilon$-free, 
the derivation trees have the same height.
\label{ftn:leftmost}}
\end{lemma}

\begin{proof}
Each occurrence of a nonterminal $Y\in N'$ in $v$ 
corresponds to a node of $d$ with label $Y$ in a well-known way.
Let $u$ be the node of $d$ corresponding to the leftmost occurrence of $Z$ in $v$.
Clearly the number of occurrences of $Z$ in $v$ is equal to 
the number of edges on the path from $u$ to the root of $d$.
\qed
\end{proof}

By~\cite[Theorem~2.5]{Pap94} it follows from Theorem~\ref{thm:lintime} that
a composition of deterministic \ttt's can be computed
by a deterministic Turing machine in cubic time, more precisely in time $O(n^3)$ 
where $n$ is the sum of the sizes of the input and the output tree. 
Our second complexity result says that a composition of deterministic \ttt's can be computed 
by a deterministic multi-tape Turing machine $N$ in linear space
(in the sum of the sizes of the input and output tree). 
On a work tape of $N$ we will represent the input tree $t$ over $\Sigma$ 
by the string $\phi(t)$ over $\Sigma\cup\{(,)\}$, 
where $\{(,)\}$ is the set consisting of the left- and right-parenthesis,
defined such that if $\phi(t_1)=t'_1,\dots,\phi(t_m)=t'_m$ then 
$\phi(\sigma t_1\cdots t_m)=\sigma(t'_1\cdots t'_m)$.
In other words, we formally insert the parentheses (but not the commas) 
that are always used informally to denote trees. 
The parentheses allow $N$ to walk on the tree $t$, from node to node, because 
it can recognize a subtree of $t$ by checking that 
the numbers of left- and right-parentheses 
in the corresponding substring of $\phi(t)$ are equal.
In particular, it can determine the child number of a node of~$t$ by 
counting the number of its younger siblings. 
Obviously, the mapping $\phi$ is injective, and can be computed in linear space
(simulating a one-way push-down transducer). 
In what follows we identify $t$ and $\phi(t)$. 

\begin{theorem}\label{thm:linspace}
For every $k\geq 1$ and every $\tau\in \TT^k$ there is a deterministic Turing machine that computes, 
given an input $t$, the output $s=\tau(t)$ in space $O(|t|+|s|)$.
\end{theorem}

\begin{proof}
Again, we first show this for $k=1$. 
Let $M = (\Sigma, \Delta, Q, q_0, R)$ be a~\dttt, and 
let $t\in T_\Sigma$ be an input tree. 
As usual we assume that the output rules of $M$ only contain stay-instructions. 
We describe a deterministic multi-tape Turing machine~$N$ that computes $\tau_M$ in linear space. 
By Corollary~\ref{cor:regdom}, $\dom(M)$ is a regular tree language
and hence a context-free language, which can be recognized in deterministic linear space. 
Thus, $N$ starts by deciding whether or not $t\in\dom(M)$. Now assume that $t\in\dom(M)$. 
To compute $s=\tau_M(t)$, the machine $N$ simulates the (unique) leftmost derivation 
of the forward deterministic context-free grammar $G_{M,t}$. 
Every leftmost sentential form of $G_{M,t}$ is of the form $w\tup{q_1,u_1}\cdots\tup{q_n,u_n}$
with $w\in \Delta^*$ and $\tup{q_i,u_i}\in \con(t)$. If one views the states of $M$ as 
recursive procedures with one parameter of type `node of $t$', then $\tup{q_1,u_1}\cdots\tup{q_n,u_n}$ 
corresponds to the contents of the stack in the usual implementation of recursive procedures: 
each configuration $\tup{q_i,u_i}$ is a call of procedure $q_i$ with actual parameter $u_i$. 
The machine $N$ uses a one-way output tape on which it prints $w$ (which will finally be $s$), 
a work tape with the input tree $t$ (or rather $\phi(t)$),
and a work tape that contains a stack representing $\tup{q_1,u_1}\cdots\tup{q_n,u_n}$, 
with the top of the stack to the left. 
At each moment of time, a reading head of $N$ is at node $u_1$ of $t$, 
and another reading head is at the top of the stack. 
Note that $n\leq |s|$ because every configuration $\tup{q_i,u_i}$ will generate at least one symbol of $s$. 
If $N$ would represent the parameters $u_2,\dots,u_n$ by their Dewey notation, the size of the stack could be 
$|s|\cdot|t|$, which is too much. Thus, we need a more compact representation of the nodes $u_2,\dots,u_n$. 
In a rule of~$G_{M,t}$ with left-hand side $\tup{q,u}$, 
every node $u'$ in the right-hand side is a neighbour of~$u$, or $u$ itself, and so, the ``difference'' 
between $u$ and $u'$ can be expressed by an instruction 
in $I=\{\up,\stay\}\cup\{\down_i\mid i\in[1,\m_\Sigma]\}$.
This allows us to represent $\tup{q_1,u_1}\cdots\tup{q_n,u_n}$ by the node $u_1$
and a stack of the form $q_1\gamma_1q_2\gamma_2\cdots q_n\gamma_n$ where $\gamma_i\in I^*$ is a sequence of 
instructions that leads from $u_i$ to $u_{i+1}$ (with $u_{n+1}=\rt_t$). 
Let us now consider in detail how $N$ simulates the leftmost derivation of $G_{M,t}$. 

At each moment of time, the current node of $t$ and the current contents of the output tape 
and the stack tape represent a leftmost sentential form of $G_{M,t}$, which is an element of 
$\Delta^*\cdot\con(t)^*$. The stack tape contains a string in $(Q\cup I)^*\bot$, 
where $\bot$ is the bottom stack symbol and $I$ is as above.
The current node $u$ of~$t$ and the current contents $w\in\Delta^*$ and $\xi\in(Q\cup I)^*\bot$ of the 
output tape and stack tape, respectively, represent the leftmost sentential form $w\cdot\mu(u,\xi)$,
where the string $\mu(u,\xi)\in\con(t)^*$ is defined as follows (for every $q\in Q$ and $\beta\in I$):
$\mu(u,q\xi)=\tup{q,u}\cdot \mu(u,\xi)$, $\mu(u,\beta\xi)=\mu(\beta(u),\xi)$, 
and $\mu(\bot)=\epsilon$. 
Initially, $N$~starts at the root of $t$, with empty output tape and with stack tape $q_0\bot$, 
representing the initial output form $\tup{q_0,\rt_t}$. 
If the top symbol of the stack is $\bot$, then $N$~halts. 
Otherwise, to compute the next leftmost sentential form, $N$ first pops the top symbol off the stack.
If that symbol was $q\in Q$, and the current node $u$ of $t$ has label $\sigma$ and child number $j$, 
then $N$ selects the unique rule $\tup{q,\sigma,j,T}\to \zeta$ that is applicable to $\tup{q,u}$.
Note that it can test in linear space whether or not $(t,u)\in T$, because $\tmark(T)$ is a 
context-free language. If $\zeta=\tup{q',\alpha}$, then $N$ moves to node $\alpha(u)$ of $t$
and pushes the string $q'\beta$ on the stack where $\beta$ is defined as follows: 
if $\alpha$ is $\up$, $\stay$, or $\down_i$, then $\beta$ is $\down_j$, $\stay$, or $\up$, respectively. 
If $\zeta=\delta(\tup{q_1,\stay},\dots,\tup{q_k,\stay})$, 
then $N$ outputs $\delta$, and pushes $q_1\cdots q_k$ on the stack (if $k>0$). 
It is easy to check that in both these cases the resulting configuration of $N$ 
represents the next leftmost sentential form of~$G_{M,t}$.
If the top symbol of the stack was $\beta\in I$, the machine~$N$ 
moves to node $\beta(u)$ of $t$. This does not change the represented leftmost sentential form.
Thus, after applying a rule $\tup{q,\sigma,j,T}\to \delta$ (with $\delta$ of rank~$0$), 
$N$ removes instructions from the stack (and moves its reading head on~$t$ accordingly) 
until the top of the stack is a state again.
When $N$ halts, the output tape contains~$s$.

It remains to show that the length of the stack is linear in $|t|+|s|$.
As mentioned above, since every configuration $\tup{q,u}$ will generate at least one symbol of $s$, 
the number of occurrences of states in the stack is at most $|s|$. 
To estimate the number of occurrences of instructions in the stack,
we use Lemma~\ref{lem:leftmost}.
In the above case where $q$ is the top stack symbol and $\tup{q,\sigma,j,T}\to \tup{q',\alpha}$
is the rule applicable to~$\tup{q,u}$, the machine~$N$ does not apply the rule
$\tup{q,u}\to\tup{q',\alpha(u)}$ of $G_{M,t}$, but rather the rule $\tup{q,u}\to\tup{q',\alpha(u)}\beta$
where $\beta$ is defined as above. Moreover, when $\beta$~is the top stack symbol, 
$N$ applies the rule $\beta\to\epsilon$. From this it should be clear that,
by Lemma~\ref{lem:leftmost} and footnote~\ref{ftn:leftmost}, 
the number of occurrences of instructions in the stack 
is at most the height of the derivation tree corresponding to the derivation 
$\tup{q_0,\rt_t}\Rightarrow^*_{M,t} s$ of $G_{M,t}$. 
As observed in Section~\ref{sec:trees} after Lemma~\ref{lem:finpump}, 
that height is at most $\#(\con(t))$, i.e., $\#(Q)\cdot |t|$. Thus, the length of the stack 
is indeed $O(|s|+|t|)$. 

The induction step can be proved in exactly the same way as in the proof of Theorem~\ref{thm:lintime},
with `time' replaced by `space'. 
\qed
\end{proof}

For a class $\cT$ of tree translations and a class $\cL$ of tree languages, 
we denote by $\cT(\cL)$ the class of tree languages
$\tau(L)$ with $\tau\in\cT$ and $L\in\cL$. 
The elements of $\cT(\REGT)$ are called the output tree languages 
(or surface languages) of $\cT$. Since $\TT\subseteq\MT$ by Lemma~\ref{lem:ttvsmt},
it follows from the proof of~\cite[Theorem~7.5]{EngVog85} that the output tree languages 
of $\TT^k$ are recursive.
From Theorem~\ref{thm:linspace} we now obtain that they are in $\family{DSPACE}(n)$, 
i.e., can be recognized by a Turing machine in deterministic linear space.
This was shown for classical top-down tree transducers in~\cite{Bak78}.

\begin{theorem}\label{thm:outspace}
For every $k\geq 1$, $\TT^k(\REGT)\subseteq \family{DSPACE}(n)$.
\end{theorem}

\begin{proof}
Let $L\in\REGT$ and $\tau\in\TT^k$.
By Corollary~\ref{cor:prod}(2), $\tau=\tau_1\circ\tau_2$ 
such that $\tau_1\in\TTP$, $\tau_2\in \TT^k$, 
and $(\tau_1,\tau_2)$ is linear-bounded for some constant~$c$.
Let $L'=\tau_1(L)$, and note that $\tau(L)=\tau_2(L')$ and that 
$L'\in\REGT$ by Lemma~\ref{lem:pru}.
It is straightforward to show that for every $s\in \tau(L)$ there exists $t\in L'$ such that
$(t,s)\in\tau_2$ and $|t|\leq c\cdot|s|$. 
To check whether a given tree~$s$ is in $\tau(L)$, a deterministic Turing machine 
systematically enumerates all input trees $t$ (of $\tau_2$) such that $|t|\leq c\cdot|s|$.
For each such $t$ it first checks that $t\in L'$ in space $O(|t|)$.
Then it uses the algorithm of Theorem~\ref{thm:linspace} to compute 
$\tau_2(t)$ in space $c'\cdot(|t|+|\tau_2(t)|)$, but rejects $t$ as soon as 
the computation takes more than space $c'\cdot(|t|+|s|)$;
thus, the space used is $O(|t|+|s|)=O(|s|)$.
Clearly, $s\in\tau(L)$ if and only if $\tau_2(t)=s$ for some such $t$. 
\qed
\end{proof}

For a tree $t$ we denote its yield by $yt$, 
for a tree language $L$ we define $yL=\{yt\mid t\in L\}$,
and for a class $\cL$ of tree languages we define $y\cL=\{yL\mid L\in\cL\}$.
For a class $\cT$ of tree translations, the languages in $y\cT(\REGT)$
are called the output string languages (or target languages) of $\cT$. 

\begin{corollary}\label{cor:outspace}
For every $k\geq 1$, $y\TT^k(\REGT)\subseteq \family{DSPACE}(n)$.
\end{corollary}

\begin{proof}
For an alphabet $\Delta$, let $\Gamma=\Delta\cup\{e\}$ be the ranked alphabet such that
$e$~has rank~0 and every element of $\Delta$ has rank~1. 
For a string $w$ over $\Delta$ we define $\mon(w)=we\in T_\Gamma$.  
It is easy to see that for every ranked alphabet $\Sigma$ there is a \dttl\ $M$ such that
$\tau_M(t)=\mon(yt)$. From this and Theorem~\ref{thm:outspace} the result follows.
\qed
\end{proof}

We observe here, for $k=1$, that $\TT(\REGT)$ and $y\TT(\REGT)$ 
are included in $\LOGCF$, the class of languages that are 
log-space reducible to a context-free language. This will be proved in Corollaries~\ref{cor:logcf}
and~\ref{cor:ylogcf}.
Note that $\LOGCF\subseteq \family{DSPACE}(\log^2 n)$.

We also observe that Theorem~\ref{thm:outspace} and Corollary~\ref{cor:outspace}
also hold for nondeterministic \ttt's, as will be proved in Theorem~\ref{thm:outdspace}
(and was proved for classical top-down tree transducers in~\cite{Bak78}).

As before, Theorems~\ref{thm:linspace} and~\ref{thm:outspace} and Corollary~\ref{cor:outspace} 
also hold for deterministic macro tree transducers, pebble tree transducers, 
and high-level tree transducers.
It is proved in~\cite[Theorem~23]{EngMan02} that composition of deterministic \mt's yields a proper hierarchy 
of output string languages (called the $y\MT$-hierarchy), i.e., that
$y\MT^k(\REGT)\subsetneq y\MT^{k+1}(\REGT)$ for every $k\geq 1$.
The \emph{\abb{IO}-hierarchy} consists of the classes of string languages $\family{IO}(k)$ 
generated by level-$k$ grammars, with the inside-out (\abb{io}) derivation mode
(see, e.g.,~\cite{Dam82}). By~\cite[Theorem~7.5]{EngSch78} the \abb{IO}-hierarchy can be defined 
as output string languages of tree transformations: $\family{IO}(k)=y\YIELD^k(\REGT)$. Since 
$\YIELD\subseteq \TT$ by~\cite[Lemma~36]{EngMan03},
we obtain that $\family{IO}(k)\subseteq y\TT^k(\REGT)$. 
Thus, the next corollary is immediate from Corollary~\ref{cor:outspace}.
Note that it was already proved in~\cite[Theorem~3.3.8]{Fis68} that the \abb{io} languages 
(i.e., the languages in $\family{IO}(1)$) are in $\family{NSPACE}(n)$;
in~\cite{Asv81} this was improved to $\LOGCF$. 
It was proved in~\cite[Corollary~8.12]{Dam82} that the languages in the \abb{IO}-hierarchy
are recursive. 

\begin{corollary}\label{cor:iok}
For every $k\geq 1$, $\family{IO}(k)\subseteq \family{DSPACE}(n)$.
\end{corollary}

Note that by~\cite[Theorem~36]{EngMan02} the $\family{EDTOL}$ control hierarchy 
is included in the \abb{IO}-hierarchy.

By Corollary~\ref{cor:ttmttt}, 
$y\TT^k(\REGT)\subseteq y\MT^k(\REGT)\subseteq y\TT^{k+1}(\REGT)$.
It is proved in~\cite[Theorem~32]{EngMan02} that there exists a language in $\family{IO}(k+1)$
that is not in $y\MT^k(\REGT)$. Since $\family{IO}(k+1)\subseteq y\TT^{k+1}(\REGT)$,
that implies the following stronger version of Proposition~\ref{pro:hier}.

\begin{corollary}
For every $k\geq 1$, $y\TT^k(\REGT)\subsetneq y\TT^{k+1}(\REGT)$.
\end{corollary}

\section{Nondeterministic Complexity}\label{sec:nondetcomplex}

We now turn to the complexity of compositions of nondeterministic \ttt's.
We first consider the case where all the transducers in the composition are finitary.
The next lemma shows that Theorem~\ref{thm:prod} and Corollary~\ref{cor:prod} 
also hold for $\FTT$. 

\begin{lemma}\label{lem:finprod}
$\FTT^k\subseteq \NTTP \ast \FTT^k$ and $\FTT\circ\FTT^k = \FTT\ast\FTT^k$ for every $k\geq 1$.
\end{lemma}

\begin{proof}
To show that $\FTT\subseteq \NTTP \ast \FTT$, let $\tau\in\FTT$. 
By Theorem~\ref{thm:prod}, $\tau=\tau_1\circ\tau_2$ 
such that $\tau_1\in\NTTP$, $\tau_2\in\NTT$, and $(\tau_1,\tau_2)$ is linear-bounded. 
Since $\ran(\tau_1)\in\REGT$ by Lemma~\ref{lem:pru}, 
we may assume that $\dom(\tau_2)\subseteq\ran(\tau_1)$ by Lemma~\ref{lem:domrestr}.
Then $\tau_2$ is finitary too. 

Theorem~\ref{thm:nondetrightcomp} implies that $\FTT \circ \NTTP \subseteq \FTT$, 
because the composition of two finitary translations is finitary. 
The remainder of the proof is now entirely similar to the one of Corollary~\ref{cor:prod}.
\qed
\end{proof}

We will prove that a composition of \ttt's can be computed 
by a nondeterministic Turing machine in linear space and polynomial time 
(in the sum of the sizes of the input and output tree),
which generalizes Theorem~\ref{thm:linspace}. In the next lemma we consider the case where 
all \ttt's are finitary. 

\begin{lemma}\label{lem:finlinspace}
For every $k\geq 1$ and every $\tau\in \FTT^k$ there is a nondeterministic Turing machine that computes, 
given an input $t$, any output $s\in\tau(t)$ in space $O(|t|+|s|)$ 
and in time polynomial in $|t|+|s|$.
\end{lemma}

\begin{proof}
For the case $k=1$ the proof is exactly the same as that of Theorem~\ref{thm:linspace}
except, of course, that the Turing machine $N$ nondeterministically simulates any leftmost derivation 
of $G_{M,t}$, selecting nondeterministically a rule of $M$ to compute a next leftmost sentential form.
It follows from Lemmas~\ref{lem:leftmost} and~\ref{lem:finpump} that the number~$n$ of occurrences 
of instruction symbols in the stack is $O(|t|)$. 
In fact, since $M$ is finitary, it suffices by Lemma~\ref{lem:finpump} to simulate 
leftmost derivations of~$G_{M,t}$ for which the corresponding
derivation tree in $L(G^\mathrm{der}_{M,t})$ has height at most $\#(Q)\cdot|t|$. 
As in the proof of Theorem~\ref{thm:linspace}, Lemma~\ref{lem:leftmost} implies that 
$n$ is at most that height, i.e., at most $\#(Q)\cdot|t|$. Thus, $N$ works in space $O(|t|+|s|)$.
Moreover, it works in time $O(|t|^2\cdot|s|)$, because
the size of such a derivation tree (and hence the length of the leftmost derivation) 
is at most $\#(Q)\cdot|t|\cdot|s|$, and each step in the leftmost derivation takes time $O(|t|)$.
Note that regular tree languages (which are context-free languages) 
can be recognized in nondeterministic linear time.

Now let $\tau=\tau_1\circ\tau_2$ such that $\tau_1\in\FTT$ and $\tau_2\in \FTT^k$, $k\geq 1$.
We may assume by Lemma~\ref{lem:finprod} that $(\tau_1,\tau_2)$ is linear-bounded.
So, there is a constant $c\in\nat$ such that for every $(t,s)\in \tau$ there exists a tree $r$ 
such that $(t,r)\in\tau_1$, $(r,s)\in\tau_2$, and $|r|\leq c\cdot|s|$. By the case $k=1$, 
the intermediate tree $r$ can be computed from $t$ in nondeterministic space $O(|t|+|r|)$,
and by induction, the output tree~$s$ can be computed from $r$ in nondeterministic space $O(|r|+|s|)$.
Hence, since $|r|=O(|s|)$, $s$ can be computed from $t$ in nondeterministic space $O(|t|+|s|)$.
The time is polynomial in $|t|+|r|$ and $|r|+|s|$, and hence polynomial in $|t|+|s|$.
\qed
\end{proof}

By Lemma~\ref{lem:mtintt}, $\NMT\subseteq \FTT^2$. Consequently Lemma~\ref{lem:finlinspace}
also holds for every $\tau\in\NMT^k$.

We now turn to the output languages of $\FTT^k$. 
By $\family{NSPACE}(n)\wedge \NP$ we will denote the class of languages that 
can be recognized by a nondeterministic Turing machine in simultaneous linear space and polynomial time. 
Trivially, $\family{NSPACE}(n)\wedge \NP$ is included in both $\family{NSPACE}(n)$ and 
$\NP$.

\begin{lemma}\label{lem:finoutspace}
For every $k\geq 1$, $\FTT^k(\REGT)\subseteq \family{NSPACE}(n)\wedge \NP$.
\end{lemma}

\begin{proof}
The proof is similar to the one of Theorem~\ref{thm:outspace}.
Let $L\in\REGT$ and $\tau\in\FTT^k$.
By Lemma~\ref{lem:finprod}, $\tau=\tau_1\circ\tau_2$ 
where \mbox{$\tau_1\in\NTTP$}, $\tau_2\in \FTT^k$, 
and $(\tau_1,\tau_2)$ is linear-bounded for some constant~$c$.
Let $L'=\tau_1(L)$. Then $s\in \tau(L)$ if and only if there exists $t\in L'$ such that
$(t,s)\in\tau_2$ and $|t|\leq c\cdot|s|$. 
To check whether a given tree~$s$ is in $\tau(L)$, a nondeterministic Turing machine 
guesses an input tree $t$ such that $|t|\leq c\cdot|s|$, it
checks that $t\in L'$ in time and space $O(|t|)$ (because $L'$ is a context-free language), 
and then computes any $s'\in\tau(t)$ with $|s'|\leq |s|$ in space $O(|t|+|s'|)$
and time polynomial in $|t|+|s'|$, by Lemma~\ref{lem:finlinspace}.
Finally it checks that $s'=s$ in time and space $O(|s|)$. 
Thus the space used is $O(|t|+|s|)=O(|s|)$, and the time is polynomial in $|s|$.
\qed
\end{proof}

Although \mt's are finitary, whereas \ttt's need not be finitary, 
it is proved in~\cite[Theorem~38 and Corollary~39]{EngMan03} that 
compositions of \mt's have the same output languages as compositions of (local) \ttt's.
This implies that Lemmas~\ref{lem:finoutspace} and~\ref{lem:finlinspace}
also hold for $\NTT^k$.

\begin{theorem}\label{thm:nondetoutspace}
For every $k\geq 1$, $\NTT^k(\REGT)\subseteq \family{NSPACE}(n)\wedge \NP$, and moreover,
$\NMT^k(\REGT)\subseteq \family{NSPACE}(n)\wedge \NP$.
\end{theorem}

\begin{proof}
By Lemma~\ref{lem:mtintt}, $\NMT\subseteq \FTT^2$. 
Thus, by Lemma~\ref{lem:finoutspace},
$\NMT^k(\REGT)\subseteq \family{NSPACE}(n)\wedge \NP$.
From Lemma~\ref{lem:msorel} and Theorem~\ref{thm:nondetrightcomp} it follows
(by induction on $k$) that $\NTT^k\subseteq \TTR\circ (\NTTL)^k$
and hence $\NTT^k(\REGT)\subseteq (\NTTL)^k(\REGT)$ by Lemma~\ref{lem:pru}.
Finally, by~\cite[Theorem~38 and Corollary~39]{EngMan03}, 
$(\NTTL)^k(\REGT)\subseteq \NMT^m(\REGT)$ for some $m\geq 1$. 
Hence $\NTT^k(\REGT)\subseteq \family{NSPACE}(n)\wedge \NP$, by the above.
\qed
\end{proof}

As observed already after Corollary~\ref{cor:outspace},
the space part of Theorem~\ref{thm:nondetoutspace} will be strengthened to $\family{DSPACE}(n)$ 
in Theorem~\ref{thm:outdspace}.

\begin{theorem}\label{thm:nondetlinspace}
For every $k\geq 1$ and every $\tau\in \NTT^k$ there is a nondeterministic Turing machine that computes, 
given an input $t$, any output $s\in\tau(t)$ in space $O(|t|+|s|)$
and in time polynomial in $|t|+|s|$. The same holds for $\tau\in \NMT^k$.
\end{theorem}

\begin{proof}
For $\tau\in \NMT^k$ this was already observed after Lemma~\ref{lem:finlinspace}.
Now let $\tau\in \NTT^k$ with input alphabet $\Sigma$.
Let $\bar{\Sigma}=\{\bar{\sigma}\mid \sigma\in\Sigma\}$ with 
$\rank(\bar{\sigma})=\rank(\sigma)$  be a set of new symbols,
and let $\bar{t}\in T_{\bar{\Sigma}}$ be obtained from $t\in T_\Sigma$
by changing each label $\sigma$ into~$\bar{\sigma}$. 
Finally, let $\#$ be a new symbol of rank~2. It is easy to show that the tree language 
$L_\tau=\{\#(\bar{t},s)\mid (t,s)\in\tau\}$ is in $\NTT^k(\REGT)$: the first transducer 
additionally copies the input to the output (with bars), and each other transducer copies 
the first subtree of the input to the output. 
By Theorem~\ref{thm:nondetoutspace}, there is 
a nondeterministic Turing machine $N$ that recognizes $L_\tau$ in linear space and polynomial time. 
We construct the nondeterministic Turing machine $N'$ that, on input $t$, 
guesses a possible output tree $s$, writing $\#(\bar{t},s)$ on a worktape, 
uses $N$ as a subroutine to verify that $(t,s)\in\tau$, and outputs $s$.
Clearly, $N'$ satisfies the requirements. 
\qed
\end{proof}

Since \abb{IO} (multi-return) macro tree translations, pebble tree translations, and 
high-level tree translations can be 
realized by compositions of \ttt's (see Section~\ref{sec:lsi}),
Theorems~\ref{thm:nondetoutspace} and~\ref{thm:nondetlinspace} also hold for 
those translations.

By the proof of Corollary~\ref{cor:outspace}, we additionally obtain from Theorem~\ref{thm:nondetoutspace}
that $y\NTT^k(\REGT)\subseteq \family{NSPACE}(n)\wedge \NP$ for every $k\geq 1$,
and the same is true for $y\NMT^k(\REGT)$. 
The \emph{\abb{OI}-hierarchy} consists of the classes of string languages $\family{OI}(k)$ 
generated by level-$k$ grammars, with the outside-in (\abb{oi}) derivation mode
(see, e.g.,~\cite{Dam82,EngSch78}). 
It was shown in~\cite[Theorem~4.2.8]{Fis68} that $\family{OI}(1)$ equals 
the class of indexed languages of~\cite{Aho68}, and hence that 
$\family{OI}(1)\subseteq \family{NSPACE}(n)$ by~\cite[Theorem~5.1]{Aho68}.
Moreover, it was shown in~\cite[Proposition~2]{Rou73} that $\family{OI}(1)\subseteq \NP$.
In~\cite[Corollary~7.26]{Dam82} it was proved that the languages in the \abb{OI}-hierarchy
are recursive. 
As observed in the last paragraph of~\cite{EngVog88},
$\family{OI}(k)$ is included in $y\NMT^m(\REGT)$ for some $m$. 

\begin{corollary}\label{cor:oik}
For every $k\geq 1$, $\family{OI}(k)\subseteq \family{NSPACE}(n)\wedge \NP$.
\end{corollary}

It is shown in~\cite{Rou73} that there is an NP-complete language 
in both $\family{OI}(1)$ and $y\FTTD(\REGT)$, and it is shown in~\cite{vLe75} 
that there even is one in the class $\family{ETOL}$, which is a subclass of both
$\family{OI}(1)$ and $y\FTTD(\REGT)$. 
Note that by~\cite[Theorem~14]{Vog88} the $\family{ETOL}$ control hierarchy 
is included in the \abb{OI}-hierarchy.

It will be shown in Corollary~\ref{cor:oikdet} that $\family{OI}(k)\subseteq \family{DSPACE}(n)$.

\section{Translation Complexity}\label{sec:transcomp}

In this section we study the time and space complexity of the 
membership problem of the tree translations in $\NTT^k$,
i.e., for a fixed tree translation $\tau\subseteq T_\Sigma\times T_\Delta$
we want to know, for given trees $t\in T_\Sigma$ and $s\in T_\Delta$, 
how hard it is to decide whether or not $(t,s)\in\tau$.
To formalize this, 
we denote by $L_\tau$ the string language $\{\#ts\mid (t,s)\in\tau\}$,
where $\#$ is a new symbol. For simplicity, and without loss of generality, 
we assume that $\Sigma\cap\Delta=\nothing$. Otherwise, we replace $\Sigma$
by $\bar{\Sigma}=\{\bar{\sigma}\mid \sigma\in\Sigma\}$ 
as in the proof of Theorem~\ref{thm:nondetlinspace}. 
So, $L_\tau$ is a tree language over $\Sigma\cup\Delta\cup\{\#\}$, 
where $\#$ has rank~2. 
For a class $\cT$ of tree translations and a complexity class $\cC$, 
we will write $\cT\subseteq \cC$ to mean that $L_\tau\in\cC$ for every $\tau\in\cT$. 
As usual, we denote the class of languages that are accepted by a deterministic Turing machine 
in polynomial time by $\PTIME$, and the class of languages that are 
log-space reducible to a context-free language by $\LOGCF$. Note that every regular tree language
is a context-free language and hence is in $\LOGCF$. 
Note also that $\LOGCF\subseteq \PTIME$ (see~\cite{Ruz}) and 
$\LOGCF\subseteq \family{DSPACE}(\log^2 n)$ (see~\cite{LSH,Ruz} and~\cite[Theorem~12.7.4]{Har}).

If $\tau\in \TT^k$ then, on input $\#(t,s)$, we can compute $\tau(t)$ 
according to Theorems~\ref{thm:lintime} and~\ref{thm:linspace}
(rejecting the input when the computation takes more 
than time or space $c\cdot(|t|+|s|)$ for the given constant $c$) 
and then verify that $\tau(t)=s$, cf. the proof of Theorem~\ref{thm:outspace}.
Thus, $L_\tau$ can be accepted
by a RAM program in linear time and by a deterministic Turing machine in linear space.
This means that $\TT^k\subseteq \PTIME$ and $\TT^k\subseteq \family{DSPACE}(n)$.
If $\tau\in \NTT^k$, then, as mentioned in the proof of Theorem~\ref{thm:nondetlinspace},
the tree language $L_\tau$ is in the class of output languages $\NTT^k(\REGT)$, 
and hence in $\family{NSPACE}(n)\wedge \NP$ by Theorem~\ref{thm:nondetoutspace}.
This means that $\NTT^k\subseteq \NP$ and $\NTT^k\subseteq \family{NSPACE}(n)$.
Due to the presence of both the input tree and the output tree in $L_\tau$,
one would expect that better upper bounds can be shown. Indeed, we will prove that 
$\NTT^k\subseteq \family{DSPACE}(n)$.

Our main aim in this section is to prove that $\NTT\circ \TT \subseteq \LOGCF$. 
We follow the approach of~\cite{EngCompAG}, using multi-head automata.

A \emph{multi-head tree-walking tree transducer} $M=(\Sigma,\Delta,Q,Q_0,R)$ (in short, \mht) 
is defined in the same way as a~\ttt, but has an arbitrary, fixed number of reading heads. 
Each of these heads can walk on the input tree,
independent of the other heads. It can test the label and child number of the node 
that it is currently reading, and additionally apply a regular test to that node. Moreover, 
we assume that the heads are ``sensing'', which means that $M$ can test which heads are 
currently scanning the same node. Thus, if $M$ has $\ell$ heads, then its move rules are of the form 
\[
\tup{q,\sigma_1,j_1,T_1,\dots,\sigma_\ell,j_\ell,T_\ell,E}\to \tup{q',\alpha_1,\dots,\alpha_\ell}
\]
where $E\subseteq [1,\ell]\times [1,\ell]$ is an equivalence relation. 
A configuration of $M$ on input tree $t$ is of the form $\tup{q,u_1,\dots,u_\ell}$,
to which the rule is applicable if $M$ is in state~$q$, 
each $u_i$ satisfies the tests $\sigma_i$, $j_i$, and $T_i$,
and $u_i=u_j$ for every $(i,j)\in E$. 
After application the new configuration is $\tup{q',\alpha_1(u_1),\dots,\alpha_\ell(u_\ell)}$.
The output rules are defined in a similar way. 
Initially all reading heads are at the root of the input tree. 
This is all similar to how multi-head automata on strings are defined. 

We will use the \mht\ $M$ as an acceptor of its domain. We will say 
that it accepts $\dom(M)$ in polynomial time if there is a polynomial $p(n)$ 
such that for every $t\in\dom(M)$ there is a computation $\tup{q_0,\rt_t}\Rightarrow^*_{M,t} s$ 
of length at most $p(|t|)$ for some $q_0\in Q_0$ and $s\in T_\Delta$. 
Note that we consider nondeterministic \mht's only.

\begin{lemma}\label{lem:mht}
For every multi-head \ttt\ $M$, $\dom(M)\in\PTIME$. 
Moreover, if $M$ accepts $\dom(M)$ in polynomial time, then $\dom(M)\in\LOGCF$.
\end{lemma}

\begin{proof}
After this paragraph we will show that the domain of a multi-head \ttt\ can be accepted by an
alternating multi-head finite automaton (in short, \amfa), in a straightforward way. 
Moreover, we will show that if the \mht\ accepts in polynomial time, 
then the corresponding \amfa\ accepts in polynomial tree-size. 
That proves the lemma because 
$\PTIME$ is the class of languages accepted by \amfa's (see~\cite{CKS,Coo}) and  
$\LOGCF$ is the class of languages accepted by \amfa's in polynomial tree-size (see~\cite{Ruz,Sud}).

It is well known that the domain of a classical local \ttt\ can be accepted by an
alternating (one-head) tree-walking automaton, see, e.g., \cite{Slu}, 
\cite[Section~4]{Eng86}, and~\cite[Section~4]{MilSucVia03}, 
and the same is true for the multi-head case.
Let $M=(\Sigma,\Delta,Q,Q_0,R)$ be an \mht. 
The \amfa\ $M'$ that accepts $\dom(M)$ simulates~$M$ on the input $t\in T_\Sigma$, 
without producing output. The reading heads of $M$ are simulated by reading heads 
of $M'$ in the obvious way.
Every (initial) state $q$ of~$M$ is simulated by the existential (initial) state $q$ of $M'$,
and a move rule of~$M$ is simulated by a transition of $M'$ in an obvious way. 
If $M$ applies an output rule in state $q$, then $M'$ first goes into a universal state $q'$
and then branches in the same way as $M$, going into existential states.
A regular test $T$ of $M$ is simulated by $M'$ in a side branch, using an \amfa\ subroutine 
that accepts the context-free language $\tmark(T)$, with additional reading heads. 
Note that since the heads are sensing, the node to be tested is ``marked'' by a reading head. 
Similarly, to move a head $h$ from a parent $u$ to its $i$-th child $ui$, 
$M'$ first moves an auxiliary head $h'$ nondeterministically to a position to the right of $u$,
then checks in a side branch that the string between $h$ and $h'$ 
belongs to the context-free language $T_\Sigma^{i-1}$, and finally moves $h$ to $h'$. 
In a similar way $M'$ can move from $ui$ to $u$, and can determine the child number of $u$. 

If $M$ accepts $t$ in time $m$, 
then the size of the corresponding computation tree of $M'$ is polynomial in $m$,
because each computation step of $M$ takes polynomial tree-size. 
Thus, if $M$ accepts in polynomial time, then $M'$ accepts in polynomial tree-size. 

Note that if we assume that the simulation of a step of $M$ takes constant tree-size, 
and we assume moreover that $M$ only uses output rules (by eventually replacing 
the right-hand side $\zeta$ of each move rule by $\delta(\zeta)$, where $\delta$ has rank~1),
then the output tree of $M$ can be viewed both as 
the derivation tree of the computation of $M$ and as the computation tree of $M'$, roughly speaking.
\qed
\end{proof}

Thus, to prove that $\NTT\circ \TT \subseteq \LOGCF$ it suffices to show,
for every $\tau=\tau_1\circ\tau_2$ with $\tau_1\in\NTT$ and $\tau_2\in\TT$,
that $L_\tau$ can be accepted by a multi-head \ttt\ $M$ in polynomial time. 
Let $M_1$ and $M_2$ be \ttt's that realize $\tau_1$ and~$\tau_2$. 
For an input tree~$t$ and an output tree $s$ of $\tau$, $M$ will simulate $M_1$ on $t$, 
generating an intermediate tree~$r$, and verify that $M_2$ translates $r$ into $s$. 
Since $M$ cannot store its output tree~$r$, it must 
verify the translation of $r$ into $s$ on the fly, i.e., while generating $r$.  
That can be done because the context-free grammar $G_{M_2,r}$ is forward deterministic,
and hence its reduced version has a unique fixed point: 
during the generation of the nodes $v$ of $r$, $M$ can guess the values 
of the nonterminals $\tup{q,v}$ of $G_{M_2,r}$
(which are subtrees of $s$) and check the fixed point equations for them.
However, since $G_{M_2,r}$ need not be reduced, we have to be more careful.  

Let $G=(N,\Delta,\{S\},R)$ be a forward deterministic context-free grammar,
and let $\#$ be a symbol not in $N\cup \Delta$ (which stands for `undefined').
A string homomorphism $h: N\to \Delta^*\cup\{\#\}$ is a \emph{fixed point} of $G$ if 
(1)~$h(S)\neq \#$, 
(2)~$h(X)$ is a substring of $h(S)$ for every $X\in N$ such that $h(X)\neq\#$, and
(3)~$h(X)=h(\zeta)$ for every rule $X\to\zeta$ in $R$ such that $h(X)\neq\#$,
where $h$ is extended to $\Delta$ by defining $h(a)=a$ for every $a\in \Delta$.
In the special case that $G$ is a regular tree grammar, a \emph{tree fixed point} of $G$
is a fixed point $h$ of $G$ such that 
$h(X)\in T_\Delta\cup\{\#\}$ for every $X\in N$
and $h(X)$ is a subtree of $h(S)$ for every $X\in N$ such that $h(X)\neq \#$.

\begin{lemma}\label{lem:fix}
Let $G=(N,\Delta,\{S\},R)$ be a forward deterministic context-free grammar such that $L(G)\neq \nothing$.
For every $w\in \Delta^*$, $L(G)=\{w\}$ if and only if there is a fixed point~$h$ of $G$ such that $h(S)=w$.
If $G$ is a regular tree grammar, then the same statement holds 
for $w\in T_\Delta$ and $h$ a tree fixed point. 
\end{lemma}

\begin{proof}
Let $L(G)=\{w\}$, and define $h_G(X)$ to be the unique string generated by $X$, 
if that exists and is a substring of $w$, and otherwise $h_G(X)=\#$. It is easy to see that 
$h=h_G$ satisfies the requirements.

Let $h$ be a fixed point of $G$ such that $h(S)=w$. 
Then $h(v)=w$ for every sentential form $v$ of $G$. Since $L(G)\neq \nothing$, this shows that $L(G)=\{w\}$.
\qed
\end{proof}

\begin{theorem}\label{thm:logcf}
$\NTT\circ \TT \subseteq \LOGCF$.
\end{theorem}

\begin{proof}
Let $M_1=(\Sigma,\Omega,P,P_0,R_1)$ be a \ttt, and 
let $M_2=(\Omega,\Delta,Q,q_0,R_2)$ be a \dttt. 
We will denote $\tau_{M_1}$ and $\tau_{M_2}$ by $\tau_1$ and $\tau_2$, respectively. 
Since it is easy to prove (as in the proof of Corollary~\ref{cor:prod}) 
that $\NTT\circ \TT = \NTT\ast \TT$, 
we may assume that $(\tau_1,\tau_2)$ is linear-bounded. 
We may also assume, by Lemma~\ref{lem:msorel} and Theorem~\ref{thm:nondetrightcomp},
that $M_2$ is local. That does not change the linear-boundedness of the composition:
if $(\tau_1,\tau'_2\circ\tau''_2)$ is linear-bounded and $\tau'_2\in\NTTR$, then 
$(\tau_1\circ\tau'_2,\tau''_2)$ is linear-bounded because $\tau'_2$ is size-preserving. 
Similarly, we may assume that $\ran(\tau_1)\subseteq\dom(\tau_2)$ 
by Corollaries~\ref{cor:regdom} and~\ref{cor:ranrestr}. 
Finally we assume (as in the proofs of Lemmas~\ref{lem:ttottdl} and~\ref{lem:nondetttottdl}) that 
$M_1$ keeps track in its finite state of the child number of the output node to be generated,
through a mapping $\chi:P\to[0,\m_\Sigma]$. 

On the basis of Lemma~\ref{lem:mht}, we will describe a multi-head \ttt\ $M$ 
that accepts~$L_\tau$ in polynomial time, where $\tau=\tau_1\circ\tau_2$. 
Initially $M$ verifies by a regular test that the input tree is
of the form $\#(t,s)$ with $t\in T_\Sigma$ and $s\in T_\Delta$.
We will denote the root of $\#(t,s)$ by its label $\#$. 
As mentioned before, on input~$\#(t,s)$ the transducer~$M$ simulates $M_1$ on $t$ 
generating an output tree $r$ of $M_1$, which is in the domain of~$M_2$ 
because $\ran(\tau_1)\subseteq\dom(\tau_2)$. 
It keeps the state $p$ of~$M_1$ in its finite state, 
uses one of its heads to point at a node of $t$ (which it initially moves to the root of $t$), 
and instead of a regular test $T$ applies 
the regular test $\{(\#(t,s),1u)\mid (t,u)\in T\}$.\footnote{Note that a node of $t$ has the same 
label and child number in $t$ and $\#(t,s)$, except when it has child number~1 in $\#(t,s)$ in which case 
it has child number $0$ or $1$ in $t$, depending on whether or not its parent in $\#(t,s)$ has label $\#$.
}
While generating $r$ it guesses a tree fixed point $h:\con(r)\to T_\Delta\cup\{\#\}$ 
of the regular tree grammar $G_{M_2,r}$ such that $h(\tup{q_0,\rt_r})=s$. 
If that fixed point can be guessed, then $\tau_2(r)=s$ by Lemma~\ref{lem:fix},
and hence $(t,s)\in\tau$. 

Initially, $M$ guesses the values under $h$ of the configurations 
in $\con(r)$ that contain the root of $r$, in linear time. 
For each $q\in Q$ the value of $\tup{q,\rt_r}$ is 
guessed by nondeterministically moving a reading head named $(q,\stay)$
to a node~$x$ of $s$, i.e., node $2x$ of $\#(t,s)$, meaning that $h(\tup{q,\rt_r})=s|_x$,
or to node~$\#$, meaning that $h(\tup{q,\rt_r})=\#$ (i.e., that $h(\tup{q,\rt_r})$ is ``undefined'').
In particular, the head $(q_0,\stay)$ is moved to the root of $s$, thus guessing that $\tau_2(r)=s$.

Suppose that $M$ is going to produce a node $v$ of $r$ with label $\omega$, by simulating an output rule 
$\tup{p,\sigma,j,T}\to\omega(\tup{p_1,\alpha_1},\dots, \tup{p_k,\alpha_k})$ of $M_1$. 
In such a situation, $M$ has already guessed the values under $h$ of the configurations 
in $\con(r)$ that contain $v$, and also of those that contain the parent $v'$ of $v$ (if it has one). 
For each $q\in Q$ the value of $\tup{q,v}$ is stored using the reading head named $(q,\stay)$,
as explained above for $v=\rt_r$, and the value of $\tup{q,v'}$ is stored in a similar way 
using a reading head named $(q,\up)$.
Now $M$ guesses the values of the configurations that contain the children of $v$, in linear time.
For every $q\in Q$ and $i\in[1,k]$, the value $h(\tup{q,vi})$ 
is guessed by nondeterministically moving a reading head named $(q,\down_i)$ to some node of~$s$ or to $\#$.
Then $M$ checks that these values 
satisfy requirement~(3) of a fixed point of $G_{M_2,r}$ as follows, in linear time. 
If $\tup{q,\omega,\chi(p)}\to \tup{q',\alpha}$ is a move rule of $M_2$ 
such that head $(q,\stay)$ does not point to $\#$,
then $M$ checks that the heads $(q,\stay)$ and $(q',\alpha)$ point to nodes with the same subtree.
It can do this using two auxiliary heads that simultaneously perform a depth-first left-to-right traversal
of those subtrees. Similarly, if 
$\tup{q,\omega,\chi(p)}\to \delta(\tup{q_1,\alpha_1},\dots,\tup{q_m,\alpha_m})$ 
is an output rule of $M_2$ such that head $(q,\stay)$ does not point to $\#$,
then $M$~checks that it points to a node with label $\delta$ 
and that the subtree at the $i$-th child of that node 
equals the subtree at the head $(q_i,\alpha_i)$, for every $i\in[1,m]$. 
After checking the fixed point requirement~(3), 
$M$ outputs the node $v$ and branches in the same way as $M_1$. 
In the $i$-th branch (apart from simulating $M_1$'s rule in the obvious way) it moves
head $(q,\up)$ to the position of head $(q,\stay)$ 
and then moves head $(q,\stay)$ to the position of head $(q,\down_i)$,
for every $q\in Q$, in linear time. 

This ends the description of $M$. It should be clear that 
$\tau_M$ is the set of all pairs $(\#(t,s),r)$ such that 
$(t,r)\in \tau_1$ (because $M$ simulates $M_1$) and 
$\tau_2(r)=s$ (because $M$ computes a tree fixed point $h$ of $G_{M_2,r}$ such that $h(\tup{q_0,\rt_r})=s$).
Hence $\dom(M)=\{\#(t,s)\mid \exists\,r: (t,r)\in \tau_1, \,\tau_2(r)=s\}=L_\tau$.
It remains to show that $M$ accepts $L_\tau$ in polynomial time.

There is a computation of $M_1$ of length at most $\#(P)\cdot|t|\cdot|r|$ 
that translates~$t$ into $r$, because if the number of move rules applied 
between two output rules is more than the number of configurations of $M_1$ on $t$,
then there is a loop in the computation that can be removed. 
Since $(\tau_1,\tau_2)$ is linear-bounded, we may assume that 
the size of $r$ is at most linear in the size of $s$.
Hence the length of that computation is polynomial in $|t|$ and $|s|$, and hence in $|\#(t,s)|$. 
Since $M$~simulates~$M_1$, and each simulated computation step takes linear time (as shown above),
$M$ accepts $\#(t,s)$ in polynomial time. 
\qed
\end{proof}

From Theorem~\ref{thm:logcf} and Lemma~\ref{lem:nmtio},
which says that $\mrNMT\subseteq \FTTD\circ \TT$,
we obtain the following corollary. Note that 
$\NTT\circ \TT$ is larger than $\FTTD\circ \TT$ in two respects.
First, it contains non-finitary translations. Second, it contains total functions
for which the height of the output tree can be double exponential in the height 
of the input tree, viz. $\tau_\mathrm{exp}^2$ in the proof of Proposition~\ref{pro:hier}, 
whereas that is at most exponential for total functions in $\NTTD\circ \TT$
by Theorem~\ref{thm:fun}, Lemma~\ref{lem:ttvsmt}, 
and the paragraph after Corollary~\ref{cor:ttmttt}.

\begin{corollary}\label{cor:mtiologcf}
$\NMTIO\subseteq \mrNMT\subseteq \LOGCF$.
\end{corollary}

As another corollary we even obtain an upper bound on the complexity of the output languages of $\TT$ 
that improves the one of Theorem~\ref{thm:outspace}.
It was proved for attribute grammars in~\cite{EngCompAG}.

\begin{corollary}\label{cor:logcf}
$\TT(\REGT)\subseteq \LOGCF$. 
\end{corollary}

\begin{proof}
Let $L$ be a regular tree language over $\Omega$ and 
let $\tau_2\subseteq T_\Omega\times T_\Delta$ be in~$\TT$. 
Let $\Sigma=\{e\}$ with $\rank(e)=0$, and 
let $\tau_1=\{(e,r)\mid r\in L\}$.
The one-state \ttl\ with rules $\tup{p,e,0}\to \omega(\tup{p,\stay},\dots,\tup{p,\stay})$
for every $\omega\in\Omega$ realizes the translation $\{(e,r)\mid r\in T_\Omega\}$, 
and hence $\tau_1\in \NTT$ by Corollary~\ref{cor:ranrestr}.
Let $\tau=\tau_1\circ\tau_2$.
Then $L_\tau = \{\#(e,s)\mid \exists r: r\in L,\,\tau_2(r)=s\} =
\{\#(e,s)\mid s\in \tau_2(L)\}$.
By Theorem~\ref{thm:logcf} $L_\tau\in\LOGCF$, and hence $\tau_2(L)\in\LOGCF$
because $\tau_2(L)$ is log-space reducible to~$L_\tau$. 
\qed
\end{proof}

Theorem~\ref{thm:logcf} and Corollary~\ref{cor:logcf} can be extended to deal with the yields
of the output trees, as also proved in~\cite{EngCompAG} for attribute grammars
(generalizing the proof in~\cite{Asv81} of $\family{IO}(1)\subseteq\LOGCF$).
For a ranked alphabet $\Sigma$ we define the mapping $y_\Sigma: T_\Sigma\to (\Sigma^{(0)})^*$
such that $y_\Sigma(t)=yt$, the yield of $t$. 
Let $\yield$ be the class of all such mappings $y_\Sigma$.
In what follows we will identify each string $w$ with the monadic tree $\mon(w)$ 
as defined in the proof of Corollary~\ref{cor:outspace}. Hence, as mentioned in that proof,
$\yield\subseteq \TTL$. This even holds if we assume the existence
of special symbols in $\Sigma^{(0)}$ that are skipped when taking the yield of $t$
(such as the symbols $X_0$ in the derivation trees of context-free grammars with $\epsilon$-rules,
cf. Section~\ref{sec:trees}).

\begin{corollary}\label{cor:ylogcf}
$\NTT\circ \TT \circ \yield \subseteq \LOGCF$ and $y\TT(\REGT)\subseteq \LOGCF$.
\end{corollary}

\begin{proof}
It is straightforward to show that $\yield\subseteq \TTP \ast \yield$. 
In fact, the deterministic pruning \ttt\ removes all nodes of rank~1 
and, using regular look-ahead, all subtrees of which the yield is the empty string $\epsilon$ 
(due to the special symbols mentioned above). Consequently, as in the proof of Corollary~\ref{cor:prod},
$\NTT\circ \TT \circ \yield = \NTT\ast (\TT \circ \yield)$.
This allows us to repeat the proof of Theorem~\ref{thm:logcf}, this time with respect to the 
forward deterministic context-free grammar $G'_{M_2,r}$ 
that generates the yields of the trees generated by $G_{M_2,r}$: if $X\to\zeta$ is a rule of $G_{M_2,r}$,
then $X\to y\zeta$ is a rule of $G'_{M_2,r}$. 
Thus, this time the \mht\ $M$ guesses a fixed point $h$ of $G'_{M_2,r}$, rather than a tree fixed point. 
To do this it uses two heads $\tup{q,\stay,\text{left}}$ and $\tup{q,\stay,\text{right}}$ 
instead of the one head $\tup{q,\stay}$, to guess the left- and right-end 
of the substring generated by the configuration $\tup{q,v}$, and similarly for $\up$ and $\down_i$. 
It should be clear that the fixed point requirement~(3) can easily be checked, 
showing that one such substring equals another one or is the concatenation of several other ones. 
\qed
\end{proof}

The inclusion $\NTT\circ \TT \subseteq \LOGCF$ of Theorem~\ref{thm:logcf}
has consequences for both space and time complexity. We first consider space complexity. 

Since $\LOGCF\subseteq \family{DSPACE}(n)$, we obtain that $\NTT\subseteq \family{DSPACE}(n)$ 
from Theorem~\ref{thm:logcf}. This can easily be generalized to arbitrary compositions of \ttt's.

\begin{theorem}\label{thm:transcompdspace}
For every $k\geq 1$, $\NTT^k\subseteq \family{DSPACE}(n)$.
\end{theorem}

\begin{proof}
The proof is by induction on $k$, with an induction step similar to the one 
in the proof of Theorem~\ref{thm:lintime}.

Let $\tau=\tau_1\circ\tau_2$ such that $\tau_1\in\NTT$ and $\tau_2\in \NTT^k$, $k\geq 1$. 
For a given input string $\#ts$ it has to be checked whether $(t,s)\in \tau$. 
By Corollary~\ref{cor:prod}(1) we may assume that $(\tau_1,\tau_2)$ is linear-bounded.
Hence there is a constant $c\in\nat$ such that for every $(t,s)\in \tau$ 
there is an intermediate tree $r$ such that $|r|\leq c\cdot|s|$. 
To check whether $(t,s)\in \tau$ a deterministic Turing machine systematically enumerates 
all trees~$r$ such that $|r|\leq c\cdot|s|$ (cf. the proof of Theorem~\ref{thm:outspace}).
For each such $r$ it can check in linear space whether $(t,r)\in \tau_1$ by the case $k=1$.
Moreover, by induction it can check in linear space whether $(r,s)\in \tau_2$.
Thus it uses space $O(|t|+|r|) + O(|r|+|s|) = O(|t|+|s|)$. 
\qed
\end{proof}

This result allows us to prove one of our main results, viz. 
that the output languages of $\NTT^k$ are in $\family{DSPACE}(n)$, originally proved in~\cite{Ina}.
It generalizes the main result of~\cite{Bak78} from classical top-down tree transducers 
to tree-walking tree transducers and macro tree transducers. 

\begin{theorem}\label{thm:outdspace}
For every $k\geq 1$, 
\[
\NTT^k(\REGT)\subseteq \family{DSPACE}(n) \text{\quad and \quad} \NMT^k(\REGT)\subseteq \family{DSPACE}(n).
\]
\end{theorem}

\begin{proof}
The proof is similar to the one of Theorem~\ref{thm:outspace}. 
Let $L\in\REGT$ and $\tau\in\NTT^k$.
By Corollary~\ref{cor:prod}(1), $\tau=\tau_1\circ\tau_2$ 
such that $\tau_1\in\NTTP$, $\tau_2\in \NTT^k$, 
and $(\tau_1,\tau_2)$ is linear-bounded for some constant~$c$.
Let $L'=\tau_1(L)$, and note that $\tau(L)=\tau_2(L')$ and that 
$L'\in\REGT$ by Lemma~\ref{lem:pru}.
It is straightforward to show that for every $s\in \tau(L)$ there exists $t\in L'$ such that
$(t,s)\in\tau_2$ and $|t|\leq c\cdot|s|$. 
To check whether a given tree~$s$ is in $\tau(L)$, a deterministic Turing machine 
enumerates all input trees~$t$ (of $\tau_2$) such that $|t|\leq c\cdot|s|$.
For each such~$t$ it first checks that $t\in L'$ in space $O(|t|)=O(|s|)$.
Then it uses the algorithm of Theorem~\ref{thm:transcompdspace} to 
check that $(t,s)\in\tau_2$ in space $O(|t|+|s|)=O(|s|)$.

The inclusion for $\NMT^k$ is now immediate from Lemma~\ref{lem:mtintt}.
\qed
\end{proof}

As before, Theorems~\ref{thm:transcompdspace} and~\ref{thm:outdspace}
also hold for \abb{IO} (multi-return) macro tree translations, pebble tree translations, and 
high-level tree translations, which can be 
realized by compositions of \ttt's (see Section~\ref{sec:lsi}).

By the proof of Corollary~\ref{cor:outspace}, 
Theorem~\ref{thm:outdspace} implies that 
\[
y\NTT^k(\REGT)\subseteq \family{DSPACE}(n) \text{\quad and \quad} y\NMT^k(\REGT)\subseteq \family{DSPACE}(n)
\]
for every $k\geq 1$. Hence the \mbox{\abb{OI}-hierarchy} is also contained in $\family{DSPACE}(n)$,
cf. Corollaries~\ref{cor:oik} and~\ref{cor:iok}.

\begin{corollary}\label{cor:oikdet}
For every $k\geq 1$, $\family{OI}(k)\subseteq \family{DSPACE}(n)$.
\end{corollary}

Next we consider time complexity. 
Since $\LOGCF\subseteq \PTIME$, it follows from Theorem~\ref{thm:logcf}
that $\NTT \circ \TT \subseteq \PTIME$. This result can be generalized as follows. 

One way to increase the power of the \ttt\ is to give it a more powerful feature of look-around.
For a class $\cL$ of tree or string languages, we define the \ttt\ \mbox{\emph{with~$\cL$~look-around}}
by allowing the \ttt\ to use node tests $T$ such that $\tmark(T)\in\cL$. Similarly we obtain the 
\mht\ with $\cL$ look-around. We now consider in particular the case where $\cL=\PTIME$.
Obviously, (the proof of) the first sentence of Lemma~\ref{lem:mht} is still valid 
for a multi-head \ttt\ $M$ with $\PTIME$ look-around.
Thus, the domain of an \mht\ with $\PTIME$ look-around is in $\PTIME$, and hence, in particular, 
the domain of a \ttt\ with $\PTIME$ look-around is in $\PTIME$.
This implies that Lemma~\ref{lem:nondetttottdl}, and hence Theorem~\ref{thm:nondetrightcomp}, also holds 
if the first transducer has $\PTIME$ look-around.
From the proof of Theorem~\ref{thm:logcf} it now easily follows that 
$\NTTPTIME \circ \TT \subseteq \PTIME$,
where the feature of $\PTIME$ look-around is indicated by a superscript $\family{P}$.
This, in its turn, implies the following variant of Corollary~\ref{cor:mtiologcf}
for (multi-return) \abb{io} macro tree transducers with $\PTIME$ look-around 
(appropriately defined): $\NMTIOPTIME\subseteq \mrNMTPTIME\subseteq \PTIME$.
Examples of tree languages in $\PTIME$ that can be used as look-around are those in $\TT(\REGT)$,
by Corollary~\ref{cor:logcf}, and the tree languages defined by bottom-up tree automata 
with equality and disequality constraints (\cite{BogTis}), 
which can obviously be accepted by a multi-head \ttt. 

In the remainder of this section we show that there are translations in \mbox{$\TT\circ\NTT$}, 
even in $\TTD\circ\NTT$, for which the membership problem is $\family{NP}$-complete.
We will use a reduction of SAT, the satisfiability problem of boolean formulas
(see, e.g.,~\cite{GarJoh}), to such a membership problem.

Let $\Delta=\{\vee,\wedge,\neg,\sv,\se\}$ with  
$\Delta^{(2)}=\{\vee,\wedge\}$, $\Delta^{(1)}=\{\neg,\sv\}$, and $\Delta^{(0)}=\{\se\}$.
Let $\cB$ be the set of all trees over $\Delta$ 
generated by the regular tree grammar with nonterminals $F$ and $V$, initial nonterminal $F$, 
and rules $F\to \vee(F,F)$, $F\to \wedge(F,F)$, $F\to \neg(F)$, $F \to V$, $V\to \sv(V)$, and $V\to \sv(\se)$.
Thus, $\cB$ is the set of all boolean formulas that use boolean
variables of the form $\sv^\ell\se$ for $\ell\geq 1$. 
For a boolean formula $\phi$ we define $\nu(\phi)$
to be the nesting-depth of its boolean operators, i.e., $\nu(\phi)=0$ if $\phi$ is a variable,
$\nu(\vee(\phi_1,\phi_2))=\nu(\wedge(\phi_1,\phi_2))=\max\{\nu(\phi_1),\nu(\phi_2)\}+1$, and 
$\nu(\neg(\phi))=\nu(\phi)+1$. 
For every $m\geq 0$ and $n\geq 1$, let $\cB(m,n)$ be the set of all 
formulas $\phi\in\cB$ such that $\nu(\phi)\leq m$, and 
$\ell\in[1,n]$ for every $\sv^\ell\se$ that occurs in $\phi$.
Thus, the formulas in $\cB(m,n)$ have nesting-depth at most $m$ and 
use at most the variables $\sv\se,\sv\sv\se,\dots,\sv^n\se$.  

The proof of the next lemma is essentially a variant of the one of~\cite[Theorem~3.1]{vLe75}.
Let $\Sigma=\{c,d,0,1,a\}$ with $\Sigma^{(1)}=\{c,d,0,1\}$ and $\Sigma^{(0)}=\{a\}$.

\begin{lemma}\label{lem:leeuw}
There is a translation $\tau\in\FTTDL$ such that, 
for every $m\geq 0$ and every string $w\in\{0,1\}^*$ of length $n\geq 1$, 
the set $\tau(d^mcwa)$ consists of all boolean formulas $\phi\in\cB(m,n)$ 
such that $\phi$ is true when the value of $\sv^\ell\se$ is the $\ell$-th symbol of $w$ 
for every $\ell\in[1,n]$. 
\end{lemma}

\begin{proof}
We construct the top-down local \ttt\ $M=(\Sigma,\Delta,\{q_0,q_1\},\{q_1\},R)$. 
Note that the initial state is $q_1$.
The boolean operations $i\vee j$, $i\wedge j$, and $\neg\,i$ on $\{0,1\}$ are defined as usual, 
where $0$ stands for `false' and $1$ for `true'. 
Since the child numbers of the nodes of the input tree will be irrelevant, 
we omit them from the left-hand sides of the rules of $M$. 
The only instruction used in the right-hand sides of the rules is $\alpha=\down_1$.
The rules are the following, for every $i,j\in\{0,1\}$.
\[
\begin{array}{llllll}
\tup{q_{i\vee j},d} & \to & \vee(\tup{q_i,\alpha},\tup{q_j,\alpha}) & 
\tup{q_i,c} & \to & \sv(\tup{q_i,\alpha}) \\[0.4mm]
\tup{q_{i\wedge j},d} & \to & \wedge(\tup{q_i,\alpha},\tup{q_j,\alpha}) \quad & 
\tup{q_i,j} & \to & \sv(\tup{q_i,\alpha}) \\[0.4mm]
\tup{q_{\neg\,i},d} & \to & \neg(\tup{q_i,\alpha}) & 
\tup{q_i,i} & \to & \se \\[0.4mm]
\tup{q_i,d} & \to & \tup{q_i,\alpha} &  &  & 
\end{array}
\]
Let $u$ be the node of the input tree $t=d^mcwa$ with label $c$.
After consuming~$d^m$, the \ttt\ $M$ has nondeterministically generated any output form 
that is a boolean formula $\phi$ of nesting-depth at most $m$ 
and with the two configurations $\tup{q_i,u}$ as variables, such that $\phi$ is true when 
the value of $\tup{q_i,u}$ is $i$. 
For instance, in the first step of that computation $M$ consumes $d$ 
and changes the initial output form $\tup{q_1,\rt_t}$ into one of the output forms 
$\vee(\tup{q_1,x},\tup{q_0,x})$, $\vee(\tup{q_0,x},\tup{q_1,x})$, $\vee(\tup{q_1,x},\tup{q_1,x})$, $\wedge(\tup{q_1,x},\tup{q_1,x})$, $\neg(\tup{q_0,x})$, or $\tup{q_1,x}$, 
where $x$ is the child of $\rt_t$.
After that, each $\tup{q_i,u}$ generates any variable $\sv^\ell\se$ such that
the $\ell$-th symbol of $w$ is $i$. 
Note that since $i$ and $j$ are not necessarily distinct, 
$M$ has in particular the rule $\tup{q_i,i} \to \sv(\tup{q_i,\alpha})$ for every $i\in\{0,1\}$.
Thus, $q_i$ can nondeterministically choose any occurrence of $i$ in $w$ 
to output $\se$ and end the computation.
\qed
\end{proof}

Applying the translation $\tau$ of Lemma~\ref{lem:leeuw} to the regular tree language $L$
consisting of all trees $d^mcwa$ such that $m\geq 0$ and $w$ is a nonempty string over $\{0,1\}$, 
produces the set $\tau(L)$ of all satisfiable formulas in $\cB$. 
Thus, since the membership problem for that set is $\family{NP}$-complete, 
we obtain the following corollary that was proved in~\cite{Rou73}, 
as already mentioned after Corollary~\ref{cor:oik}.
Note that it is easy to prove
that $\FTTD(\REGT)\subseteq y\FTTD(\REGT)$: just change every output rule 
$\tup{q,\sigma,j,T}\to \delta(\tup{q_1,\alpha_1},\dots,\tup{q_k,\alpha_k})$ 
into the (general) rule
$\tup{q,\sigma,j,T}\to \omega_{k+1}(\delta,\tup{q_1,\alpha_1},\dots,\tup{q_k,\alpha_k})$
where $\omega_{k+1}$ has rank $k+1$ (and $\delta$ now has rank~$0$).

\begin{corollary}
There is an $\family{NP}$-complete language in $\FTTD(\REGT)$,
and hence there is one in $y\FTTD(\REGT)$.
\end{corollary}

We now prove the existence of a translation in $\TTD\circ\NTT$ 
for which the membership problem is $\family{NP}$-complete. 
Recall that, for a tree translation $\tau$, we denote by $L_\tau$ the tree language 
$\{\#(t,s)\mid (t,s)\in\tau\}$.

\begin{theorem}\label{thm:npcomplete}
There is a translation $\tau\in \TTDL\circ \TTL\circ \NTTPL\subseteq \TTD\circ \FTT$ 
such that $L_\tau$ is $\family{NP}$-complete.
\end{theorem}

\begin{proof}
The inclusion $\TTL\circ \NTTPL\subseteq \FTT$ is immediate from Lemma~\ref{lem:nondetttottdl}.
We first describe a translation $\tau\in \TTDL\circ \FTTL$ 
such that $L_\tau$ is $\family{NP}$-complete.
Let $\Gamma=\{a,b,c,d,e\}$ with $\Gamma^{(1)}=\{a,b,c,d\}$ and $\Gamma^{(0)}=\{e\}$. 
The translation $\tau\subseteq T_\Gamma\times T_\Delta$ transforms each tree
$t=ab^ncd^me$ into all satisfiable boolean formulas in $\cB(m,n)$.
This will be realized by the composition of two \ttt's $M_1$ and $M_2$ such that 
the deterministic \ttt\ $M_1$ transforms $t$ into a tree $s$ 
of which the path language\footnote{The path language of a tree $s\in T_\Omega$ 
consists of all strings in $\Omega^*$ that are obtained by walking along a path 
from the root of $s$ to one of its leaves, 
writing down the labels of the nodes of that path from left to right.
} 
consists of all strings $awcd^me$ with $w\in\{0,1\}^*$ of length~$n$,
and $M_2$ nondeterministically chooses a leaf of $s$ and then 
walks back to the root of $s$ while 
simulating the transducer $M$ of (the proof of) Lemma~\ref{lem:leeuw} 
on the tree $d^mcwa\in T_\Sigma$. 
Thus, $M_1$~provides all possible valuations of the variables $\sv\se,\sv\sv\se,\dots,\sv^n\se$
and $M_2$ chooses one such valuation and produces all formulas in $\cB(m,n)$
that are true for that valuation.

Let $\Omega=\{a,0,1,c,d,e\}$ with $\Omega^{(2)}=\{a,0,1\}$, $\Omega^{(1)}=\{c,d\}$, 
and $\Omega^{(0)}=\{e\}$.
We define $\tau=\tau_{M_1}\circ\tau_{M_2}\subseteq T_\Gamma\times T_\Delta$ 
where $M_1$ and $M_2$ are the following \ttt's.
The deterministic \ttdl\ $M_1=(\Gamma,\Omega,\{q,q_0,q_1,p\},\{q\},R_1)$ 
has the following rules, for $i\in\{0,1\}$ and $\alpha=\down_1$. 
\[
\begin{array}{llllll}
\tup{q,a,0} & \to & a(\tup{q_0,\alpha},\tup{q_1,\alpha}) & 
\tup{p,d,1} & \to & d(\tup{p,\alpha}) \\[0.4mm]
\tup{q_i,b,1} & \to & i(\tup{q_0,\alpha},\tup{q_1,\alpha}) \quad & 
\tup{p,e,1} & \to & e \\[0.4mm]
\tup{q_i,c,1} & \to & c(\tup{p,\alpha}) &  &  & 
\end{array}
\]
It should be clear that for an input tree $ab^ncd^me$, with $m\geq 0$ and $n\geq 1$, 
the path language of the tree $\tau_{M_1}(ab^ncd^me)$ consists of all strings $awcd^me$ 
with $w\in\{0,1\}^*$ of length $n$.

The \ttl\ $M_2=(\Omega,\Delta,Q,Q_0,R_2)$ 
has states $Q=\{q_2,q_0,q_1\}$ and $Q_0=\{q_2\}$. 
On an input tree $\tau_{M_1}(ab^ncd^me)$,
it walks nondeterministically in state $q_2$ from the root to some leaf (without producing output), 
moves to the parent of that leaf, 
and then simulates the transducer $M$ of Lemma~\ref{lem:leeuw} on the 
tree $d^mcwa\in T_\Sigma$ while walking back to the root. 
It starts that simulation in the state $q_1$ of $M$,
and then uses the rules of $M$ with $\alpha=\up$. 

With this definition of $M_1$ and $M_2$, it follows from Lemma~\ref{lem:leeuw} that 
the set $\tau(ab^ncd^me)$ consists of all boolean formulas 
$\phi\in\cB(m,n)$ such that $\phi$ is satisfiable. 
Thus, for a formula $\phi\in\cB(m,n)$, $\phi$ is satisfiable if and only if 
$\#(ab^ncd^me,\phi)$ is in~$L_\tau$. This shows that satisfiability is reducible to membership in~$L_\tau$,
because the nesting-depth $m$ of $\phi$ and the number $n$ of variables it uses, 
can easily be computed from any $\phi\in\cB$ in polynomial time. 

We finally show that $\tau_{M_2}\in \TTL\circ \NTTPL$, by a standard technique
(see, e.g.,~\cite[Section~6.1]{EngVog85}). In fact, we will show 
that $\tau_{M_2}\in \TTL\circ \SET$, cf. the proof of Lemma~\ref{lem:mtintt}.
Let $+$ be a new symbol of rank~2, and $\theta$ a new symbol of rank~0. 
Let $M_2'$ be the deterministic \ttl\ with output alphabet $\Delta\cup\{+,\theta\}$
that is obtained from $M_2$ as follows. For every triple $\tup{q,\omega,j}$ such that $q\in Q$, $\omega\in\Omega$,
and $j\in[0,\m_\Omega]$, if $\tup{q,\omega,j}\to \zeta_1, \dots,\tup{q,\omega,j}\to \zeta_r$ 
are all the rules of $M_2$ with left-hand side $\tup{q,\omega,j}$, then $M'_2$ has the rule
$\tup{q,\omega,j}\to +(\zeta_1,+(\zeta_2,\zeta_3))$ if $r=3$, the rule 
$\tup{q,\omega,j}\to +(\zeta_1,\zeta_2)$ if $r=2$, the rule 
$\tup{q,\omega,j}\to \zeta_1$ if $r=1$, and the rule $\tup{q,\omega,j}\to \theta$ if $r=0$.
Let $M_3$ be the pruning \ttt\ 
with one state~$p$ and rules $\tup{p,\delta,j}\to \delta(\tup{p,\down_1},\dots,\tup{p,\down_k})$
for every $\delta\in\Delta^{(k)}$, plus the rules $\tup{p,+,j}\to \tup{p,\down_1}$ and 
$\tup{p,+,j}\to \tup{p,\down_2}$ (for every child number~$j$). 
Since $M_2$ first moves from the root to a leaf, and then moves back to the root, 
it does not have infinite computations. From that it should be clear that 
$\tau_{M_2}= \tau_{M'_2}\circ\tau_{M_3}$. 
\qed
\end{proof}

\begin{corollary}
There is a translation $\tau\in \NMT$ 
such that $L_\tau$ is $\family{NP}$-complete.
\end{corollary}

\begin{proof}
By Lemma~\ref{lem:ttvsmt}, $\TTDL\circ \TTL\subseteq\MT$.
Moreover, by~\cite[Theorem~7.6(3)]{EngVog85}, $\MT\circ\NTTPL\subseteq \NMT$.
Hence the translation $\tau$ of Theorem~\ref{thm:npcomplete} is in $\NMT$.
\qed
\end{proof}

Since $\NMT\subseteq \family{MT$^2_\text{\abb{IO}}$}$ 
by~\cite[Theorem~6.10]{EngVog85}, this also shows that 
there is a translation $\tau\in\family{MT$^2_\text{\abb{IO}}$}$
such that $L_\tau$ is $\family{NP}$-complete, cf. Corollary~\ref{cor:mtiologcf}.

\section{Forest Transducers}\label{sec:forest}

Whereas we have considered ranked trees until now, i.e., trees over a ranked alphabet,
XML documents naturally correspond to unranked trees or \emph{forests}, 
over an ordinary unranked alphabet.  
For that reason we now consider transducers that transform forests into forests.
Rather than generalizing the \ttt\ to a ``forest-walking forest transducer'', 
we take the equivalent, natural approach of letting the \ttt\ transform representations
of forests by (ranked) trees, cf.~\cite{MilSucVia03} and~\cite[Section~11]{EngHooSam}.  

For an ordinary (unranked) alphabet $\Sigma$ the set $F_\Sigma$ of forests 
over $\Sigma$ is the language generated by the context-free grammar with nonterminals $F$ and $T$,
initial nonterminal $F$, set of terminals $\Sigma\cup \{[\,,]\}$, 
where $\{[\,,]\}$ is the set consisting of the left and right square bracket,
and rules $F\to \epsilon$, $F\to TF$, and $T\to \sigma[F]$ for every $\sigma\in\Sigma$.
Thus, intuitively, a forest is a sequence of unranked trees, and an unranked tree is of the form
$\sigma[t_1\cdots t_n]$ where each $t_i$ is an unranked tree. Note that 
every forest $f\in F_\Sigma$ can be uniquely written as 
$f=\sigma[f_1]f_2$ with $\sigma\in\Sigma$ and $f_1,f_2\in F_\Sigma$. 

As usual, forests can be encoded as binary trees. 
With $\Sigma$ we associate the ranked alphabet $\Sigma_e =\Sigma\cup\{e\}$
where $e$ has rank~0 and every $\sigma\in\Sigma$ has rank~2. 
The mapping $\enc_\Sigma: F_\Sigma \to T_{\Sigma_e}$ is defined as follows. 
The encoding of the empty forest is $\enc_\Sigma(\epsilon)=e$, and recursively, 
the encoding of a forest $f=\sigma[f_1]f_2$ is 
$\enc_\Sigma(f)=\sigma(\enc_\Sigma(f_1),\enc_\Sigma(f_2))$. 
The mapping $\enc_\Sigma$ is a bijection, and the inverse decoding is denoted by $\dec_\Sigma$. 
Let $\Enc$ and $\Dec$ denote the classes of encodings $\enc_\Sigma$ and decodings $\dec_\Sigma$,
respectively, for all alphabets $\Sigma$. 
We define $\NFT=\Enc\circ\NTT\circ \Dec$ to be the class of \emph{\ttt\ forest translations}.
Thus, a \ttt\ forest translation is of the form 
$\tau = \enc_\Sigma \circ \tau_M \circ \dec_\Delta$ where $\Sigma$ and $\Delta$ are alphabets and 
$M$ is a \ttt\ with input alphabet $\Sigma_e$ and output alphabet $\Delta_e$, 
which in this context can be called a \ttt\ forest transducer. 
We first restrict attention to deterministic \ttt\ forest transducers, 
i.e., to the class $\FT=\Enc\circ\TT\circ \Dec$.

The next simple lemma shows that the encodings of compositions are the compositions of encodings 
(of deterministic \ttt's).

\begin{lemma}\label{lem:fttt}
For every $k\geq 1$, $\FT^k = \Enc\circ\TT^k\circ \Dec$.
\end{lemma} 

\begin{proof}
The inclusion $\FT^k \subseteq \Enc\circ\TT^k\circ \Dec$ is obvious, 
because $\dec_\Delta\circ \enc_\Delta$ is the identity on $T_{\Delta_e}$ 
for every (unranked) alphabet $\Delta$.
To show that $\Enc\circ\TT^k\circ \Dec\subseteq \FT^k$, it suffices to prove that 
$\TT\circ\TT \subseteq \TT\circ\Dec\circ\Enc\circ\TT$. Let~$\Gamma$~be 
the (ranked) output alphabet of a first transducer, which is also the input alphabet of the second, and 
let $\mathrm{id}_\Gamma$ be the identity on $T_\Gamma$. By the composition results of 
Theorems~\ref{thm:rightcomp} and~\ref{thm:ttsuott}, it now suffices to show that 
$\mathrm{id}_\Gamma\in \TTDL\circ\Dec\circ\Enc\circ\TTLSU$. We do this by encoding the trees 
over $\Gamma$ as binary trees, similar to the transformation of the derivation trees 
of a context-free grammar into those of its Chomsky Normal Form. Let $\omega$ be a new symbol,
and let $\Delta$ be the unranked alphabet $\Gamma\cup\{\omega\}$. We encode the trees over 
$\Gamma$ as trees over the ranked alphabet $\Delta_e$, which are the usual encodings 
of forests over $\Delta$. The encoding $h: T_\Gamma\to T_{\Delta_e}$ is defined as follows:
for every $\gamma\in\Gamma^{(k)}$, if $h(t_i)=t'_i$ for every $i\in[1,k]$,
then $h(\gamma(t_1,t_2,\dots,t_k))= \gamma(e,\omega(t'_1,\omega(t'_2,\dots \omega(t'_k,e)\cdots)))$.
It should be clear that $h$ is an injection. It should also be clear that $h\in \TTDL$ 
(in fact, $h$ is a tree homomorphism, which can be realized by a classical top-down tree transducer).
Finally, it is also easy to construct a local top-down single-use \ttt\ $M$ such that 
$\tau_M(h(t))=t$ for every $t\in T_\Gamma$. It has the set of states $Q=\{q_i\mid i\in[0,\m_\Gamma]\}$
with initial state $q_0$, and the following rules 
(where $\gamma\in\Gamma^{(k)}$, $j\in[0,2]$, and $q_i\in Q$, $i\neq 1$): 
\[
\begin{array}{lll}
\tup{q_0,\gamma,j} & \to & \gamma(\tup{q_1,\down_2},\dots,\tup{q_k,\down_2}) \\[0.4mm]
\tup{q_1,\omega,2} & \to & \tup{q_0,\down_1} \\[0.4mm]
\tup{q_i,\omega,2} & \to & \tup{q_{i-1},\down_2}
\end{array}
\]
Note that $\gamma$ and $\omega$ have rank~2 in $\Delta_e$. 
\qed
\end{proof}

We now wish to show that our main results also hold for deterministic \ttt\ forest translations.
Let us first consider the complexity results of Section~\ref{sec:complex}. 
It is easy to see that for every alphabet $\Sigma$, the mappings $\enc_\Sigma$ and $\dec_\Sigma$ 
can be computed by a deterministic Turing machine in linear time and space, 
simulating a one-way pushdown transducer.\footnote{In fact, 
$\enc_\Sigma$ can even be computed without pushdown: for every forest $f\in F_\Sigma$,
$\enc_\Sigma(f)$ can be obtained from $f$ by removing all left-brackets, 
changing each right-bracket into~$e$, and adding one $e$ at the end.
\label{foo:enc}
}
This implies, by Lemma~\ref{lem:fttt}, 
that Theorems~\ref{thm:lintime} and~\ref{thm:linspace} also hold for $\FT^k$.
We define a set of forests $L\subseteq F_\Sigma$ to be a \emph{regular forest language}
if $\enc_\Sigma(L)\in\REGT$, and we denote the class of regular forest languages by $\REGF$.
Then, for every $k\geq 1$, the class $\FT^k(\REGF)$ of output forest languages is included in 
the class $\Dec(\TT^k(\REGT))$ by Lemma~\ref{lem:fttt}.
Let $L\in\REGT$ and $\tau\in\TT^k$ with output alphabet $\Delta_e$. 
Then a forest $f$ over $\Delta$ is in $\dec_\Delta(\tau(L))$ if and only if 
$\enc_\Delta(f)$ is in $\tau(L)$. 
That implies that Theorem~\ref{thm:outspace} also holds for $\FT^k$, in the sense that
$\FT^k(\REGF)\subseteq \family{DSPACE}(n)$. 

Next we consider the results of Section~\ref{sec:lsi},
and extend the class $\LSI$ in the obvious way to forest translations.
Since it is easy to show that for every forest $f\in F_\Sigma$,
we have $|\enc_\Sigma(f)| = \frac{2}{3}|f| +1$ (see footnote~\ref{foo:enc}), a translation 
$\tau' = \enc_\Sigma \circ \tau \circ \dec_\Delta$ is of linear size increase if and only if 
$\tau$ is of linear size increase. 
Thus, since $\FT^k = \Enc\circ\TT^k\circ \Dec$ by Lemma~\ref{lem:fttt}, 
it is decidable for a given composition of deterministic \ttt\ forest transducers whether or not 
it is of linear size increase. And if so, an equivalent deterministic \ttt\ forest transducer 
can be constructed:  
$\FT^k \cap \LSI = \Enc\circ(\TT^k\cap \LSI)\circ \Dec=
\Enc\circ\TTSU\circ\Dec\subseteq \FT$. 
Intuitively, $\Enc\circ\TTSU\circ\Dec$ is the class of translations realized by
``single-use forest-walking forest transducers''. 
Since $\TTSU=\MSOT$ by Proposition~\ref{pro:mso=ttsu}, it is also the class $\Enc\circ\MSOT\circ\Dec$.
Viewing forests as graphs, and hence as logical structures, in the obvious way (just as trees), 
every encoding $\enc_\Sigma$ and every decoding $\dec_\Sigma$ is a deterministic (i.e., parameterless)
\abb{mso} translation, as defined in~\cite[Chapter~7]{thebook}.
Hence, by the closure of \abb{mso} translations under composition \cite[Theorem~7.14]{thebook}, 
$\Enc\circ\MSOT\circ\Dec$ equals the (natural) class of 
deterministic \abb{mso} translations from forests to forests. 

As observed in~\cite{PerSei} for macro tree transducers,
whereas the encoding of forests as binary trees is quite natural for the input forest of a \ttt,
for the output forest it is less natural, because it forces the \ttt\ to 
generate the output forest $f$ in its unique form $f=\sigma[f_1]f_2$. It is more natural 
to additionally allow the \ttt\ to generate $f$ as a concatenation $f_1f_2$ 
of two forests $f_1$ and~$f_2$. To formalize this, as in~\cite[Section~7]{Eng09}
and in accordance with~\cite{PerSei}, 
we associate with an alphabet $\Delta$ the ranked alphabet $\Delta_\at=\Delta\cup\{\at,e\}$
where $\at$ has rank~2, $e$ has rank~0, and every $\delta\in\Delta$ has rank~1.  
The mapping $\fl_\Delta: T_{\Delta_\at}\to F_\Delta$ is a ``flattening'' defined as follows 
(for $t_1,t_2\in T_{\Delta_\at}$ and $\delta\in\Delta$): 
$\fl_\Delta(e)=\epsilon$, $\fl_\Delta(\at(t_1,t_2)) = \fl_\Delta(t_1)\fl_\Delta(t_2)$, 
the concatenation of $\fl_\Delta(t_1)$ and $\fl_\Delta(t_2)$, and 
$\fl_\Delta(\delta(t_1))= \delta[\fl_\Delta(t_1)]$. 
The mapping $\fl_\Delta$ is surjective but, in general, not injective.  
Let $\Flat$ denote the class of flattenings $\fl_\Delta$, for all alphabets $\Delta$. 
We define $\NFT_\at=\Enc\circ\NTT\circ \Flat$ to be the class of \emph{extended \ttt\ forest translations}.
An extended \ttt\ forest tree transducer is a \ttt\ 
with input alphabet $\Sigma_e$ and output alphabet $\Delta_\at$. 
Again, we first restrict attention to deterministic transducers, i.e., to the class
$\FT_\at=\Enc\circ\TT\circ \Flat$.

Let us show that there is an extended \ttt\ forest translation in $\FT_\at$ that is not in $\FT$.
That was shown for macro tree transducers in~\cite[Theorem~8]{PerSei} by a similar argument. 
Let $\Gamma=\{\sigma\}$ and $\Omega=\{\delta\}$ be alphabets, and let us identify 
the forest $\sigma[\;]$ with the symbol $\sigma$, and similarly $\delta[\;]$ with $\delta$.
Then $\Gamma^*\subseteq F_\Gamma$ and $\Omega^*\subseteq F_\Omega$. 
There is a deterministic extended \ttt\ forest transducer that translates the string $\sigma^n$ 
into the string $\delta^{2^{n+1}}$ for every $n\in\nat$. 
In fact, let $M$ be the \dttt\ (with general rules) 
that is obtained from the \dttt\ $M_\mathrm{exp}$ of Example~\ref{exa:exp} by changing 
its output alphabet into $\Omega_\at=\{\at,\delta,e\}$, and changing $\sigma$ into $\at$ 
and $e$ into $\delta(e)$ in the right-hand sides of its rules.
Note that the input alphabet $\Sigma$ of $M_\mathrm{exp}$ and $M$ equals $\Gamma_e$.
The input tree $t_n=\enc_\Gamma(\sigma^n)=\sigma(e,\sigma(e,\dots\sigma(e,e)\cdots))$ 
is translated by~$M_\mathrm{exp}$ into the full binary tree $s_n$ over $\Sigma$ with $2^{n+1}$ leaves. 
Clearly, $M$ translates $t_n$ into the tree $s'_n$ that is obtained from $s_n$ 
by changing every $\sigma$ into $\at$ and every $e$ into $\delta(e)$. 
Thus, $\fl_\Omega(s'_n)=\delta^{2^{n+1}}$. This forest translation is not in $\FT$, because 
$|\enc_\Gamma(\sigma^n)|=|t_n|=2n+1$ but  
the height of $s''_n=\enc_\Omega(\delta^{2^{n+1}})$ is $2^{n+1}$, and so, by Lemma~\ref{lem:sizeheight}, 
there is no \dttt\ that translates $t_n$ into $s''_n$. 

We will show that $\FT\subseteq\FT_\at\subseteq\FT^2$. A similar result was proved for 
macro tree transducers in~\cite[Theorem~8 and Corollary~12]{PerSei}. 
To compare $\FT$ and $\FT_\at$, and their compositions, we establish two relationships between 
$\Dec$ and $\Flat$ in the next lemma. 

\begin{lemma}\label{lem:decflat}
$\Dec\subseteq\TTDL \circ \Flat$ and $\Flat\subseteq\TTLSU\circ\Dec$.
\end{lemma}

\begin{proof}
To show the first inclusion, 
let $\Delta$ be an alphabet and define the mapping $h\colon T_{\Delta_e}\to T_{\Delta_\at}$ such that 
$h(e)=e$ and if $h(t_1)=t'_1$ and $h(t_2)=t'_2$, then $h(\delta(t_1,t_2))= \at(\delta(t'_1),t'_2)$. 
It is straightforward to prove that $h\circ \fl_\Delta = \dec_\Delta$. 
It is also easy to show that $h\in \TTDL$ (as in the proof of Lemma~\ref{lem:fttt}, 
$h$ is a tree homomorphism, which can be realized by a classical top-down tree transducer). 
Hence $\dec_\Delta \in \TTDL \circ \Flat$.

For the second inclusion, let $\Delta$ be an alphabet. 
The mapping $\fl_\Delta\circ\enc_\Delta$ can be realized by a local single-use \dttt\ 
$M = (\Delta_\at, \Delta_e, Q, q_0, R)$ that performs a depth-first left-to-right tree traversal 
in a special way. Rather than performing this traversal in one branch, it does so 
in all its branches together, each branch performing a separate piece of the traversal. 
When $M$ arrives from above at a node $u$ with label $\delta\in\Delta$, it outputs $\delta$ 
and splits into two branches. The first branch traverses the subtree at $u$, 
and the second branch continues the traversal after that subtree. Each branch outputs $e$ 
when arriving from below at a \mbox{$\Delta$-labeled} node (or at the root, at the end of the traversal).
Formally, $M$ has the state set $Q=\{d,u_1,u_2\}$ with initial state $q_0=d$, 
cf. Examples~\ref{exa:query} and~\ref{exa:exp}.
It has the following (general) rules, where $j'\in[0,\m_\Sigma]$, $j\in[1,\m_\Sigma]$, 
and $\delta\in\Delta$: 
\[
\begin{array}{llllll}
\tup{d,\at,j'} & \to & \tup{d,\down_1} & 
\tup{d,e,j} & \to & \tup{u_j,\up} \\[0.4mm]
\tup{d,\delta,j} & \to & \delta(\tup{d,\down_1},\tup{u_j,\up}) \quad & 
\tup{d,e,0} & \to & e \\[0.4mm]
\tup{d,\delta,0} & \to & \delta(\tup{d,\down_1},e) & & & \\[2mm]
\tup{u_1,\at,j'} & \to & \tup{d,\down_2} & 
\tup{u_1,\delta,j'} & \to & e \\[0.4mm]
\tup{u_2,\at,j} & \to & \tup{u_j,\up} & & &  \\[0.4mm]
\tup{u_2,\at,0} & \to & e &  &  & 
\end{array}
\]
Thus, since $\tau_M=\fl_\Delta\circ\enc_\Delta$, it follows that 
$\fl_\Delta = \tau_M\circ\dec_\Delta \in \TTLSU\circ\Dec$.

We note that the mapping $\fl_\Delta\circ\enc_\Delta$ is denoted `eval' in~\cite[Section~4]{PerSei},
`APP' in~\cite{ManBerPerSei}, and `app' in~\cite[Section~7]{Eng09}. 
For the reader familiar with \abb{mso} translations we observe that it is also easy 
to show that both $\fl_\Delta$ and $\enc_\Delta$ are deterministic \abb{mso} translations,
and hence their composition is one. 
The second inclusion then follows from Proposition~\ref{pro:mso=ttsu}. 
\qed
\end{proof}

It follows from the first inclusion of Lemma~\ref{lem:decflat} that $\FT\subseteq\FT_\at$.
In fact, $\Enc\circ\TT\circ \Dec \subseteq \Enc\circ\TT \circ \TTDL \circ \Flat$, 
which is included in $\Enc\circ\TT\circ \Flat$ by Theorem~\ref{thm:rightcomp}.
It follows from the second inclusion that $\FT_\at^k \subseteq \FT^{k+1}$ 
for every $k\geq 1$.
In fact, $\FT_\at^k =(\Enc\circ\TT\circ \Flat)^k \subseteq 
(\Enc\circ\TT\circ \TTLSU\circ\Dec)^k \subseteq \Enc\circ(\TT\circ \TTLSU)^k\circ\Dec$,
which is included in $\Enc\circ\TT^k\circ\TTLSU\circ\Dec$ by Theorem~\ref{thm:ttsuott}
and hence in $\Enc\circ\TT^{k+1}\circ\Dec$, which equals $\FT^{k+1}$ by Lemma~\ref{lem:fttt}.

\begin{corollary}
$\FT^k\subseteq\FT_\at^k \subseteq \FT^{k+1}$ for every $k\geq 1$. 
\end{corollary}

From the second inclusion we obtain that our main results also hold for 
deterministic extended \ttt\ forest transducers.
It is decidable whether or not a composition of such transducers is of linear size increase, and 
\[
\FT_\at^k \cap\LSI = \Enc\circ\TTSU\circ\Dec \subseteq \FT\subseteq\FT_\at.
\]
The complexity results of Theorems~\ref{thm:lintime}, \ref{thm:linspace}, and~\ref{thm:outspace}
also hold for $\FT_\at^k$. 

The class of deterministic macro forest translations of~\cite{PerSei} can be defined as 
$\MFT_\at=\Enc\circ\MT\circ\Flat$. Since $\TT\subseteq \MT\subseteq \TT^2$ by 
Lemma~\ref{lem:ttvsmt}, we conclude by similar arguments as for $\FT_\at$ 
that $\MFT_\at^k\subseteq \FT^{2k+1}$ and hence our main results also hold for 
deterministic macro forest transducers. It is decidable whether or not 
a composition of such transducers is of linear size increase, and 
\[
\MFT_\at^k \cap\LSI = \Enc\circ\TTSU\circ\Dec \subseteq \FT\subseteq\FT_\at\subseteq\MFT_\at.
\]
The complexity results of Theorems~\ref{thm:lintime}, \ref{thm:linspace}, and~\ref{thm:outspace}
also hold for $\MFT_\at^k$.

The main results of Sections~\ref{sec:lsi} and~\ref{sec:nondetcomplex}
also hold for nondeterministic forest transducers.
Instead of Lemma~\ref{lem:fttt} we use the obvious fact that 
$\NTT\circ\Dec\circ\Enc\circ\NTT\subseteq \NTT\circ\NTT$.\footnote{It can be shown that
the nondeterministic version of Lemma~\ref{lem:fttt} also holds, but we will not do that here.
}
This implies, together with Lemma~\ref{lem:decflat}, that it suffices to prove that the results 
for $\NTT^k$ also hold for the class $\Enc\circ \NTT^k\circ\Dec$. 
For the nondeterministic version of Theorem~\ref{thm:mainlsi} in Section~\ref{sec:lsi}, we note that
a translation $\enc_\Sigma\circ \tau \circ \dec_\Delta$ is a function 
if and only if $\tau$ is a function. Consequently, 
$(\Enc\circ \NTT^k\circ\Dec) \cap\cF \subseteq \Enc\circ (\NTT^k \cap\cF)\circ\Dec \subseteq
\Enc\circ\TT^{k+1}\circ\Dec=\FT^{k+1}$ by Theorem~\ref{thm:fun} and Lemma~\ref{lem:fttt}.
Hence $(\Enc\circ \NTT^k\circ\Dec) \cap \LSI = \Enc\circ\TTSU\circ\Dec$.
Obviously, the complexity results of Theorems~\ref{thm:nondetoutspace} and~\ref{thm:nondetlinspace} 
in Section~\ref{sec:nondetcomplex} hold for $\Enc\circ \NTT^k\circ\Dec$, 
with the same proof as in the deterministic case. 
The class of nondeterministic macro forest translations of~\cite{PerSei} can be defined as 
$\NMFT_\at=\Enc\circ\NMT\circ\Flat$. From Lemmas~\ref{lem:mtintt} and~\ref{lem:decflat}
we obtain that $\NMFT^k_\at\subseteq \Enc\circ\NTT^{3k}\circ\Dec$, 
and hence all these results also hold for macro forest transducers. 

We finally show that the results of Section~\ref{sec:transcomp} also hold for 
nondeterministic forest transducers. We first consider 
$\Enc\circ \NTT\circ\TT\circ\Dec$ and $\Enc\circ \NTT\circ\TT\circ\Flat$. 
For a forest translation $\tau$ we define the forest language $L_\tau =\{\#[fg]\mid (f,g)\in\tau\}$.
If $\tau =\enc_\Sigma\circ \tau'\circ \dec_\Delta$ with $\tau'\in\NTT\circ\TT$, then 
$\#[fg]\in L_\tau$ if and only if $\#(\enc_\Sigma(f),\enc_\Delta(g))\in L_{\tau'}$. 
Since $\enc_\Sigma(f)$ can be computed by a deterministic finite-state transducer 
(see footnote~\ref{foo:enc}), and similarly for $\enc_\Delta(g)$, 
$L_\tau$ is log-space reducible to $L_{\tau'}$. 
Hence $\Enc\circ \NTT\circ\TT\circ\Dec\subseteq\LOGCF$ by Theorem~\ref{thm:logcf}.
Similarly if $\tau'\in\TT$, then $g\in \tau(L)$ if and only if $\enc_\Delta(g) \in \tau'(\enc_\Sigma(L))$
for every $L\in\REGF$, and hence $\FT(\REGF)\subseteq \LOGCF$ by Corollary~\ref{cor:logcf}.
To show the same results for $\Flat$ instead of $\Dec$, we need the following small lemma.

\begin{lemma}\label{lem:flatyield}
$\Flat \subseteq \TTD\circ\yield$. 
\end{lemma}

\begin{proof}
For an alphabet $\Delta$, let $\Omega$ be the ranked alphabet 
$\Delta\cup\{[\,,]\}\cup\{\lambda,\at,\omega\}$ such that
$\Omega^{(0)}=\Delta\cup\{[\,,],\lambda\}$, $\Omega^{(2)}=\{\at\}$, and $\Omega^{(4)}=\{\omega\}$.
We define the deterministic \ttdl\ $N=(\Delta_\at,\Omega,\{p\},p,R)$ 
with the following (general) rules.
\[
\begin{array}{lll}
\tup{p,j,\at} & \to & \at(\tup{p,\down_1},\tup{p,\down_2}) \\[0.4mm]
\tup{p,j,e} & \to & \lambda \\[0.4mm]
\tup{p,j,\delta} & \to & \omega(\delta,[\,,\tup{p,\down_1},]\,)
\end{array}
\]
for every $\delta\in\Delta$. 
Assuming that the symbol $\lambda$ is skipped when taking yields
(cf. the sentence before Corollary~\ref{cor:ylogcf}),
it should be clear that $\fl_\Delta(t)$ is the yield of $\tau_N(t)$
for every $t\in T_{\Delta_\at}$. 
\qed
\end{proof}

It follows from Lemma~\ref{lem:flatyield} and Theorem~\ref{thm:rightcomp} that 
$\Enc\circ \NTT\circ\TT\circ\Flat\subseteq\Enc\circ \NTT\circ\TT\circ\yield$
and $\FT_\at=\Enc\circ \TT\circ\Flat\subseteq\Enc\circ \TT\circ\yield$.
If $\tau$ is a forest translation such that $\tau=\enc_\Sigma\circ \tau'$ with 
$\tau'\in \NTT\circ\TT\circ\yield$, then 
$\#[fg]\in L_\tau$ if and only if $\#(\enc_\Sigma(f),g)\in L_{\tau'}$.
Hence $\Enc\circ \NTT\circ\TT\circ\Flat\subseteq\LOGCF$ by Corollary~\ref{cor:ylogcf}.
Similarly if $\tau'\in\TT\circ\yield$, then $\tau(L)=\tau'(\enc_\Sigma(L))$ for every $L\in\REGF$, 
and so $\FT_\at(\REGF)\subseteq \LOGCF$ by Corollary~\ref{cor:ylogcf}. 
If we define the class of \abb{io} macro forest translations to be 
$\Enc\circ\NMT_\text{\abb{io}}\circ\Flat$, then that class is included in 
$\Enc\circ \NTT\circ\TT\circ\Flat$ by Lemma~\ref{lem:nmtio} and hence in $\LOGCF$ by the above.
Thus, Corollary~\ref{cor:mtiologcf} also holds for macro forest transducers.

For a forest translation $\tau =\enc_\Sigma\circ \tau'\circ \dec_\Delta$ with $\tau'\in\NTT^k$
it is easy to prove that $L_\tau \in \family{DSPACE}(n)$ and that 
$\tau(L)\in \family{DSPACE}(n)$ for every $L\in\REGF$, 
as we did above for $\tau'\in\NTT\circ\TT$ and $\tau'\in\TT$, respectively,
thus generalizing Theorems~\ref{thm:transcompdspace} and~\ref{thm:outdspace}. 
That also holds for $\fl_\Delta$ instead of $\dec_\Delta$, because 
$\Enc\circ \NTT^k\circ\Flat \subseteq \Enc\circ \NTT^{k+1}\circ\Dec$ by Lemma~\ref{lem:decflat}.

The $\family{NP}$-completeness results of Section~\ref{sec:transcomp} also hold 
for extended forest translations. 
The translation $\tau$ of Theorem~\ref{thm:npcomplete} can be changed into a translation in 
$\Enc\circ\TTD\circ \FTT\circ\Flat$ as follows. First, change $M_1$ in the proof of 
Theorem~\ref{thm:npcomplete} such that it
obtains as input the encodings of the strings $ab^ncd^m$ (viewed as forests). 
Second, change $M_2$ such that it outputs trees over $\Delta_\at$ rather than~$\Delta$
(by changing the rule $\tup{q_{i\vee j},d}\to \vee(\tup{q_i,\alpha},\tup{q_j,\alpha})$ of $M$ 
in the proof of Lemma~\ref{lem:leeuw} into the general rule
$\tup{q_{i\vee j},d}\to \vee(\at(\tup{q_i,\alpha},\tup{q_j,\alpha}))$, and similarly for $\wedge$). 
As a result $\tau$~outputs boolean expressions as forests rather than ranked trees. 
Thus we obtain an $\family{NP}$-complete extended forest translation in $\Enc\circ\TTD\circ \FTT\circ\Flat$,
and hence one in $\NMFT_\at$. In a similar way we also obtain an $\family{NP}$-complete
forest language in $\NFT_\at(\REGF)$. The details are left to the reader. 
It is not clear whether these results hold for $\Dec$ instead of $\Flat$.

\section{Conclusion}\label{sec:conc}

Our main technical result transforms a composition of $k$ \ttt's into a linear-bounded composition 
of $k$ \ttt's, cf. Corollary~\ref{cor:prod}. As observed in Remark~\ref{rem:size}, our proof of this result
can involve a $2(k-1)$-fold blow-up of the sizes of the transducers, which also influences the 
constants of their time and space complexities, cf. the sentence after Theorem~\ref{thm:lintime}.
We do not know whether this transformation can be realized in a more efficient way.  

Our main result on expressivity is that $\TT^k \cap\LSI \subseteq \TT$ for every $k\geq 1$, 
i.e., that every composition of \dttt's that is of linear size increase 
can be realized by one \dttt. Moreover, it is decidable 
whether or not such a composition is of linear size increase. 
Do similar results hold for polynomial size increase? 
For instance, does there exist $m\geq 1$ such that 
every translation in $\bigcup_{k\geq 1}\TT^k$ of quadratic size increase is in $\TT^m$?
The same question can be asked for $\ell$-fold exponential size increase, 
for each fixed $\ell\in\nat$. 

We have shown in Section~\ref{sec:func} that even 
$\NTT^k \cap\LSI\subseteq \TT$ for every $k\geq 1$,
generalizing Theorem~\ref{thm:mainlsi}. Although this result is effective, 
we do not know whether Theorem~\ref{thm:mainlsidec} can also be generalized, 
i.e., whether it is decidable for a nondeterministic \ttt\ $M$ whether or not 
$\tau_M$ is a function of linear size increase. This would be solved if it was decidable 
whether or not $\tau_M$ is a function. But that is also unknown, whereas it has been proved
for classical top-down tree transducers (with regular look-ahead) in~\cite[Theorem~8]{Esik}.
Note that deciding functionality of $\tau_M$ also solves the
equivalence problem for \dttt's, which is already a long standing open problem (cf.~\cite{Eng80,Man15});
in fact, $\tau_1,\tau_2\in\TT$ are the same if and only if they have the same domain 
and $\tau_1\cup\tau_2$ is functional. 
 
Another open question for nondeterministic \ttt's is whether or not 
there exists $m\geq 1$ such that
the inclusion $\NTT^k \cap\LSIR\subseteq \NTT^m$ holds for every $k\geq 1$,
where $\LSIR$ consists of all relations $\tau\subseteq T_\Sigma\times T_\Delta$
of linear size increase, which means that there is a constant $c\in\nat$ such that 
$|s|\leq c\cdot|t|$ for every $(t,s)\in \tau$.
It follows from (the proof of)~\cite[Theorem~3.21]{Ina} (see also~\cite{InaHos,InaHosMan}) 
that $\NTT^2 \cap\LSIR$ is not included in~$\NMT$, 
and hence not in $\NTT$ by the remark following Lemma~\ref{lem:mtintt}.

Similar questions can be asked for macro tree transducers, i.e., for the classes $\MT$ and $\NMT$. 

We have shown in Lemma~\ref{lem:d=ds} that $\TTD = \TTDS$, 
but we do not know whether or not $\TT = \TTS$.
In other words, we do not know whether for every~\ttt\ there is an equivalent sub-testing \ttt, 
in which the regular test of a rule only inspects the subtree of the current node.
Or even more informally, can regular look-around be simulated by regular look-ahead? 

We have shown in Corollary~\ref{cor:oik}
that the string languages in the \mbox{\abb{OI}-hierarchy}, which are generated by high-level grammars,
are in $\family{NSPACE}(n)\wedge\NP$, 
and in Corollary~\ref{cor:oikdet} that they are in $\family{DSPACE}(n)$.
However, the languages 
of the \abb{OI}-hierarchy are generated by so-called ``safe'' high-level grammars,
and it is not known whether the same results hold for unsafe high-level grammars.
It is proved in~\cite{KobInaTsu} that the languages generated by unsafe level-2 grammars,
the unsafe version of $\family{OI}(2)$, are in $\family{NSPACE}(n)$. 

In Section~\ref{sec:transcomp} we have shown that $\TT^k\subseteq \PTIME$,
that $\NTT\circ \TT \subseteq \LOGCF \subseteq \PTIME$, and that $\TT\circ \NTT$ 
contains an $\family{NP}$-complete translation. 
It remains to find out for $k\geq 2$ whether $\NTT\circ \TT^k \subseteq \PTIME$ or whether
it contains an $\family{NP}$-complete translation.

\medskip
{\bf Acknowledgements.}
We are grateful to the reviewers for their constructive comments. 


\end{document}